%% file: string_diagrams_process_mining.tex
\documentclass[sn-basic]{sn-jnl-arxiv}

\input{styles/packages}
\input{styles/tikz_setup}

\input{styles/string_styles}

\input{styles/petri_styles}

\input{styles/math_macros}
\input{styles/theorems}

\input{styles/document_macros}
\input{styles/colors}

\renewenvironment{strip}{}{}
\newcommand{\twocolbreak}{}
\newcommand{\twocolalignbreak}{}

\begin{document}

\title[String Diagrams for Process Mining]{String Diagrams for Process Mining}

\author*[1]{\fnm{Antony R.} \sur{Lee}}\email{arl290@student.bham.ac.uk}
\author[1]{\fnm{Peter} \sur{Ti\v{n}o}}
\author[2]{\fnm{Iain B.} \sur{Styles}}

\affil[1]{\orgdiv{School of Computer Science}, \orgname{University of Birmingham},
  \orgaddress{\street{Edgbaston}, \city{Birmingham}, \postcode{B15 2TT}, \country{UK}}}
\affil[2]{\orgdiv{School of Electronics, Electrical Engineering, and Computer Science},
  \orgname{Queen's University Belfast},
  \orgaddress{\street{University Road}, \city{Belfast}, \postcode{BT7 1NN}, \country{UK}}}

\abstract{Run two process-discovery algorithms on the same event log and they return two different
  pictures of the same process. Petri nets, causal nets, process trees and BPMN each rely on routing
  machinery of their own, and their only common ground is the traces they generate. A trace lists
  activities one after another, so it discards the concurrency the notations exist to express, and
  two models with identical trace languages can describe genuinely different processes. Whether two
  discovered models mean the same thing therefore has no notation-independent answer. We show that
  all four notations admit one canonical presentation, a signature recording a model's activities
  and the typed interfaces along which they compose, and nothing of the routing machinery. Four
  construction theorems establish this presentation notation by notation, so two models are compared
  in a form that each notation determines on its own. Signature equality implies trace equivalence
  and is strictly finer, separating genuine concurrency from interleaved choice, which a trace
  comparison cannot, and signature inclusion implies trace inclusion. In a recovery study on one
  object-centric log, the discovered causal net has exactly the ground-truth signature, while the
  discovered Petri net's signature strictly contains it and locates every behaviour the Petri net
  adds in its silent structure. The
  gain is largest for object-centric data, where a single log supports several notations at once.
  The four notations become one calculus, in which a translation between them is a claim that can
  be checked.}

\keywords{process mining, hypergraph categories, string diagrams, decorated cospans,
  object-centric processes, compositional semantics}

\maketitle

\input{sections/01-introduction}
\input{sections/02-categories}
\input{sections/02b-general-framework}

\section{The Construction, Notation by Notation}\label{sec:notation-by-notation}

Each of the four notations below is one instance of the construction of
Section~\ref{sec:general-framework}, and nothing in that construction is redone here. An
instance is fixed by two choices, the object-type set $\mathcal{O}$ and the mediator kinds
$\nu$. The defaults are the same throughout.
$\mathcal{O}$ is either $\{\ast\}$ (untyped) or the object types, and every edge carries its
object type. Each instance table below therefore varies in
only three rows, the AND/OR graph $(V,E)$, which nodes are the activities (always the
notation's observable nodes), and the mediators with their AND or XOR tag.

The object-centric causal net is the reference instance, the most developed of the four,
where the mediators appear in their most collapsed form and the multiplicity decoration
does its fullest work. The other three add one wrinkle each, a silent transition read
as a transparent mediator (Petri nets), an operator tree compiled to an AND/OR graph
(process trees), and inclusive-choice gateways expanded into AND and XOR (BPMN). Petri
nets are presented first all the same, being the notation most readers arrive from.

Two pieces of shared structure are stated here once rather than four times below. First,
every instance proof follows one schema. The instance map is checked against
Definition~\ref{def:lm-graph}, finiteness of the model supplies that definition's
finiteness conditions, finiteness of the graph bounds the walk of
Definition~\ref{def:contexts-general}, and Definitions~\ref{def:generator-cospan-general}
and~\ref{def:free-category} then yield the signature and $\mathbf{F}(\Sigma)$, to which
Theorem~\ref{thm:canonical-presentation} applies. Each subsection's proof records only its delta against this schema.
Second, the classical untyped case is uniform across instances.

\begin{corollary}[Untyped models as a special case]\label{cor:untyped-general}
  For every instance of the construction of Section~\ref{sec:general-framework}, setting
  $\mathcal{O}=\{\ast\}$, so that every edge carries the single trivial type,
  specialises the construction to the classical untyped formalism. All boundary ports
  carry the trivial type, and boundary-matching composability reduces to the untyped structural condition.
\end{corollary}
\input{sections/03-petri_nets}

\input{sections/04-causal_nets}

\input{sections/05-process_trees}
\input{sections/06-bpmn}

\input{sections/07-validation}
\input{sections/09-conversion_framework}
\input{sections/08-related_work}
\input{sections/09b-conclusion}

\backmatter

\bmhead{Funding}
ARL acknowledges the receipt of studentship awards from the Health Data Research UK-The Alan Turing Institute Wellcome PhD Programme in Health Data Science (Grant Ref: 218529/Z/19/Z). P.~Ti\v{n}o was supported by the EPSRC Prosperity Partnerships grant ARCANE, EP/X025454/1.

\bmhead{Author contributions}
Antony R. Lee: conceptualisation, methodology, software, formal analysis, investigation, validation, visualisation, writing -- original draft, writing -- review and editing. Peter Ti\v{n}o: supervision, writing -- review and editing. Iain B. Styles: supervision, writing -- review and editing.

\begin{appendices}
\input{sections/10-hyper_axioms}
\input{sections/11-process-tree-extras}
\end{appendices}

\bibliography{references}

\end{document}

%% file: styles/packages.tex
\usepackage{amsmath,amssymb,amsfonts,amsthm}
\usepackage{mathrsfs}
\usepackage{MnSymbol}
\let\mnTextCheckmark\checkmark
\renewcommand{\checkmark}{\ensuremath{\mnTextCheckmark}}
\let\cite\citep

\usepackage{xurl}     %
\usepackage{multirow}
\usepackage{placeins} %
\usepackage{tabularx}
\usepackage[title]{appendix}
\usepackage{textcomp}
\usepackage{manyfoot}

\usepackage{xcolor}
\usepackage{soul}      %

\definecolor{placeholder}{HTML}{0B4FA8}

\usepackage{graphicx}
\usepackage{subcaption}

\usepackage{tikz}
\usepackage{forest}

\usepackage{cuted}

%% file: styles/tikz_setup.tex
\usetikzlibrary{
  graphs,
  shapes.geometric,
  arrows.meta,
  bending,
  positioning,
  knots,
  intersections,
  calc,
  math,
  angles,
  quotes,
  petri,
  patterns,
  braids,
  decorations.pathreplacing
}

%% file: styles/string_styles.tex
\tikzset{
  cospanblue/.style={
    draw=boxblueborder,
    minimum width=0.8cm,
    minimum height=0.6cm,
    thick,
    dashed,
    fill=boxblue,
    font=\small\sffamily,
    outer sep=0pt
  },
  morphismblue/.style={
    draw=boxblueborder,
    minimum width=0.8cm,
    minimum height=0.6cm,
    thick,
    fill=boxblue,
    font=\small\sffamily,
    outer sep=0pt
  },
  cospangrey/.style={
    draw=boxgreyborder,
    minimum width=0.8cm,
    minimum height=0.6cm,
    thick,
    dashed,
    fill=boxgrey,
    font=\small\sffamily,
    outer sep=0pt
  },
  lc/.style 2 args={
    coordinate,
    label={[#1, font=\upshape\footnotesize]above:#2}
  },
  wire/.style={},
  gap/.style={line width=3pt, white},
  spider/.style={fill=black, circle, inner sep=1.5pt},
  port/.style={fill=black, circle, inner sep=0.75pt},
  inline/.style={baseline={([yshift=-.5ex]current bounding box.center)}}
}

%% file: styles/petri_styles.tex
\tikzset{%
  petriBase/.style={>=Stealth, thick},
  typedflow/.style={thick, -Stealth, rounded corners=2pt},
  typedwire/.style={-},
  typeRed/.style={},
  typeBlue/.style={},
  typeGreen/.style={},
  redflow/.style={typedflow, typeRed, draw=redcol},
  blueflow/.style={typedflow, typeBlue, draw=bluecol},
  greenflow/.style={typedflow, typeGreen, draw=greencol},
  blackflow/.style={typedflow, draw=black!80},
  redwire/.style={typedwire, typeRed, draw=redcol},
  bluewire/.style={typedwire, typeBlue, draw=bluecol},
  greenwire/.style={typedwire, typeGreen, draw=greencol},
  place/.style={circle, draw=black!80, minimum size=8mm, font=\small},
  blackP/.style={place, fill=black!10, draw=black!80},
  redP/.style={place, fill=redcol!10, draw=redcol},
  blueP/.style={place, fill=bluecol!10, draw=bluecol},
  greenP/.style={place, fill=greencol!10, draw=greencol},
  transition/.style={rectangle, draw=black!80, fill=nullcol, minimum size=8mm},
  T/.style={transition}
}

%% file: styles/math_macros.tex
\newcommand{\xmapsto}[2][]{%
  \mathrel{\shortmid\mkern-7mu\xrightarrow[#1]{#2}}%
}

%% file: styles/theorems.tex
\theoremstyle{plain}
\newtheorem{theorem}{Theorem}
\newtheorem{lemma}{Lemma}[section]
\newtheorem{proposition}{Proposition}[section]
\newtheorem{corollary}{Corollary}[section]
\theoremstyle{definition}
\newtheorem{definition}{Definition}
\theoremstyle{remark}
\newtheorem{exmp}{Example}[section]
\newtheorem{remark}{Remark}

%% file: styles/colors.tex
\definecolor{boxblue}{RGB}{238,238,255}
\definecolor{boxblueborder}{RGB}{170,170,255}
\definecolor{boxgrey}{RGB}{238,238,238}
\definecolor{boxgreyborder}{RGB}{191,191,191}

\definecolor{redcol}{RGB}{200,30,30}
\definecolor{bluecol}{RGB}{30,30,200}
\definecolor{greencol}{RGB}{30,150,30}
\definecolor{nullcol}{RGB}{240,240,240}

\definecolor{typepat}{HTML}{2A78D6}
\definecolor{typelab}{HTML}{008300}
\definecolor{typeimg}{HTML}{EB6834}

%% file: sections/01-introduction.tex
\section{Introduction}\label{sec:introduction}

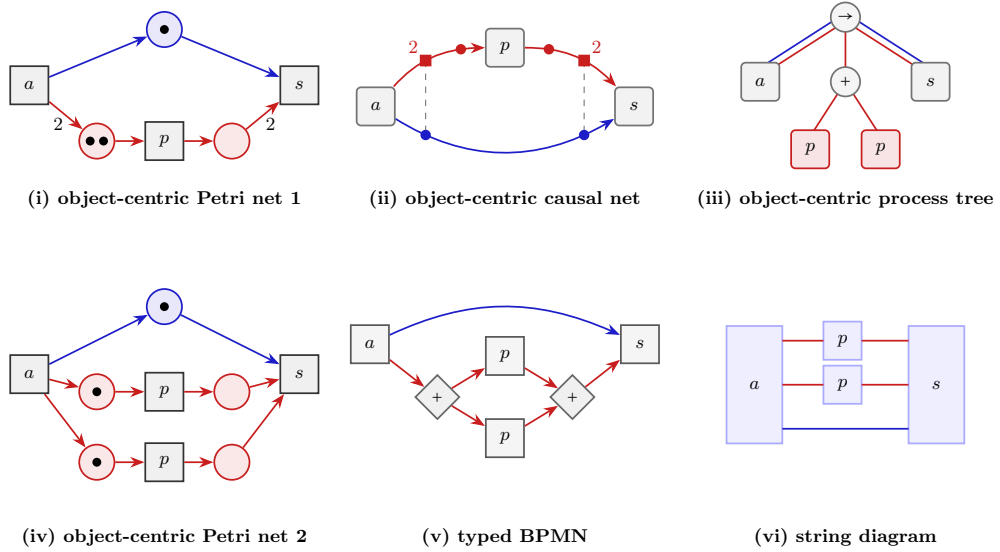
\begin{figure*}[!t]
  \centering
  \resizebox{\linewidth}{!}{\input{tikz/categories/motive_three_faces}\unskip}
  \caption{One process in four notations and its string diagram. Accept an order with two
    items~($a$), pack each item~($p$), and ship once both are packed~($s$), over the object types
    order (blue) and item (red). Panels~(i) and~(iv) are two object-centric Petri nets, whose
    circles are places holding resources. The object-centric causal net of panel~(ii) has no places
    and carries the same resources as markers on its arcs. Circle nodes in the process tree of
    panel~(iii) encode structure rather than resources, and the diamond gateways of the typed BPMN
    model in panel~(v) do the same. Panel~(vi) is the string diagram, which encodes both minimally,
    structure as topology and resources as wires. Each panel is equivalent to panel~(vi),
    reconstructed from the local connectivity of its activities, and all six emit the language
    $\{\langle a,p,p,s\rangle\}$.}
  \label{fig:oc-two-faces}
\end{figure*}

Organisations audit, redesign, and argue over a picture of the process that a process-mining tool
drew from an event log. A second tool given the same log draws a different picture. The pictures
come in several incompatible kinds, and a
modeller given two of them cannot decide whether they describe the same process without first
choosing a translation from one kind into the other. The comparisons that survive such a translation
run over sequences, either the activity sequences a model allows or the states a model passes
through as it produces them, and both readings replace concurrent activities by the orders in which
they might have occurred. Here we give the four notations used in practice a single common form in
which the concurrency survives, and a test that decides whether two models, drawn in the same
notation or in two, have the same structure. The test computes for each model a signature, a
finite description of which activities hand which objects to which activities, and compares the
two signatures for equality.

No comparison in current use decides whether two models describe the same process. Comparing
trace languages conflates two activities that run concurrently with an exclusive choice between
their two orders, because both allow the same two sequences. The transition systems behind those
sequences make the same identification, one interleaved state graph for the concurrent pair and
one for the choice, and the two graphs are
bisimilar~\cite{sassoneModelsConcurrencyClassification1996}, so bisimulation inherits the
conflation. The equivalences that do see the difference, among them history-preserving
bisimulation~\cite{vanglabbeekRefinementActionsEquivalence2001}, compare runs as partial orders,
but each is defined inside a single notation's unfolding semantics and settled by a matching game
over its state space, so applying one across notations requires the very translation whose
correctness is in question. Replay-based conformance checking measures a model against a log
rather than against a second model. None of the four routes settles the question.

The reason is shared by all four. Every notation records a process through routing devices of its
own, the constructs that steer objects between activities. We take the four that
process-discovery tools return, and between them they exhaust the devices in use, namely a buffer drawn
as a node, a buffer folded into an edge, an operator node, and a gateway.

\emph{Object-centric Petri nets}~\cite{vanderaalstDiscoveringObjectCentricPetri2020} hold objects
in typed places and count them on arc weights. We assume a net is type-consistent, so that no chain
of silent transitions joins places of different object type
(Definition~\ref{def:type-consistent}), which nets produced by object-centric discovery satisfy
automatically. \emph{Object-centric causal nets}~\cite{lissObjectCentricCausalNets2025} have no
places at all and record the same fact in binding markers on their arcs. They are the reference
instance of the construction of Section~\ref{sec:general-framework}, the one in which the routing devices appear in their most
collapsed form. \emph{Object-centric process
trees}~\cite{leemansDiscoveringBlockStructuredProcess2013} route through operator nodes, and are
block-structured and hence series-parallel at the level of
occurrences~(Theorem~\ref{thm:pt-as-sp-language}), the assumption that guarantees their soundness and
excludes the non-series-parallel behaviours. \emph{Typed models in the Business Process Model and
Notation} (BPMN)~\cite{BusinessProcessModel} route through gateways. We take a model to be finite
and expand its OR gateways into combinations of AND and XOR, and require nothing else of it.

No routing device is
determined by the behaviour it encodes, so two correct models of one process may use different
devices, and comparing drawings then compares devices. Section~\ref{sec:related-work} takes up
this literature in detail.

Figure~\ref{fig:oc-two-faces} shows the routing devices at their smallest, on the process its
caption describes, drawn in panels (i)--(v) in the four notations a process-mining tool is most
likely to hand back. Write $a$ for accept order, $p$ for pack item,
and $s$ for ship order, over the two object types $\mathrm{order}$ (blue) and
$\mathrm{item}$ (red).

The five drawings agree on the activities and on little else. The object-centric Petri net of
Figure~\ref{fig:oc-two-faces}(i) records the two items in a typed place carrying the arc weight
$2$. Figure~\ref{fig:oc-two-faces}(iv) draws the same process with one transition per item and
every arc weight $1$, so the device of panel~(i) is a property of the drawing alone.
The object-centric causal net of Figure~\ref{fig:oc-two-faces}(ii) has no places anywhere and
records the same fact on its arcs, so the buffer one notation draws as a node the other folds into
an edge. No picture is a special case of another. The process tree of
Figure~\ref{fig:oc-two-faces}(iii) routes on the arity of a parallel operator, and the BPMN model
of Figure~\ref{fig:oc-two-faces}(v) on a pair of gateways, devices of a third and a fourth kind.

Figure~\ref{fig:oc-two-faces}(vi) is the same process with the routing devices removed, and it is
the picture this paper builds on. A box is an activity occurrence and a wire is a typed object
passing from the activity that produces it to the one that consumes it
(Table~\ref{tab:intuition}). Here $a$ sends one $\mathrm{order}$ wire and two $\mathrm{item}$ wires outward, each
$\mathrm{item}$ wire passes through its own $p$ box, and all three arrive at $s$. Nothing else is
drawn, because the picture records only which activity hands which object to which other activity.

A picture of this kind is a string diagram. It is not a fifth notation to model in.
Nothing in it is left to choose, since every routing device is gone and only the typed hand-overs
remain. The string diagram is the normal
form the four notations share, the way regular expressions are compared through their minimal
automata, and it comes with an algebra that says how two diagrams combine.

The reading of a run as a partially ordered set of occurrences is itself old. Partial orders have
modelled concurrency since the sixties~\cite{prattModelingConcurrencyPartial1986}, net theory
reads each run of a Petri net as such an order, and the equivalences above compare those orders
configuration by configuration. The missing step is from semantics to syntax. In net theory the
partial order is something a model means, computed from the model by unfolding it. Here the partial order is
something a modeller writes, the drawing shared by four notations, and equality of drawings is
decided by comparing finite sets rather than by playing a game.

The signature turns this into the promised test. Every model in the four notations reduces to its
signature, the finite record of its activities and their typed inputs and outputs, and the
reduction removes the routing devices and nothing else. Equal signatures certify the same
structure, the same activities with the same typed inputs and outputs, and hence the same trace language,
whichever notations the two models were drawn in, while unequal signatures witness a structural
difference that no drawing convention removes. Included signatures certify included trace
languages in the same way. The comparison fixes no modelling convention, and
it runs on models as discovery tools emit them. Section~\ref{sec:validation} runs it on the output
of two discovery algorithms from one log. The discovered causal net has exactly the ground-truth
signature, and the discovered Petri net's signature strictly contains it, so the comparison names
each behaviour the Petri net adds. The step from structure to behaviour runs one way, since two models of different
structure can still allow the same traces, and the contribution list below states the direction.

Three properties make the string diagram usable as a common form, and
Table~\ref{tab:intuition} records the vocabulary the rest of the paper uses, with a plain reading
for each term.

\begin{description}
  \item[Connectivity, not order.] The wires record which activities exchange which objects, not
    when anything fires. Two activities with no wire between them are concurrent, and no ordering
    is implied or stored.
  \item[Typed resources with multiplicity.] Each wire carries an object type, and two boxes join
    only where their types agree. An order and a pair of items are therefore different data rather
    than differently labelled data, and object-centric coherence follows from the joining rule with
    no further bookkeeping.
  \item[Deformation invariance.] Only the connections matter. The picture may be stretched, and its
    boxes slid past one another, without changing the process it denotes.
\end{description}

Two rules assemble diagrams from smaller ones. Sequential composition $(;)$ joins two diagrams end
to end by matching outgoing wires to incoming wires of the same type. Parallel composition
$(\otimes)$ sets two diagrams side by side with no interaction between them. A box may have several
output wires of one type, which is how $a$ hands its two $\mathrm{item}$ wires to the two packs of
Figure~\ref{fig:oc-two-faces}(vi). Wires may also split and merge, which copies one object to
several activities and joins the copies again.

\begin{table*}[t]
  \centering
  \caption{The vocabulary of the construction, in the order the paper introduces it. Each term is
    given its plain reading here and its formal definition at the section listed.}
  \label{tab:intuition}
  \begin{tabular}{l p{0.52\linewidth} l}
    \hline
    term & plain reading & defined in \\
    \hline
    string diagram & a picture of one process, with a box for each activity and a wire for each
      object passing between activities & Section~\ref{sec:hypergraphs} \\
    wire & one object of one type, running from the activity that produces it to the activity that
      consumes it & Section~\ref{sec:hypergraphs} \\
    cospan & a fragment with a designated input side and output side, so that two fragments
      glue by matching the wires on the sides they meet & Section~\ref{sec:hypergraphs} \\
    pushout & the gluing itself, which identifies the wires two fragments hand to each other and
      keeps everything else & Section~\ref{sec:hypergraphs} \\
    decorated cospan & a cospan carrying data of its own, which the gluing carries along
      & Section~\ref{sec:hypergraphs} \\
    object type & the kind of thing a wire carries, so that two boxes join only where their types agree
      & Section~\ref{sec:hypergraphs} \\
    generator $g_{a,c}$ & one box, meaning an activity together with the types of the objects it
      consumes and produces. An activity gets one box for each context, one way of joining what it
      receives to what it emits &
      Section~\ref{sec:general-framework} \\
    signature $\Sigma$ & the set of a model's boxes, unconnected, each distinguished by its label, its typed
      boundary of incoming and outgoing wires and its constraints on their counts & Section~\ref{sec:hypergraphs} \\
    mediator & a notation's routing device, meaning a place, a gateway, an operator node, or a
      binding set & Section~\ref{sec:general-framework} \\
    $(;)$ and $(\otimes)$ & join two diagrams end to end, and set two diagrams side by side &
      Section~\ref{sec:hypergraphs} \\
    hypergraph category & the setting in which the two joins and wire splitting and merging are
      always available, subject to the axioms of Appendix~\ref{sec:axioms} &
      Section~\ref{sec:hypergraphs} \\
    trace semantics $\Gamma$ & the map sending a diagram to the sequences its partial order admits
      & Section~\ref{sec:hypergraphs} \\
    \hline
  \end{tabular}
\end{table*}

Because routing devices are never recorded, models that differ only in their devices become the
same string diagram. Figure~\ref{fig:routing-collapse} shows the effect within a single notation.
Two object-centric Petri nets returned by two discovery algorithms from one event log, agreeing on
their activity labels and differing in their places and silent transitions, both reduce to the one
string diagram of Figure~\ref{fig:routing-collapse}(iii) (Remark~\ref{rem:ab-place-vs-tau}).

The same picture keeps the structure a sequence comparison loses. The parallel composite
$a\otimes b$, in which two activities are both required with no constraint on their order, and the
exclusive choice between the sequential runs $a\,;\,b$ and $b\,;\,a$ admit the same two sequences
$\langle a,b\rangle$ and $\langle b,a\rangle$ and have bisimilar transition systems. As string diagrams the two are plainly different, one diagram
whose causal order is discrete against two diagrams that are each a total order. For a modeller
the two models are not interchangeable, and no sequence-based comparison warns of it. The reading
extends to choice in general, which Section~\ref{sec:general-framework} makes formal. A
model with choice denotes a family of diagrams, one for each resolution of its choices.

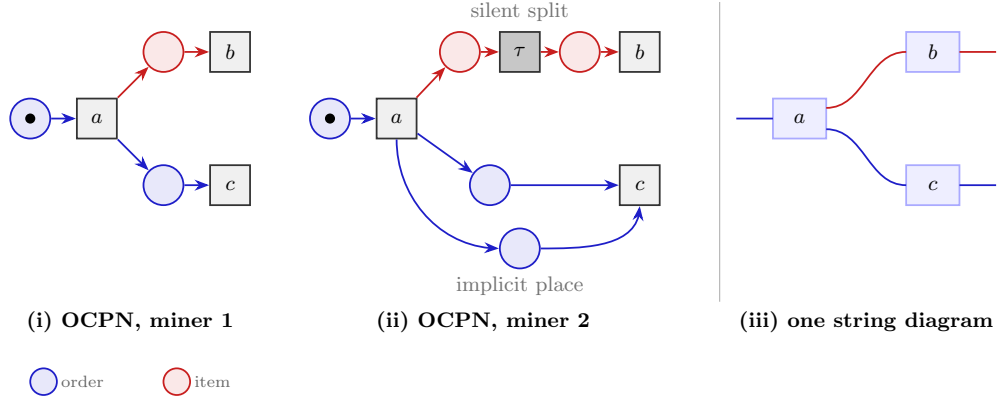
\begin{figure*}[!t]
  \centering
  \resizebox{\linewidth}{!}{\input{tikz/introduction/routing_collapse}\unskip}
  \caption{One behaviour over two object types, receive order~($a$), pick item~($b$), and
    invoice~($c$). Panels (i) and (ii) are object-centric Petri nets as two discovery algorithms
    emit them, differing only in a silent transition splitting the item place and an implicit
    place duplicating the order constraint. Both reduce to the single string diagram (iii).}
  \label{fig:routing-collapse}
\end{figure*}

The rest of the paper builds the mathematics and the machinery this picture needs. We work in a
hypergraph category, the setting of compositional category
theory~\cite{fongHypergraphCategories2018} in which wires may split and merge, and we assign to each
notation a signature $\Sigma$, the collection of boxes its models are built from. Places,
silent transitions, gateways and tree operators are all mediators, and contribute nothing to it. Every notation is then presented in
one vocabulary and joined by one composition algebra, so comparing and translating models across
notations becomes well defined.

The paper makes three lines of contribution, each stated as a modelling deliverable.

\begin{itemize}

  \item \textbf{Four canonical presentations as one category.} Every model in each of the four
    notations reduces to a signature, and one construction~(Section~\ref{sec:general-framework})
    does it in all four cases. Formally,
    object-centric Petri nets~(Theorem~\ref{thm:petri-to-cospan-signature}), object-centric causal
    nets~(Theorem~\ref{thm:causal-to-cospan-signature}), object-centric process
    trees~(Theorem~\ref{thm:ptree-to-cospan-signature}), and typed
    BPMN~(Theorem~\ref{thm:bpmn-cospan}) each receive a canonical presentation as a cospan-algebra
    signature. BPMN needs no structural precondition, and its OR
    gateways are handled by finite expansion into AND and XOR combinations. The presentation is
    notation-independent, so the same behaviour receives the same signature whichever notation
    expresses it (Remark~\ref{rem:ptree-recovers-pn}).

  \item \textbf{An equivalence certificate more detailed than traces.} Two models can be checked for
    agreement by comparing their signatures, with no modelling convention to fix first. Equal
    signatures certify the same structure~(Lemma~\ref{lem:sig-complete}) and hence the same trace
    language~(Theorem~\ref{thm:signature-equivalence}), and the certificate separates the genuine
    concurrency $a\otimes b$ from the exclusive choice that trace comparison and bisimulation both
    conflate with it. Included signatures certify included trace
    languages~(Corollary~\ref{cor:signature-inclusion}).
    Because the maps preserve $(;)$ and $(\otimes)$, the verdict persists when a fragment is
    composed into a larger model. The implication runs in one direction only and we do not claim its
    converse.

  \item \textbf{Object-centric by construction.} Object types live on the wires. Generator-cospan
    boundaries carry object-type labels, and type-matching composability enforces multi-object
    coherence structurally, with no reachability bookkeeping, so a model that mixes types cannot be
    assembled wrongly in the first place. The classical untyped notations are the one-type
    specialisation.

\end{itemize}

As a byproduct we prove that the sequence and parallel fragment of process trees generates the
finite series-parallel posets, and only those, at the level of
occurrences~(Theorem~\ref{thm:pt-as-sp-language}, Corollary~\ref{cor:pt-sp-iff}).

The construction covers notations that are procedural and local, in which a model is assembled from
pieces each wired to its immediate neighbours. Globally-constrained declarative languages such as
Declare~\cite{diciccioGeneratingEventLogs2015} or Dynamic Condition Response (DCR)
graphs~\cite{deboisDeclarativeProcessMining2017}, whose constraints range over whole traces, fall
outside it.

Section~\ref{sec:hypergraphs} develops the
categorical background and the trace semantics $\Gamma$ that connects string diagrams to process
behaviour, and Section~\ref{sec:general-framework} presents the single construction the four
notation sections instantiate. The construction theorems appear in Sections~\ref{sec:petri-nets},
\ref{sec:causal-nets}, \ref{sec:process_trees}, and~\ref{sec:bpmn}, each exercised on a small
object-centric model that reduces to its cospan signature. Section~\ref{sec:validation} reports
the recovery study, in which each of two discovered models is compared against a known
ground-truth signature. The full pipeline, from log to signature to verdict, is implemented in the
\texttt{proc-posets} package released alongside the paper. Section~\ref{sec:conversion-framework} develops cross-notation equivalence and
conversion, Section~\ref{sec:related-work} places the work, and Section~\ref{sec:conclusion}
concludes. Appendix~\ref{sec:axioms} collects the hypergraph-category axioms, and
Appendix~\ref{sec:pt-extras} develops the series-parallel representation theorem for process
trees.

%% file: tikz/categories/motive_three_faces.tex
\begin{tikzpicture}[petriBase, inline, scale=0.95, every node/.style={transform shape},
    place/.append style={minimum size=6mm},
    transition/.append style={minimum size=6.5mm, font=\small},
    plab/.style={font=\footnotesize, black!60, inner sep=1pt},
    ptitle/.style={font=\small\bfseries, inner sep=2pt},
    note/.style={font=\scriptsize, black!55, inner sep=1pt},
    rule/.style={black!30, line width=0.5pt},
    cnact/.style={rectangle, rounded corners=2pt, draw=black!55, fill=black!5,
                  minimum size=6.5mm, font=\small},
    omk/.style={circle, draw=#1, fill=#1, minimum size=4.2pt, inner sep=0pt},
    smk/.style={rectangle, draw=#1, fill=#1, minimum size=4.8pt, inner sep=0pt},
    card/.style={font=\scriptsize, inner sep=0.5pt},
    bind/.style={black!55, line width=0.5pt, dashed},
    op/.style={circle, draw=black!55, fill=black!4, inner sep=1.2pt, minimum size=5mm,
               font=\small},
    lboth/.style={rectangle, rounded corners=2pt, draw=black!55, fill=black!5,
                  minimum size=6.5mm, font=\small},
    litem/.style={rectangle, rounded corners=2pt, draw=redcol, fill=redcol!12,
                  minimum size=6.5mm, font=\small},
    task/.style={rectangle, draw=black!65, fill=black!4, minimum size=6.5mm, font=\small},
    gw/.style={diamond, draw=black!60, fill=black!6, inner sep=1pt, minimum size=7.6mm,
               font=\small},
    wt/.style={font=\scriptsize, inner sep=1pt},
    trace/.style={font=\small, inner sep=1.5pt},
    uwire/.style={typedwire, draw=black!75},
    sbox/.style={rectangle, draw=boxblueborder, thick, fill=boxblue, outer sep=0pt,
                 minimum width=0.44cm, minimum height=0.44cm, font=\scriptsize\sffamily}]

  \def\ycap{-1.95}
  \def\ycapB{-2.60}
  \def\rowB{-5.10}

  \providecommand{\mtfbrief}{0}
  \ifnum\mtfbrief=1
    \def\xSD{7.10}\def\ySD{0}
    \def\ycap{-2.30}
    \def\capSD{\ycap}
    \def\legendY{-3.55}
    \def\numOCPN{}\def\numSD{}\def\sufOCPN{}
  \else
    \def\xSD{12.35}\def\ySD{-5.10}
    \def\capSD{\ycapB}
    \def\numOCPN{(i) }\def\numSD{(vi) }\def\sufOCPN{ 1}
  \fi

  \begin{scope}[shift={(0,0)}]
    \node[T]                (a)   at (0,0)          {$a$};
    \node[blueP, tokens=1]  (po)  at (2.3,0.95)     {};
    \node[redP,  tokens=2]  (pi)  at (1.15,-0.95)   {};
    \node[T]                (p)   at (2.3,-0.95)    {$p$};
    \node[redP]             (pi2) at (3.45,-0.95)   {};
    \node[T]                (s)   at (4.6,0)        {$s$};

    \draw[blueflow] (a)   -- (po);
    \draw[blueflow] (po)  -- (s);
    \draw[redflow]  (a)   -- (pi) node[wt, pos=0.5, below left]  {$2$};
    \draw[redflow]  (pi)  -- (p);
    \draw[redflow]  (p)   -- (pi2);
    \draw[redflow]  (pi2) -- (s)  node[wt, pos=0.5, below right] {$2$};

    \ifnum\mtfbrief=1
      \node[note, anchor=north] at (2.3,-1.45) {arc weight};
    \fi
    \node[ptitle] at (2.3,\ycap) {\numOCPN object-centric Petri net\sufOCPN};
  \end{scope}

  \ifnum\mtfbrief=0
  \begin{scope}[shift={(5.90,-0.36)}]
    \node[cnact] (ca) at (0,0)     {$a$};
    \node[cnact] (cp) at (2.2,1.0) {$p$};
    \node[cnact] (cs) at (4.4,0)   {$s$};

    \draw[redflow,  name path=rap] (ca) to[out=48,in=180] (cp.west);
    \draw[redflow,  name path=rps] (cp.east) to[out=0,in=132] (cs);
    \draw[blueflow, name path=blu] (ca) to[out=-32,in=212] (cs);

    \path[name path=vA]    (0.85,-1.0) -- (0.85,1.7);
    \path[name path=vS]    (3.55,-1.0) -- (3.55,1.7);
    \path[name path=vPin]  (1.45,-1.0) -- (1.45,1.7);
    \path[name path=vPout] (2.95,-1.0) -- (2.95,1.7);
    \path[name intersections={of=rap and vA,    by=mItemA}];
    \path[name intersections={of=blu and vA,    by=mOrdA}];
    \path[name intersections={of=rps and vS,    by=mItemS}];
    \path[name intersections={of=blu and vS,    by=mOrdS}];
    \path[name intersections={of=rap and vPin,  by=pIn}];
    \path[name intersections={of=rps and vPout, by=pOut}];

    \draw[bind] (mItemA) -- (mOrdA);
    \draw[bind] (mItemS) -- (mOrdS);

    \node[smk=redcol, label={[card,redcol]above left:$2$}]  at (mItemA) {};
    \node[smk=redcol, label={[card,redcol]above right:$2$}] at (mItemS) {};
    \node[omk=redcol]  at (pIn)   {};
    \node[omk=redcol]  at (pOut)  {};
    \node[omk=bluecol] at (mOrdA) {};
    \node[omk=bluecol] at (mOrdS) {};

    \node[ptitle] at (2.2,{\ycap+0.36}) {(ii) object-centric causal net};
  \end{scope}
  \fi

  \ifnum\mtfbrief=0
  \begin{scope}[shift={(13.90,-0.39)}]
    \def\dd{0.06}
    \providecommand{\dualleg}[2]{%
      \draw[bluewire] ($(#1)!\dd cm!90:(#2)$)  -- ($(#2)!\dd cm!-90:(#1)$);
      \draw[redwire]  ($(#1)!\dd cm!-90:(#2)$) -- ($(#2)!\dd cm!90:(#1)$);
    }
    \coordinate (tr)   at (0,1.55);
    \coordinate (tpl)  at (0,0.45);
    \coordinate (tAt)  at (-1.37,0.67);
    \coordinate (tSt)  at (1.37,0.67);
    \coordinate (tP1t) at (-0.6,-0.48);
    \coordinate (tP2t) at (0.6,-0.48);
    \coordinate (trA)  at ($(tr)!0.20cm!(tAt)$);
    \coordinate (trS)  at ($(tr)!0.20cm!(tSt)$);

    \dualleg{tAt}{trA}
    \dualleg{trS}{tSt}
    \draw[redwire] (0,1.37) -- (0,0.62);
    \draw[redwire] ($(tpl)!0.20cm!(tP1t)$) -- (tP1t);
    \draw[redwire] ($(tpl)!0.20cm!(tP2t)$) -- (tP2t);

    \node[op]    (root)  at (tr)          {$\to$};
    \node[lboth] (ta)    at (-1.45,0.45)  {$a$};
    \node[op]    (tplus) at (tpl)         {$+$};
    \node[lboth] (ts)    at (1.45,0.45)   {$s$};
    \node[litem] (tp1)   at (-0.6,-0.70)  {$p$};
    \node[litem] (tp2)   at (0.6,-0.70)   {$p$};

    \node[ptitle] at (0,{\ycap+0.39}) {(iii) object-centric process tree};
  \end{scope}
  \fi

  \ifnum\mtfbrief=0
  \begin{scope}[shift={(5.80,\rowB+0.69)}]
    \node[task] (ba)  at (0,0)        {$a$};
    \node[gw]   (bg1) at (1.15,-0.90) {$+$};
    \node[task] (bp1) at (2.30,-0.20) {$p$};
    \node[task] (bp2) at (2.30,-1.60) {$p$};
    \node[gw]   (bg2) at (3.45,-0.90) {$+$};
    \node[task] (bs)  at (4.60,0)     {$s$};

    \draw[blueflow] (ba) to[out=25,in=155] (bs);
    \draw[redflow]  (ba)  -- (bg1);
    \draw[redflow]  (bg1) -- (bp1);
    \draw[redflow]  (bg1) -- (bp2);
    \draw[redflow]  (bp1) -- (bg2);
    \draw[redflow]  (bp2) -- (bg2);
    \draw[redflow]  (bg2) -- (bs);

    \node[ptitle] at (2.30,{\ycapB-0.69}) {(v) typed BPMN};
  \end{scope}

  \fi

  \begin{scope}[shift={(\xSD,\ySD)}]
    \node[morphismblue, minimum height=2cm, minimum width=0.95cm]    (da)  at (0,0)      {$a$};
    \node[morphismblue, minimum height=0.65cm, minimum width=0.65cm] (dp1) at (1.5,0.75) {$p$};
    \node[morphismblue, minimum height=0.65cm, minimum width=0.65cm] (dp2) at (1.5,0.0)  {$p$};
    \node[morphismblue, minimum height=2cm, minimum width=0.95cm]    (ds)  at (3.1,0)    {$s$};

    \draw[redwire] ($(da.south east)!0.875!(da.north east)$) to[out=0,in=180] (dp1.west);
    \draw[redwire] (dp1.east) to[out=0,in=180] ($(ds.south west)!0.875!(ds.north west)$);
    \draw[redwire] ($(da.south east)!0.5!(da.north east)$) to[out=0,in=180] (dp2.west);
    \draw[redwire] (dp2.east) to[out=0,in=180] ($(ds.south west)!0.5!(ds.north west)$);
    \draw[bluewire] ($(da.south east)!0.125!(da.north east)$) -- ($(ds.south west)!0.125!(ds.north west)$);

    \ifnum\mtfbrief=1
      \node[note, anchor=north] at (1.55,-1.45) {one wire per item};
    \fi
    \node[ptitle] at (1.55,\capSD) {\numSD string diagram};
  \end{scope}

  \ifnum\mtfbrief=0

  \begin{scope}[shift={(0,\rowB+0.18)}]
    \node[T]                (na)   at (0,0)          {$a$};
    \node[blueP, tokens=1]  (npo)  at (2.3,1.15)     {};
    \node[redP,  tokens=1]  (npi1) at (1.15,-0.30)   {};
    \node[T]                (np1)  at (2.3,-0.30)    {$p$};
    \node[redP]             (npq1) at (3.45,-0.30)   {};
    \node[redP,  tokens=1]  (npi2) at (1.15,-1.50)   {};
    \node[T]                (np2)  at (2.3,-1.50)    {$p$};
    \node[redP]             (npq2) at (3.45,-1.50)   {};
    \node[T]                (ns)   at (4.6,0)        {$s$};

    \draw[blueflow] (na)   -- (npo);
    \draw[blueflow] (npo)  -- (ns);
    \draw[redflow]  (na)   -- (npi1);
    \draw[redflow]  (npi1) -- (np1);
    \draw[redflow]  (np1)  -- (npq1);
    \draw[redflow]  (npq1) -- (ns);
    \draw[redflow]  (na)   -- (npi2);
    \draw[redflow]  (npi2) -- (np2);
    \draw[redflow]  (np2)  -- (npq2);
    \draw[redflow]  (npq2) -- (ns);

    \node[ptitle] at (2.3,{\ycapB-0.18}) {(iv) object-centric Petri net 2};
  \end{scope}
  \fi

  \ifnum\mtfbrief=1
    \draw[rule] (5.80,-1.85) -- (5.80,1.65);
    \node[blueP, minimum size=4mm, label distance=2pt, label={[plab]right:order}] at (0.2,\legendY) {};
    \node[redP,  minimum size=4mm, label distance=2pt, label={[plab]right:item}]  at (2.2,\legendY) {};
  \fi

\end{tikzpicture}

%% file: tikz/introduction/routing_collapse.tex
\begin{tikzpicture}[petriBase, inline, scale=1.05, every node/.style={transform shape},
    silentT/.style={transition, fill=black!22},
    place/.append style={minimum size=6mm},
    transition/.append style={minimum size=6mm, font=\small},
    plab/.style={font=\footnotesize, black!60, inner sep=1pt},
    ptitle/.style={font=\small\bfseries, inner sep=2pt},
    note/.style={font=\scriptsize, black!55, inner sep=1pt},
    rule/.style={black!30, line width=0.5pt}]

  \def\ytop{1.0}
  \def\ybot{-1.0}
  \def\ycap{-3.05}

  \begin{scope}[shift={(0,0)}]
    \node[blueP, tokens=1] (Ap0) at (0,0)          {};
    \node[T]               (Aa)  at (1.0,0)        {$a$};
    \node[redP]            (Api) at (2.0,\ytop)    {};
    \node[blueP]           (Apo) at (2.0,\ybot)    {};
    \node[T]               (Ab)  at (3.0,\ytop)    {$b$};
    \node[T]               (Ac)  at (3.0,\ybot)    {$c$};

    \draw[blueflow] (Ap0) -- (Aa);
    \draw[redflow]  (Aa)  -- (Api);
    \draw[blueflow] (Aa)  -- (Apo);
    \draw[redflow]  (Api) -- (Ab);
    \draw[blueflow] (Apo) -- (Ac);

    \node[ptitle] at (1.5,\ycap) {(i) OCPN, miner 1};
  \end{scope}

  \begin{scope}[shift={(4.5,0)}]
    \node[blueP, tokens=1] (Bp0)  at (0,0)         {};
    \node[T]               (Ba)   at (1.0,0)       {$a$};
    \node[redP]            (Bpi1) at (1.95,\ytop)  {};
    \node[silentT]         (Bt)   at (2.85,\ytop)  {$\tau$};
    \node[redP]            (Bpi2) at (3.75,\ytop)  {};
    \node[blueP]           (Bpo)  at (2.4,\ybot)   {};
    \node[blueP]           (Bpo2) at (2.85,-1.95)  {};
    \node[T]               (Bb)   at (4.65,\ytop)  {$b$};
    \node[T]               (Bc)   at (4.65,\ybot)  {$c$};

    \draw[blueflow] (Bp0)  -- (Ba);
    \draw[redflow]  (Ba)   -- (Bpi1);
    \draw[redflow]  (Bpi1) -- (Bt);
    \draw[redflow]  (Bt)   -- (Bpi2);
    \draw[redflow]  (Bpi2) -- (Bb);
    \draw[blueflow] (Ba)   -- (Bpo);
    \draw[blueflow] (Bpo)  -- (Bc);
    \draw[blueflow] (Ba.south) to[out=-90,in=180] (Bpo2);
    \draw[blueflow] (Bpo2) to[out=0,in=-90] (Bc.south);

    \node[note, anchor=south] at (2.85,\ytop+0.42) {silent split};
    \node[note, anchor=north] at (2.85,-2.32)      {implicit place};

    \node[ptitle] at (2.3,\ycap) {(ii) OCPN, miner 2};
  \end{scope}

  \draw[rule] (10.35,-2.75) -- (10.35,1.75);

  \begin{scope}[shift={(11.55,0)}]
    \node[morphismblue] (Da) at (0,0)       {$a$};
    \node[morphismblue] (Db) at (2.0,\ytop) {$b$};
    \node[morphismblue] (Dc) at (2.0,\ybot) {$c$};

    \draw[bluewire] (Da.west) -- ++(-0.55,0);
    \draw[redwire]  ([yshift=1.6mm]Da.east) to[out=0,in=180] (Db.west);
    \draw[bluewire] ([yshift=-1.6mm]Da.east) to[out=0,in=180] (Dc.west);
    \draw[redwire]  (Db.east) -- ++(0.55,0);
    \draw[bluewire] (Dc.east) -- ++(0.55,0);

    \node[ptitle] at (1.0,\ycap) {(iii) one string diagram};
  \end{scope}

  \node[blueP, minimum size=4mm, label distance=2pt, label={[plab]right:order}] at (0.2,-3.95) {};
  \node[redP,  minimum size=4mm, label distance=2pt, label={[plab]right:item}]  at (2.2,-3.95) {};

\end{tikzpicture}

%% file: sections/02-categories.tex
\section{String Diagrams as a Composition Engine for Process Models}\label{sec:hypergraphs}

Section~\ref{sec:introduction} read the string diagram of Figure~\ref{fig:oc-two-faces}(vi)
informally, and Table~\ref{tab:intuition} recorded its vocabulary with a pointer to each term's
definition, most of them into this section. Here that reading becomes precise. The section hands
the reader two things, a way to record one fragment of a process as a diagram whose boundary is
its typed input and output wires, and a way to glue two fragments along a shared boundary. That
is the equipment the rest of the paper runs on. It is what lets a whole model be assembled from a
finite set of small generators, and what lets two models, drawn in different notations, be
compared piece by piece once their routing devices are gone. A reader already at home in
hypergraph categories may skim ahead to Section~\ref{sec:general-framework} and lose nothing.

\subsection{From the picture to the algebra}

The four notations of Figure~\ref{fig:oc-two-faces} record the same structure by different devices.
The Petri net of panel~(i) makes the buffer of pending items an explicit typed place and writes the count as an
arc weight. The causal net has no places and marks the arcs instead, following \citet{lissObjectCentricCausalNets2025}, a circle for a single object against a square for
several, with the dashed links grouping the markers that must be satisfied together. The process
tree has neither, and carries the count in the arity of its parallel operator. The BPMN model
carries it in the branches opened by one gateway and closed by another. Each notation
is recovered from the string diagram by keeping only the structure it records and dropping the rest.
Pooling the two item wires into one typed place and the two $p$ boxes into one transition that fires
twice returns the Petri net, while carrying the same multiplicity on $s$'s binding returns the
causal net. We make these projections precise in Section~\ref{sec:conversion-framework}.

These properties recur across process-mining notations, yet comparing models even
within one notation, and saying what a model minimally is, remain persistently
hard~\cite{vanderaalstRepresentationalBiasProcess2011,vanderaalstImprovingRepresentationalBias2012}.
String diagrams answer both needs. They are a common language in which to state when two
models are equal, and a compositional one in which complex models are assembled
from simple generators. They are the syntax of hypergraph categories, which
we now introduce.

\subsection{Notation and conventions}\label{subsec:conventions}

We fix the notation used throughout. Objects are typed resources drawn from a finite set of object
types $\mathcal{O}$, the untyped case being $\mathcal{O}=\{\ast\}$. Morphisms compose
sequentially by $(;)$ and in parallel by the monoidal product $(\otimes)$, with unit $I$ and
identities $\mathrm{id}_X$. Composing two fragments sequentially means gluing them end to end. The
wires that one hands to the other are identified and everything else is kept. In category theory that gluing
is a pushout of cospans, and Definition~\ref{def:pushout} states it.

Finite multisets over a set $X$ are written $\mathbb{N}^{X}$, the functions $X\to\mathbb{N}$ of
finite support, equivalently the free commutative monoid on $X$, so the value at $x$ is the
multiplicity of $x$. A multiset is displayed as a formal sum, for example $2a+3b$, and multisets
are combined by $+$, $-$, and $\subseteq$. A Petri-net marking is such a multiset,
$M\in\mathbb{N}^{P}$ with $M(p)$ tokens in place $p$, drawn as an object by the tensor power
$\bigotimes_{p\in P}p^{\otimes M(p)}$, and typed multisets form $\mathbf{FinMSet}/\mathcal{O}$. The
letter $M$ denotes a marking throughout, while models are written $m_1,m_2$. Ordinary sets use
single braces $\{\dots\}$. A nested $\{\{\dots\}\}$ is a set whose elements are themselves sets, for
instance a set of subsets, and never a multiset. We do not use bag braces.

\subsection{Split, merge, start, stop}

We model process behaviour with concurrency and causal structure using hypergraph
categories, following the string-diagrammatic framework of
\citet{bonchiStringDiagramRewrite2022,bonchiStringDiagramRewrite2022a,bonchiStringDiagramRewrite2022b}.
Figure~\ref{fig:oc-two-faces}(vi) supplies the running intuition. Wires are
objects, the resources a process manipulates and the interfaces through which
subsystems interact. Boxes are activities, which consume objects and produce
others for later use. A whole model is a particular composition of atomic
boxes, complexity built from simple parts, and for a Petri-net reader the same
sentences read place for wire and transition for box.

\paragraph{Conventions.} String diagrams are read left to right. A box's inputs enter on
its left, its outputs leave on its right, and a wire runs from the box that produces a
resource to the box that consumes it. Two operations build diagrams. Sequential
composition $(;)$ connects the outputs of one box to the inputs of the next, and parallel
composition $(\otimes)$ places boxes side by side without connecting them. We write
$X\xmapsto{\llangle \cdot\cdot\cdot \rrangle} D$ to mean that the algebraic object $X$ is
drawn as the string diagram $D$.

We start with how resources, such as Petri net tokens or the typed objects of
Fig.~\ref{fig:oc-two-faces}, may begin, end, and branch within a diagram. The
structure that governs this is known, somewhat intimidatingly, as a special
commutative Frobenius algebra. The name is not important, and we explain instead
what its four operations do. Working through them is the gentlest way into the
hypergraph-category machinery that follows, since a reader who follows the four
operations has met its central idea.

Figure~\ref{fig:oc-two-faces}(vi) needs two of them. The order enters
at $a$ with no producing activity and leaves after $s$ with no consumer, so a
wire must be able to begin ($\eta_X$) and to end ($\epsilon_X$), the bare
boundaries to which $a$ and $s$ attach. The two item wires need no operation,
being two outputs of the one box $a$. Where one object is shared between
genuinely concurrent activities its wire must be able to split ($\delta_X$), and
the mirror operation merges the copies back into one ($\mu_X$). Most flow uses
none of the four. An object
handed from one activity to the next is a plain wire, a choice between
continuations is several separate diagrams rather than one branching wire, and
a multiplicity of objects is a count carried by parallel wires. The four
operators appear exactly where a process begins, ends, or shares one object
between genuinely concurrent activities, and that economy is what keeps
diagrams readable.

\begin{definition}[Special commutative Frobenius algebra]
Let $X$ be a resource (or object), possibly carrying typing information, in some space of resources. A special commutative Frobenius algebra structure on $X$ is a tuple $(X,\delta_X,\mu_X,\eta_X,\epsilon_X)$ of the four maps drawn below.
\end{definition}
\begin{align*}
\delta_X
\xmapsto{\llangle \cdot\cdot\cdot \rrangle}
\vcenter{\hbox{\input{tikz/categories/fig_03}\unskip}}
\quad\quad
\mu_X
\xmapsto{\llangle \cdot\cdot\cdot \rrangle}
\vcenter{\hbox{\input{tikz/categories/fig_04}\unskip}}
\\
\eta_X
\xmapsto{\llangle \cdot\cdot\cdot \rrangle}
\vcenter{\hbox{\input{tikz/categories/fig_05}\unskip}}
\quad\quad
\epsilon_X
\xmapsto{\llangle \cdot\cdot\cdot \rrangle}
\vcenter{\hbox{\input{tikz/categories/fig_06}\unskip}}
\end{align*}
\medskip
These operators are informally known as ``split'' ($\delta_X$), ``merge'' ($\mu_X$), ``initialise'' ($\eta_X$) and ``terminate'' ($\epsilon_X$).

The operators $(\delta_X,\mu_X,\eta_X,\epsilon_X)$ satisfy a number of equational axioms, listed in
full in Appendix~\ref{sec:axioms}. The most immediately useful is speciality~\eqref{eq:ax-special},
$\delta_X ; \mu_X = \mathrm{id}_X$, which says that splitting a resource and immediately re-merging
the two copies recovers the original wire unchanged.

Merge and split operators are equivalent up to topological deformation, so only the connections
matter and only the connections drawn with a black dot are interactions. The spider
theorem~\cite{bonchiStringDiagramRewrite2022} puts any such diagram into a normal form.

Splitting and merging concern copies of one object. Placing
several objects, or several independent activities, side by side is a different operation,
parallel composition, which sets resources next to one another without connecting them and
so records that they proceed independently. We write it with the tensor product
$(\otimes)$. The two packs of Figure~\ref{fig:oc-two-faces}(vi) form exactly this
composite, $p\otimes p$, two activities with no wire between them and no implied
order. In a Petri net the same operation is what a marking needs. The multiplicity of an
object, the number of parallel wires of a given type, is the number of tokens in the
corresponding place, and a marking is insensitive to the order of its tokens, so it is
represented by a tensor of object copies indexed by places. Let \(P\) be a set of Petri net places viewed as the objects (i.e.\ the space of resources for a model). A marking of tokens is a function
\[
M : P \to \mathbb{N},
\]
and its corresponding representation in the hypergraph category is denoted as,
\[
M
\xmapsto{\llangle \cdot\cdot\cdot \rrangle}
\bigotimes_{p \in P} p^{\otimes M(p)},
\]
where \(p^{\otimes M(p)}\) denotes the \(M(p)\)-fold tensor power of \(p\), and in the case every place has zero tokens we define $p^{\otimes{}0}=I$ where $I$ is the unit (e.g.\ zero tokens) of the category. We interpret the action of spider diagrams as creating or consuming tokens ``for free'', in the sense an intermediate activity is not needed.
This convention makes multiplicity explicit at the level of objects. Each token in place \(p\) is represented by one copy of the category object representing \(p\), so the total token content of a marking is encoded by the number of parallel wires of each object.

Later sections endow other notations (causal nets, process trees, BPMN) with the same token-like semantics, placing all of them on a common compositional footing.

However, for process mining, we are still missing the activities that capture process behaviour. The spiders only carry input or output information, e.g.\ properties of tokens. We require a way to take places (or more accurately token resources) and transform them. This is done by introducing morphisms, known as boxes, which represent an activity in a process model.

Assembling boxes and objects needs a categorical setting with the two compositions we have
used above, sequential $(;)$ and parallel $(\otimes)$. That setting is a symmetric monoidal
category, one with these two compositions in which the order of parallel factors
does not matter and wires may cross. A hypergraph category is one in which, in addition,
every object carries the Frobenius structure above.

\begin{definition}[Hypergraph category]\label{def:hypergraph-category}
A hypergraph category is a symmetric monoidal category in which every
object $X$ carries a chosen special commutative Frobenius algebra structure
$(X,\delta_X,\mu_X,\eta_X,\epsilon_X)$ compatible with the parallel product $(\otimes)$.
Compatibility means the Frobenius structure on $X\otimes Y$ is built component-wise from those of $X$ and $Y$.
\end{definition}

A morphism, or activity or box, can have multiple inputs and outputs, corresponding to the objects of the hypergraph category. The boxes of
Figure~\ref{fig:oc-two-faces}(vi) are all morphisms of this kind. The accept
$a$ has the empty input boundary $I$ and the output boundary
$\mathrm{order}\otimes\mathrm{item}^{\otimes 2}$, each pack $p$ maps
$\mathrm{item}$ to $\mathrm{item}$, and the ship $s$ mirrors $a$. A single
transition thus consumes multiple inputs and produces multiple outputs, like a
Petri net transition. Typing the objects in $P$ yields a formalism for
object-centric Petri nets, which we construct explicitly in
Section~\ref{sec:petri-nets}. For now we give the general definition of an
activity morphism in a hypergraph category.
\begin{definition}[Activity morphism]
Let $P$ denote the set of objects in the hypergraph category $\mathbf{H}$ and $n,n':P\to\mathbb{N}$ denote input and output multiplicity functions. An activity in a process model is represented by a morphism $f$ as,
\[ f : \bigotimes_{X\in{}P}X^{\otimes{}n_X} \rightarrow \bigotimes_{U\in{}P}U^{\otimes{}n'_U} \]
\end{definition}
We now have the basic building blocks for representing process models, namely objects (places), morphisms (activities) and special commutative Frobenius algebra structures (spiders), which allow us to copy, merge, create and terminate resources. To compute with them, and to read them off process models, we need a precise, machine-friendly way to record a diagram together with how it composes. Section~\ref{sec:introduction} previewed the answer, a cospan, a fragment with an input boundary and an output boundary that glue along shared boundaries. We develop this as the cospan-based semantics of hypergraph categories, well suited to reformulating process-mining models as compositional networks.

In particular, we will work with so-called decorated cospans~\cite{fongDecoratedCospans2015}, which record the compositional structure of process models but endow them with extra information that allows us to draw string diagrams.
\begin{definition}[Decorated cospan]\label{def:decorated-cospan}
A decorated cospan is a morphism of the form,
\[
f
\xmapsto{\llangle \cdot\cdot\cdot \rrangle}
\left(X_{1}\xrightarrow{i} A \xleftarrow{o} X_{2},\; d\right)
\]
where $X_1$ and $X_2$ are the left and right boundaries, $A$ is the apex, $i:X_1\to A$ and $o:X_2\to A$ are the boundary maps, and $d$ is a decoration over $A$ which gives the information needed to draw the string diagram.
\end{definition}
\begin{remark}[Decorations]
There are many ways to formulate decorations, but in this paper we will concretely define them as tuples of the form $(E,s,t,\lambda)$ for $E$ a set of hyperedges, $s,t:E\to A^{*}$ source and target maps for the edges, and $\lambda:E\to L$ a labelling function. The source and target maps give the information needed to draw the string diagram, while the labelling function allows us to address possible repeated activities.
\end{remark}

The $p$ box of Figure~\ref{fig:oc-two-faces}(vi) is already such a cospan. Each
boundary is one item port, the apex is those two ports, and the decoration is
the single hyperedge labelled $p$ from the input port to the output port.

\begin{remark}
For more complex morphisms the boundaries expand to uniquely labelled inputs and outputs, the apex is the set of all boundary and internal ports, and the decoration $d$ maps the inputs to the outputs with any necessary internal structure.
\end{remark}

Composition has been the running theme. A model is assembled by gluing fragments
along shared boundaries, and for cospans that gluing is a standard operation, the
pushout. The plain reading comes first. Composing $g_1;g_2$ glues the output
ports of $g_1$ onto the input ports of $g_2$, port by port and type by type, and
keeps both decorations, so the composite is again a valid string diagram. Under
the unique-ID convention the paper works with throughout, where every port
carries a globally unique integer identifier, the whole construction is a set
union.
\begin{definition}[Composition of decorated cospans (cospan pushout)]\label{def:pushout}
Let
\begin{align}
g_1 &= \Bigl(X_1 \xrightarrow{i_1} A_1 \xleftarrow{o_1} X_2,\; d_1\Bigr),\\
g_2 &= \Bigl(X_2 \xrightarrow{i_2} A_2 \xleftarrow{o_2} X_3,\; d_2\Bigr)
\end{align}
be decorated cospans sharing interface $X_2$, with ports drawn from a common
pool of unique identifiers, so that $X_2\subseteq A_1$, $X_2\subseteq A_2$, and
$o_1(x)=x=i_2(x)$ for every $x\in X_2$. Their composite is
\begin{equation}
g_1;g_2 \;:=\;
\Bigl(X_1\hookrightarrow A_1\cup A_2\hookleftarrow X_3,\; d_1\cup d_2\Bigr),
\end{equation}
where the apex is the set union, the boundary maps are the set inclusions, and
the composite decoration carries every hyperedge of $d_1$ and of $d_2$
unchanged. A hyperedge of $d_1$ incident to a port $x\in X_2$ and a hyperedge
of $d_2$ incident to the same port now share that vertex, which is the gluing.
\end{definition}
This is the working form of a general construction. For arbitrary boundary maps
the apex is the quotient of the disjoint union $A_1\sqcup A_2$ by the
identifications $o_1(x)\sim i_2(x)$, a pushout of finite sets, and the unique-ID
convention is what collapses that quotient to the union above, since the two
copies of each interface port carry the same identifier. The general two-step
construction is standard~\cite{fongDecoratedCospans2015} and adds nothing the
paper uses. The convention is maintained inductively. Internal ports of a newly
attached generator take fresh identifiers, and interface ports inherit the
identifiers of the shared boundary.

\begin{remark}[Process mining reading]
Think of ports as representing places, the medium through which transitions communicate with each other.
The interface $X_2$ is the set of ports that $g_1$ hands off to $g_2$.
The map $o_1$ locates those ports inside diagram $A_1$, while $i_2$
locates the same ports inside diagram $A_2$.
Composition fuses places together in the composite apex $A_1\cup A_2$, gluing the output
boundary of $g_1$ directly onto the input boundary of $g_2$.
In Figure~\ref{fig:oc-two-faces}(vi) the composition after $a$ is exactly this fusion,
$a$'s three output ports glued by type onto the input ports of the fragment
$\mathrm{id}_{\mathrm{order}}\otimes p\otimes p$, and the composition before $s$ repeats it.
\end{remark}

\begin{remark}[Typed ports]
If ports carry types, for example token object type in an object-centric
Petri net, the identification of $o_1(x)\in A_1$ with $i_2(x)\in A_2$ is imposed only when
$\mathrm{type}(o_1(x))=\mathrm{type}(i_2(x))$, and
ports of mismatched type cannot be fused.
\end{remark}

The pushout construction is heavier on the page than in practice, and an example makes it concrete.

\begin{exmp}
Consider the cospans $g_1$ and $g_2$ defined below. We compute their composition by pushout, using port types to describe the unique-ID convention,
\begin{strip}
\begin{align}
g_1;g_2
& \xmapsto{\llangle \cdot\cdot\cdot \rrangle}
\vcenter{\hbox{\resizebox{\ifdim\width>0.40\linewidth 0.40\linewidth\else\width\fi}{!}{\input{tikz/categories/fig_22}\unskip}}}
\;{};\;{}
\vcenter{\hbox{\resizebox{\ifdim\width>0.40\linewidth 0.40\linewidth\else\width\fi}{!}{\input{tikz/categories/fig_23}\unskip}}}
\\
& \xmapsto{\llangle \cdot\cdot\cdot \rrangle}
\vcenter{\hbox{\resizebox{\ifdim\width>0.85\linewidth 0.85\linewidth\else\width\fi}{!}{\input{tikz/categories/fig_24}\unskip}}}\label{eq:decoration-example}
\\
& \xmapsto{\llangle \cdot\cdot\cdot \rrangle}
\vcenter{\hbox{\resizebox{\ifdim\width>0.85\linewidth 0.85\linewidth\else\width\fi}{!}{\input{tikz/categories/fig_25}\unskip}}}\label{eq:string-diagram-example}
\end{align}
\end{strip}
\end{exmp}

The composite is the cospan $\left(X_{1}\xrightarrow{i} A \xleftarrow{o} X_{3}, d_{3}\right)$, whose
decoration carries one hyperedge per box and per spider, each recording the ports it reads and the
ports it writes together with its label. Reading those hyperedges back through the labelling
function returns the morphism $d_{3}\sim(f\otimes{}k);(\mu\otimes{}\mathrm{id});h;\mu$. The port
colours are the abstract wire types we return to in the object-centric Petri net analysis.

Many different composites of $(\mu,\delta)$ and $f,k,h$ denote this same string diagram, which is
the point, since only the connectivity of the diagram carries information.

Having described the fundamental components of hypergraphs, and how they informally map across to Petri nets, the remaining step is to describe a whole model with cospans.

The device is a set of cospans from which all behaviour of the model can be reconstructed. Such sets are called cospan signatures.

\subsection{From a model to its signature}

The link between composition and process mining is simple. Each transition of a
model, with its input and output data, is a decorated cospan, a generator, and these
generators are the building blocks. The model's signature is the set of them. Composing
generators yields the model's connected diagrams, and each diagram's traces are its
linearisations (Section~\ref{sec:semantics}), up to the deformation and reordering the
hypergraph structure permits. The signature is thus a finite set of building-block
transitions that generate the category. We now define it and unravel each component.
\begin{definition}[Cospan-algebra signature]
A cospan-algebra signature, also just called a signature, $\Sigma$ is a set of cospans $\left\{g_t = \left(X_i\xrightarrow{i_t} A_t \xleftarrow{o_t} X_j,\; d_t\right) : t \in T\right\}$ where $T$ is some indexing set, $X_i$ and $X_j$ are boundary objects indexed by integers.
\end{definition}
Intuitively, a signature is a finite set of cospan objects from which all possible morphisms in a hypergraph category can be constructed. For process mining, we will be interested in signatures which are generated by the transitions of a process model, and the cospans will be defined by the input and output data of these transitions.
The signature of the running example is the set of its three boxes,
$\Sigma=\{a,p,s\}$ with their typed boundaries, and
Figure~\ref{fig:oc-two-faces}(vi) is one composite of those generators, with
$p$ used twice.

The final ingredient we need to map traditional process models to the language of hypergraph categories is how to recover the full behaviour of a model from its signature. Typically, we construct what is known as the free hypergraph category $\mathbf{H}(\Sigma)$ generated by $\Sigma$, which holds every diagram obtainable by composing the generators in any sequential and parallel arrangement, all realisations of the subprocesses running in every configuration. This is far too large a space for most process mining discovery settings, since it includes wholly independent, concurrently-running processes. We look instead at a constrained space of morphisms which still describes a process model but enables discovery of behaviour.
\begin{definition}[Connected process diagram]
Let $\Sigma$ be a signature. A connected process diagram is a diagram with any boundary, built from generators in $\Sigma$,
that contains at least one generator occurrence and satisfies two conditions.
\begin{itemize}
\item \textbf{Connected.} All generator occurrences are joined by wires, so that there is a path of wires between any two of them
\item \textbf{No branching or merging.} There are no Frobenius operators (i.e.\ no wire mergers or splits) in the diagram
\end{itemize}
\end{definition}
A single executable run therefore uses none of the Frobenius operators introduced
above. The split-and-merge structure lives only in the ambient model, where one
branching wire stands compactly for the many runs that resolve it. This is the
precise sense in which those operators are model-level scaffolding rather than part
of any individual behaviour.
\begin{exmp}[Connected process diagram]
Consider the string diagram in Eq.~\ref{eq:string-diagram-example}. Forgetting types, and replacing any mergers or splits with individual wires results in a connected process diagram. Additionally, rewiring the output of box $k$ so that it passes through the diagram results in a disconnected diagram, which the connected process diagrams do not include. Consequently we have the diagrams as in Eq.~\ref{eq:connected-process-diagram}.
\begin{strip}
\begin{align}\label{eq:connected-process-diagram}
\vcenter{\hbox{\input{tikz/categories/fig_26}\unskip}}
\hspace{1cm}
\vcenter{\hbox{\input{tikz/categories/fig_27}\unskip}}
\end{align}
\end{strip}
\end{exmp}
Connected process diagrams are the morphisms of $\mathbf{F}(\Sigma)$
(Definition~\ref{def:free-category}) whose occurrences stay connected, and they form the
default selection of Section~\ref{subsec:theorem}. In practice a
bounded depth-first search over composites, deduplicated up to isomorphism, suffices. We
do not dwell on the enumeration here. The signature is also the object over which models written in different notations are later compared, since each notation yields one by the same construction (Section~\ref{sec:conversion-framework}).

\subsection{Reading traces off a diagram}\label{sec:semantics}

A connected string diagram in our hypergraph category framework encodes a process at the level of causal structure rather than as a single execution sequence. Intuitively, each box represents an action in the process, while the wires record which actions must occur before others and which may occur independently. Thus the diagram captures the dependency pattern shared by all traces in the same equivalence class.

The traces a string diagram represents are found by extracting a partial order from it and enumerating all its linear extensions, or total orders. We start with obtaining the underlying causal structure of the string diagram.
\begin{definition}[Causal partial order of a connected process diagram]\label{def:causal-partial-order}
Given a connected process diagram, its causal partial order is the partial order generated
by $a \le b$ whenever an output of box $a$ is directly connected to an input of box $b$.
As a relation it is a directed acyclic graph (DAG).
\end{definition}
Each node of the DAG represents a unique occurrence of a process activity, even if multiple nodes share the same activity label via some labelling function $\ell$. The transitive closure of the DAG is the causal partial order on activity occurrences, capturing the causal dependencies between them. The set of traces associated with a connected process diagram is then obtained by enumerating all linear extensions of this partial order and mapping nodes to their activity labels via $\ell$.
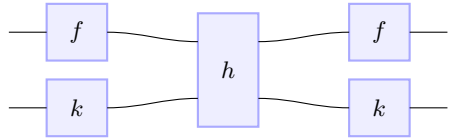
\begin{figure}[htbp]
\centering
\input{tikz/categories/fig_28}\unskip
\caption{String diagram with repeated boxes}\label{fig:repeated-boxes}
\end{figure}
\begin{exmp}[Traces of a string diagram]
Consider the string diagram in Fig.~\ref{fig:repeated-boxes}. Algebraically this is written as $(f\otimes{}k);h;(f\otimes{}k)$. Define a labelling by function $\ell:(f\otimes{}k);h;(f\otimes{}k)\to{}(f_1\otimes{}k_1);h;(f_2\otimes{}k_2)$. The partial orders can be read off as $f_1\le{}h$, $k_1\le{}h$, $h\le{}f_2$, $h\le{}k_2$. The Hasse diagram for this partial order is then given by,
\begin{center}
\input{tikz/categories/fig_29}\unskip
\end{center}
The total orders associated with this partial order are $\langle f_1k_1hf_2k_2\rangle$, $\langle k_1f_1hf_2k_2\rangle$, $\langle f_1k_1hk_2f_2\rangle$ and $\langle k_1f_1hk_2f_2\rangle$. Applying the inverse labels $\ell^{-1}$, and taking the unique set of total orders if necessary, gives us the traces represented by the string diagram, $\langle fkhfk\rangle$, $\langle kfhfk\rangle$, $\langle fkhkf\rangle$ and $\langle kfhkf\rangle$.
\end{exmp}
This reading is the trace semantics
\[
\Gamma : \mathbf{F}(\Sigma) \to \mathcal{P}(\mathcal{L}^\ast),
\]
which sends each morphism $f$ of the free symmetric monoidal category $\mathbf{F}(\Sigma)$ to its set of traces $\Gamma(f)$, where $\mathcal{L}^\ast$ denotes all finite sequences of activity labels. Section~\ref{sec:general-framework} defines $\mathbf{F}(\Sigma)$ and $\Gamma$ precisely (Definitions~\ref{def:free-category} and~\ref{def:trace-map}), with $\Gamma(f)$ empty when $f$ is unsatisfiable, and the language $L_{\mathcal{S}}(\Sigma)$ of a selection of morphisms. Here, the causal order of $f$ is the partial order induced by its DAG, and its traces are the linear extensions of that order. Intuitively, traces are obtained by listing activity occurrences in any order that respects the flow of execution defined by the diagram. Thus, trace semantics are obtained by first extracting the causal dependencies between activity occurrences, then enumerating all execution orders consistent with these dependencies, and finally mapping occurrences to activity labels if necessary.

On the running example the causal order has four occurrences, $a$ below the two
packs and both packs below $s$. It has two linear extensions, which differ only
in the order of the two $p$ occurrences, so the labelling folds them into the
single trace $\langle a,p,p,s\rangle$, and $\Gamma$ returns exactly that set.
The packs' concurrency is visible in the diagram and absent from the trace,
the same loss that Section~\ref{sec:introduction} laid out for sequence-based
comparison.

\begin{proposition}[Normal form of a scenario]\label{prop:scenario-normal-form}
Let $f$ be a scenario, a morphism of $\mathbf{F}(\Sigma)$. Any two factorisations of $f$ into generators induce the same causal partial order $D_f$ (Definition~\ref{def:causal-partial-order}) up to isomorphism, called the causal poset of $f$. The factorisations of $f$ into a total generator sequence are in bijection with the linear extensions $\mathrm{Lin}(D_f)$. Consequently $\Gamma(f)$ is the image under $\ell$ of $\mathrm{Lin}(D_f)$, and $f$ determines one labelled poset up to isomorphism. A fully resolved scenario in this sense is what partial-order accounts of process behaviour call a variant, one choice-free run presented as a single labelled poset. A loop in the model gives no cycle in a scenario and appears as repeated composition of generators (Section~\ref{subsec:free-category}).
\end{proposition}
\begin{proof}
The generators appearing in $f$ are its activity occurrences. Each wire $(a,w,b)$ of $f$ records that an output of the occurrence $a$ feeds an input of the occurrence $b$, so it fixes a covering pair $a\le b$ of the occurrence DAG. The DAG, and hence its transitive closure $D_f$, is therefore determined by the wiring of $f$ alone and does not depend on the chosen factorisation. Two factorisations of $f$ differ only by reordering generators that share no wire, an application of the interchange law that leaves the DAG unchanged, which gives the invariance up to isomorphism. A total factorisation lists the occurrences in an order that respects every wire, so total factorisations are exactly the topological sortings of the DAG, which are the linear extensions of $D_f$. Applying the labelling $\ell$ to each such sequence yields $\Gamma(f)$.
\end{proof}

\begin{remark}
In process mining, and Petri net theory more generally, the language over a net is tied to an initial marking $M_0$, with the possibility of imposing some final marking $M_1$. The trace semantics defined here is instead a property of a diagram, so $\Gamma$ applies to every morphism of $\mathbf{F}(\Sigma)$, whatever its boundary. That makes it notation-independent. A model's trace language is the union of $\Gamma$ over a selection of morphisms. Start and end generators that play the part of $M_0$ and $M_1$ are a choice made in the analysis of a model (Section~\ref{subsec:analysis}), and they leave $\Sigma$ unchanged.
\end{remark}

Computationally, realising $\Gamma$ means listing the linear extensions of the
causal partial order, which standard algorithms do in time proportional to the
output~\cite{knuthArtComputerProgramming2011,ruskeyCombinatorialGeneration2003}.
Counting them instead is
$\#$P-complete~\cite{brightwellCountingLinearExtensions1991}.

To bring the section together, here is the whole construction read back in process-mining
terms. A signature is the set of a model's generators, its building-block transitions,
each a transition with its input and output data. Composing these generators along shared
typed boundaries assembles the morphisms of $\mathbf{F}(\Sigma)$, each
diagram a partial order of activities. The trace semantics $\Gamma$ then turns each diagram
back into observable behaviour, the set of execution orders it admits. Composition is thus
how a model is built from its transitions, and $\Gamma$ is how its trace-language
behaviour is recovered. This makes $\Gamma$ a deliberately lossy projection, since it reduces each diagram to the set of its linear extensions and drops the concurrency that the partial order records, and the signature-level certificate of Section~\ref{sec:conversion-framework} later recovers what $\Gamma$ discards. The one ingredient still missing is where a model's generators come
from. For every notation this paper treats, object-centric Petri nets, causal nets,
process trees, and BPMN, they are not chosen by hand but read off the model by a single
construction, the same for all four. We turn to it now (Section~\ref{sec:general-framework}).

%% file: tikz/categories/fig_03.tex
\begin{tikzpicture}[inline]
\node[spider] (s) at (0,0) {};
\draw[wire] (s) to[out=+80, in=180] (0.5, +0.4);
\draw[wire] (s) to[out=-80, in=180] (0.5, -0.4);
\draw[wire] (s) -- (-0.5, 0);
\end{tikzpicture}

%% file: tikz/categories/fig_04.tex
\begin{tikzpicture}[inline]
  \node[spider] (s) at (0,0) {};
  \draw[wire] (s) to[out=+100, in=0] (-0.5, +0.4);
  \draw[wire] (s) to[out=-100, in=0] (-0.5, -0.4);
  \draw[wire] (s) -- (0.5, 0);
\end{tikzpicture}

%% file: tikz/categories/fig_05.tex
\begin{tikzpicture}[inline]
\node[spider] (s) at (0,0) {};
\draw[wire] (s) -- (+0.6, 0);
\end{tikzpicture}

%% file: tikz/categories/fig_06.tex
\begin{tikzpicture}[inline]
\node[spider] (s) at (0,0) {};
\draw[wire] (s) -- (-0.6, 0);
\end{tikzpicture}

%% file: tikz/categories/fig_22.tex
\begin{tikzpicture}[inline]
  \node[cospanblue, minimum height=3cm] (b1) at (-2.5,0) {};
  \node[lc={black}{2}, spider] (s1) at (-2.5,+0.750) {};
  \node[lc={black}{1}, spider] (s2) at (-2.5,-0.000) {};
  \node[lc={black}{0}, spider] (s5) at (-2.5,-1.000) {};
  \node[] (a1) at (-1.7,0) {$\rightarrow$};
  \node[cospangrey, minimum height=3cm, minimum width=2.5cm] (m1) at (0,0) {};
  \node[draw, rounded corners=5pt, minimum size=0.75cm, fill=white] (f1) at (0,0.375) {$f$};
  \draw[wire] (f1.160) to [out=135, in=0] ++(-0.5,+0.25);
  \draw[wire] (f1.200) to [out=225, in=0] ++(-0.5,-0.25);
  \draw[bluewire] (f1.20) to [out=45, in=180] ++(+0.5,+0.25);
  \draw[bluewire] (f1.340) to [out=315, in=180] ++(+0.5,-0.25);
  \node[lc={black}{2}, spider] at ([shift={(-0.9,0.4)}]f1) {};
  \node[lc={black}{1}, spider] at ([shift={(-0.9,-0.4)}]f1) {};
  \node[lc={blue}{2}, spider] at ([shift={(0.9,0.4)}]f1) {};
  \node[lc={blue}{1}, spider] at ([shift={(0.9,-0.4)}]f1) {};
  \node[draw, rounded corners=5pt, minimum size=0.75cm, fill=white] (f2) at (0,-1.000) {$k$};
  \draw[-] (f2.180) to [out=180, in=0] ++(-0.5,0);
  \draw[bluewire] (f2.0) to [out=0, in=180] ++(+0.5,0);
  \node[lc={black}{0}, spider] at ([shift={(-0.9,0)}]f2) {};
  \node[lc={blue}{0}, spider] at ([shift={(+0.9,0)}]f2) {};
  \node[] (a1) at (+1.7,0) {$\leftarrow$};
  \node[cospanblue, minimum height=3cm] (b2) at (2.5,0) {};
  \node[lc={blue}{2}, spider] (s3) at (+2.5,+0.750) {};
  \node[lc={blue}{1}, spider] (s4) at (+2.5,-0.000) {};
  \node[lc={blue}{0}, spider] (s6) at (+2.5,-1.000) {};
\end{tikzpicture}

%% file: tikz/categories/fig_23.tex
\begin{tikzpicture}[inline]
  \node[cospanblue, minimum height=3cm] (b1) at (-2.5,0) {};
  \node[lc={blue}{2}, spider] (s1) at (-2.5,+0.750) {};
  \node[lc={blue}{1}, spider] (s2) at (-2.5,-0.000) {};
  \node[lc={blue}{0}, spider] (s5) at (-2.5,-1.000) {};
  \node[] (a1) at (-1.7,0) {$\rightarrow$};
  \node[cospangrey, minimum height=3cm, minimum width=4cm] (m1) at (0.75,0) {};
  \node[draw, rounded corners=5pt, minimum size=0.75cm, fill=white] (f1) at (1,0) {$h$};
  \node[spider] (p2) at ([shift={(-1.4,0.375)}]f1) {};
  \draw[bluewire] (p2) to [out=110, in=0] ++(-0.5,+0.4);
  \draw[bluewire] (p2) to [out=-110, in=0] ++(-0.5,-0.4);
  \node[lc={blue}{2}, spider] at ([shift={(-0.5,0.4)}]p2) {};
  \node[lc={blue}{1}, spider] at ([shift={(-0.5,-0.4)}]p2) {};
  \node[lc={blue}{0}, spider] at ([shift={(-1.9,-1)}]f1) {};
  \node[spider] (p3) at ([shift={(0.9,-0)}]f1) {};
  \draw[bluewire] (f1.160) to [out=135, in=0] (p2);
  \draw[bluewire] (f1.200) to [out=225, in=0] ++(-1.5,-0.85);
  \node[lc={blue}{3}, spider] (p4) at ([shift={(-0.9,0.37)}]f1) {};
  \draw[redwire] (f1.25) to [out=0, in=135] (p3);
  \draw[redwire] (f1.-25) to [out=0, in=-135] (p3);
  \draw[redwire] (p3) to ++(0.5,0);
  \node[lc={red}{0}, spider] at ([shift={(0.5,0)}]p3) {};
  \node[] (a2) at (3.125,0) {$\leftarrow$};
  \node[cospanblue, minimum height=3cm] (b2) at (3.875,0) {};
  \node[lc={red}{0}, spider] (s4) at (3.875,0) {};
\end{tikzpicture}

%% file: tikz/categories/fig_24.tex
\begin{tikzpicture}[inline]
  \node[cospanblue, minimum height=3cm] (b1) at (-2.5,0) {};
  \node[lc={black}{2}, spider] (s1) at (-2.5,+0.750) {};
  \node[lc={black}{1}, spider] (s2) at (-2.5,-0.000) {};
  \node[lc={black}{0}, spider] (s5) at (-2.5,-1.000) {};
  \node[] (a1) at (-1.7,0) {$\rightarrow$};
  \node[cospangrey, minimum height=3cm, minimum width=6cm] (m1) at (1.75,0) {};
  \node[draw, rounded corners=5pt, minimum size=0.75cm, fill=white] (f1) at (0,0.375) {$f$};
  \draw[wire] (f1.160) to [out=135, in=0] ++(-0.5,+0.25);
  \draw[wire] (f1.200) to [out=225, in=0] ++(-0.5,-0.25);
  \draw[bluewire] (f1.20) to [out=45, in=180] ++(+0.5,+0.25);
  \draw[bluewire] (f1.340) to [out=315, in=180] ++(+0.5,-0.25);
  \node[lc={black}{2}, spider] at ([shift={(-0.9,0.4)}]f1) {};
  \node[lc={black}{1}, spider] at ([shift={(-0.9,-0.4)}]f1) {};
  \node[lc={blue}{2}, spider] (p4) at ([shift={(0.9,0.4)}]f1) {};
  \node[lc={blue}{1}, spider] (p5) at ([shift={(0.9,-0.4)}]f1) {};
  \node[draw, rounded corners=5pt, minimum size=0.75cm, fill=white] (f2) at (0,-1.000) {$k$};
  \draw[wire] (f2.180) to [out=180, in=0] ++(-0.5,0);
  \draw[bluewire] (f2.0) to [out=0, in=180] ++(+0.5,0);
  \node[lc={black}{0}, spider] at ([shift={(-0.9,0)}]f2) {};
  \node[draw, rounded corners=5pt, minimum size=0.75cm, fill=white] (f3) at (2.95,0) {$h$};
  \node[spider] (p3) at ([shift={(-1.5,0.375)}]f3) {};
  \draw[bluewire] (f3.160) to [out=135, in=0] (p3);
  \draw[bluewire] (f3.200) to [out=225, in=0] ++(-1.7,-0.85);
  \draw[bluewire] (p3) to [out=110, in=0] ++(-0.5,+0.4);
  \draw[bluewire] (p3) to [out=-110, in=0] ++(-0.5,-0.4);
  \node[lc={blue}{3}, spider] (p4) at ([shift={(-0.9,0.37)}]f3) {};
  \node[lc={blue}{0}, spider] (p3) at ([shift={(+0.9,0)}]f2) {};
  \node[spider] (p6) at ([shift={(0.9,0)}]f3) {};
  \draw[redwire] (f3.25) to [out=0, in=135] (p6);
  \draw[redwire] (f3.-25) to [out=0, in=-135] (p6);
  \draw[redwire] (p6) to ([shift={(0.5,0)}]p6);
  \node[lc={red}{0}, spider] at ([shift={(0.5,0)}]p6) {};
  \node[] (a1) at (5,0) {$\leftarrow$};
  \node[cospanblue, minimum height=3cm] (b2) at (5.75,0) {};
  \node[lc={red}{0}, spider] (s4) at (5.75,0) {};
\end{tikzpicture}

%% file: tikz/categories/fig_25.tex
\begin{tikzpicture}[inline]
  \node[morphismblue, minimum height=1.5cm] (m1a) {$f$};
  \node[morphismblue] (m1g) at (0,-1.5) {$k$};
  \node[spider] (s1) at (1,0) {};
  \node[morphismblue, minimum height=1.5cm] (m1b) at (2,-0.375){$h$};
  \node[spider] (s2) at (3,-0.375) {};
  \draw[wire] ($(m1a.south west)!0.75!(m1a.north west)$) -- ++(-0.5,0) node[left, font=\scriptsize] {};
  \draw[wire] ($(m1a.south west)!0.25!(m1a.north west)$) -- ++(-0.5,0) node[left, font=\scriptsize] {};
  \draw[bluewire] ($(m1a.south east)!0.75!(m1a.north east)$) to[out=0, in=100] (s1);
  \draw[bluewire] ($(m1a.south east)!0.25!(m1a.north east)$) to[out=0, in=-100] (s1);
  \draw[bluewire] (s1) to[out=0,in=180] ($(m1b.south west)!0.75!(m1b.north west)$);
  \draw[wire] (m1g.west) -- ++(-0.5,0) node[left] {};
  \draw[bluewire] (m1g.east) to[out=0, in=180] ($(m1b.south west)!0.25!(m1b.north west)$);
  \draw[redwire] ($(m1b.south east)!0.75!(m1b.north east)$) to[out=0,in=100] (s2);
  \draw[redwire] ($(m1b.south east)!0.25!(m1b.north east)$) to[out=0,in=-100] (s2);
  \draw[redwire] (s2) -- ++(0.5,0) node[left,font=\scriptsize] {};
\end{tikzpicture}

%% file: tikz/categories/fig_26.tex
\begin{tikzpicture}[inline]
\node[morphismblue, minimum height=1.5cm] (m1a) {$f$};
\node[morphismblue] (m1g) at (0,-1.5) {$k$};
\node[morphismblue, minimum height=1.5cm] (m1b) at (2,-0.450){$h$};
\draw[wire] ($(m1a.south west)!0.75!(m1a.north west)$) -- ++(-0.5,0) node[left, font=\scriptsize] {};
\draw[wire] ($(m1a.south west)!0.25!(m1a.north west)$) -- ++(-0.5,0) node[left, font=\scriptsize] {};
\draw[wire] ($(m1a.south east)!0.75!(m1a.north east)$) to[out=0,in=180] ($(m1b.south west)!0.75!(m1b.north west)$);
\draw[wire] ($(m1a.south east)!0.25!(m1a.north east)$) to[out=0,in=180] ($(m1b.south west)!0.55!(m1b.north west)$);
\draw[wire] (m1g.west) -- ++(-0.5,0) node[left] {};
\draw[wire] (m1g.east) to[out=0, in=180] ($(m1b.south west)!0.25!(m1b.north west)$);
\draw[wire] ($(m1b.south east)!0.75!(m1b.north east)$) -- ++(+0.5,0) node[left] {};
\draw[wire] ($(m1b.south east)!0.25!(m1b.north east)$) -- ++(+0.5,0) node[left] {};
\end{tikzpicture}

%% file: tikz/categories/fig_27.tex
\begin{tikzpicture}[inline]
\node[morphismblue, minimum height=1.5cm] (m1a) {$f$};
\node[morphismblue] (m1g) at (0,-1.5) {$k$};
\node[morphismblue, minimum height=1.5cm] (m1b) at (2,-0.450){$h$};
\draw[wire] ($(m1a.south west)!0.75!(m1a.north west)$) -- ++(-0.5,0) node[left, font=\scriptsize] {};
\draw[wire] ($(m1a.south west)!0.25!(m1a.north west)$) -- ++(-0.5,0) node[left, font=\scriptsize] {};
\draw[wire] ($(m1a.south east)!0.75!(m1a.north east)$) to[out=0,in=180] ($(m1b.south west)!0.75!(m1b.north west)$);
\draw[wire] ($(m1a.south east)!0.25!(m1a.north east)$) to[out=0,in=180] ($(m1b.south west)!0.55!(m1b.north west)$);
\draw[wire] (m1g.west) -- ++(-0.5,0) node[left] {};
\draw[wire] (m1g.east) -- ++(+2.5,0) node[left] {};
\draw[wire] ($(m1b.south east)!0.75!(m1b.north east)$) -- ++(+0.5,0) node[left] {};
\draw[wire] ($(m1b.south east)!0.25!(m1b.north east)$) -- ++(+0.5,0) node[left] {};
\end{tikzpicture}

%% file: tikz/categories/fig_28.tex
\begin{tikzpicture}[inline]
  \node[morphismblue, minimum height=0.75cm] (1) at (-2,0.5) {$f$};
  \node[morphismblue, minimum height=0.75cm] (2) at (-2,-0.5) {$k$};
  \node[morphismblue, minimum height=1.5cm] (3) at (0,0) {$h$};
  \node[morphismblue, minimum height=0.75cm] (4) at (2,0.5) {$f$};
  \node[morphismblue, minimum height=0.75cm] (5) at (2,-0.5) {$k$};

  \draw[wire] ($(1.south west)!0.50!(1.north west)$) -- ++(-0.5,0) node[left, font=\scriptsize] {};
  \draw[wire] ($(2.south west)!0.50!(2.north west)$) -- ++(-0.5,0) node[left, font=\scriptsize] {};

  \draw[wire] ($(1.south east)!0.50!(1.north east)$) to[out=0,in=180] ($(3.south west)!0.75!(3.north west)$);
  \draw[wire] ($(2.south east)!0.50!(2.north east)$) to[out=0,in=180] ($(3.south west)!0.25!(3.north west)$);

  \draw[wire] ($(3.south east)!0.75!(3.north east)$) to[out=0,in=180] ($(4.south west)!0.50!(4.north west)$);
  \draw[wire] ($(3.south east)!0.25!(3.north east)$) to[out=0,in=180] ($(5.south west)!0.50!(5.north west)$);

  \draw[wire] ($(4.south east)!0.50!(4.north east)$) -- ++(+0.5,0) node[left, font=\scriptsize] {};
  \draw[wire] ($(5.south east)!0.50!(5.north east)$) -- ++(+0.5,0) node[left, font=\scriptsize] {};

\end{tikzpicture}

%% file: tikz/categories/fig_29.tex
\begin{tikzpicture}[]

  \node[] (s1) at (0,+0.5) {$f_1$};
  \node[] (s2) at (0,-0.5) {$k_1$};
  \node[] (s3) at (1,+0) {$h$};
  \node[] (s4) at (2,+0.5) {$f_2$};
  \node[] (s5) at (2,-0.5) {$k_2$};

  \draw[->] (s1) to (s3);
  \draw[->] (s2) to (s3);
  \draw[->] (s3) to (s4);
  \draw[->] (s3) to (s5);

\end{tikzpicture}

%% file: sections/02b-general-framework.tex
\section{One Construction for Every Notation}\label{sec:general-framework}

One construction turns a process model into a signature of generator cospans, and it is the
same for every notation in Sections~\ref{sec:petri-nets}--\ref{sec:bpmn}. Pick an activity,
follow the edges around it through the mediators, and stop at the first activities met on each
side, recording the object types and edge weights seen on the way. Each way of joining what the
activity receives to what it emits is one generator. Composing generators freely gives a
symmetric monoidal category, each of its morphisms is a partial order of activity occurrences,
and its traces are the orderings of those occurrences. For any choice of morphisms made from the
signature alone, models with equal signatures have equal languages. A badly specified model appears as composites that cannot be satisfied.

\subsection{The AND/OR graph}
\label{subsec:firing-graphs}

Every notation is first mapped to one kind of graph. Its nodes are activities, which a trace
records, and mediators, which it does not. A mediator sends flow along all its
branches at once (AND) or along exactly one of them (XOR). Edges carry the typing and weighting
that the notation records.

\begin{definition}[AND/OR graph]
  \label{def:lm-graph}
  An AND/OR graph over a finite set $\mathcal{O}$ of object types is a tuple
  $G=(\mathcal{A},R,E,\ell,\nu,\mathrm{type},\mathrm{wt},\rho)$ in which $(\mathcal{A}\sqcup R,E)$
  is a finite directed graph whose nodes are the activities $\mathcal{A}$ and the mediators $R$,
  $\ell:\mathcal{A}\to\mathcal{L}$ labels the activities, $\nu:R\to\{\mathrm{AND},\mathrm{XOR}\}$
  gives each mediator its kind, $\mathrm{type}:E\to\mathcal{O}$ and
  $\mathrm{wt}:E\to\mathbb{N}_{\ge 1}\cup\{\mathrm{var}\}$ give each edge a type and a weight, and
  $\rho(a)$ is a set of linear equalities and inequalities among the weights of the edges at the
  activity $a$. The weight $\mathrm{var}$ marks an edge that carries zero or more objects, and
  $\rho(a)$ is empty unless the notation relates the weights. Silent transitions are mediators,
  and a notation with no types takes $\mathcal{O}=\{\ast\}$.
\end{definition}

Object-centric Petri nets, object-centric causal nets, process trees and BPMN are instances.
The work specific to a notation is the correctness of its map to AND/OR graphs, which each
notation's section establishes. Everything below is the same for all of them.

\subsection{Reach families, contexts and generators}
\label{subsec:generators}

\begin{definition}[Reach families and contexts]
  \label{def:contexts-general}
  Let $a\in\mathcal{A}$. For a node $x$ reached from $a$ along a path $\pi$ of out-edges, let
  $w_\pi\subseteq\mathcal{O}$ be the set of types on the edges of $\pi$ and let $N(x)$ be the
  out-neighbours $y$ of $x$ with $R_a(y,\pi y)\ne\emptyset$. The reached family $R_a(x,\pi)$, a set of
  finite sets of typed wires, is
  \begin{gather*}
    R_a(x,\pi)=\twocolbreak
    \begin{cases}
      \{\{(a,w_\pi,x)\}\} & x\in\mathcal{A},\\
      \emptyset & x\text{ earlier on }\pi,\\
      \{\{(a,w_\pi,\top)\}\} & \text{no out-edges},\\
      \textstyle\bigcup_{y\in N(x)}R_a(y,\pi y) & \nu(x)=\mathrm{XOR},\\
      \{\textstyle\bigcup_{y\in N(x)}S_y : S_y\in R_a(y,\pi y)\} & \nu(x)=\mathrm{AND},
    \end{cases}
  \end{gather*}
  taking the cases in order, so the last four apply to mediators, with the AND case empty when
  $N(x)$ is. The forward family $\mathcal{F}_a$ joins the out-edges of
  $a$ as an AND mediator does, and $\mathcal{F}_a=\{\emptyset\}$ when no out-edge reaches anything.
  The backward family $\mathcal{B}_a$ is the same construction along in-edges, with typed wires
  $(x,w_\pi,a)$ and $(\bot,w_\pi,a)$. The contexts of $a$ are the pairs
  $\mathrm{Contexts}(a)=\mathcal{B}_a\times\mathcal{F}_a$.
\end{definition}

The walk visits each node at most once per path, so $\mathcal{F}_a$ and $\mathcal{B}_a$ are
finite on every finite graph, cycles included. Activities are identified by their labels, so
typed wires are compared by label.

\begin{definition}[Generator cospan]
  \label{def:generator-cospan-general}
  For $a\in\mathcal{A}$, each context $c=(P,S)\in\mathrm{Contexts}(a)$ gives the generator
  $g_{a,c}=(P\hookrightarrow P+S\hookleftarrow S,\ d_{a,c})$ with one hyperedge $e_a$, labelled
  $\lambda(e_a)=\ell(a)$, whose source is $P$ and whose target is $S$. Its decoration $d_{a,c}$
  also carries the constraint system $\Lambda_{a,c}$ of
  Definition~\ref{def:multiplicity-decoration}. Every context gives a generator and none is
  discarded. The signature of $G$ is
  $\Sigma=\{g_{a,c} : a\in\mathcal{A},\ c\in\mathrm{Contexts}(a)\}$.
\end{definition}

A typed wire is identified by its triple $(a,w,b)$. The set of edge weights met on its paths is
data of the typed wire and plays no part in its identity, and each generator's constraints use the weight of
the edge at its own end. The counts of a generator are fixed by its own legs only. Nothing
outside the generator ties its input counts to its output counts, so an activity may create or
consume objects of a type that appears on one side only.

\begin{definition}[Constraint system]
  \label{def:multiplicity-decoration}
  Each leg $p\in P\cup S$ of $g_{a,c}$ has a count $n_p\in\mathbb{N}$, the number of objects
  on $p$ in one occurrence of $a$. If the edge at $a$'s end of the path has weight $k$ then
  $n_p=k$, and if that edge is variable then $n_p\ge 0$. The equalities and inequalities
  $\rho(a)$ of Definition~\ref{def:lm-graph} apply to the counts of the legs whose edge at $a$'s
  end they name, as the shared keys and cardinality intervals of object-centric causal nets do. The constraint system $\Lambda_{a,c}$ is the conjunction of these
  constraints. Two generators are equal when they have the same label, the same typed wires and
  constraint systems with the same solutions.
\end{definition}

\subsection{The free category and its traces}
\label{subsec:free-category}

A leg with count $n$ is the $n$-fold tensor power of its typed wire, with $n$ ports. A port joins
an output of one occurrence to an input of another on the same typed wire, and nothing else
composes.

\begin{definition}[Free category of a signature]
  \label{def:free-category}
  For $g=g_{a,c}\in\Sigma$ and a count $k_p\in\mathbb{N}$ for each leg $p$, the count assignment of
  $g$ with these counts is the box
  \[
    g^{k}:\ \textstyle\bigotimes_{p\in P}p^{\otimes k_p}\ \to\ \bigotimes_{q\in S}q^{\otimes k_q}
  \]
  labelled $\ell(a)$, with $k_p$ ports on each leg $p$ in a fixed order of the legs.
  $\mathbf{F}(\Sigma)$ is the free symmetric monoidal category whose objects are the finite lists of typed wires and
  whose generating morphisms are all count assignments. A morphism $f$ is satisfiable when every count
  assignment $g^k$ in it satisfies $\Lambda_{a,c}$ with $n_p=k_p$.
\end{definition}

A morphism of $\mathbf{F}(\Sigma)$ is a string diagram built from count assignments, identities and
symmetries by $;$ and $\otimes$. A count assignment is an activity morphism of
Section~\ref{sec:hypergraphs} whose multiplicities on each typed wire are the counts $k_p$, so each
generator of $\mathbf{H}(\Sigma)$ gives one count assignment per choice of counts. $\mathbf{F}(\Sigma)$ is
the Frobenius-free part of $\mathbf{H}(\Sigma)$, the diagrams in which no port is copied, merged,
created or discarded.

A morphism belongs to $\mathbf{F}(\Sigma)$ whatever its boundary. An open input port stands
for an object that enters from outside, and therefore every initialisation of the model is some
morphism. A loop in the model appears as repeated composition, since a diagram has no cycles.

\begin{definition}[Trace semantics of a morphism]
  \label{def:trace-map}
  The occurrences of a morphism $f$ of $\mathbf{F}(\Sigma)$ are its count assignments. An occurrence $e$ precedes an occurrence
  $e'$ when a port of $f$ joins an output of $e$ to an input of $e'$, and the causal partial
  order $D_f$ is the reflexive and transitive closure of this relation. If $f$ is satisfiable then
  \begin{gather*}
    \Gamma(f)=\{\ell(e_1)\cdots\ell(e_m) :\twocolbreak
      e_1,\dots,e_m \text{ a linear extension of } D_f\},
  \end{gather*}
  where a silent label contributes the empty word, and otherwise $\Gamma(f)=\emptyset$.
\end{definition}

The relation $D_f$ is a partial order because diagrams are acyclic. A composite that contains
an unsatisfiable count assignment is itself unsatisfiable, and $\Gamma$ is empty on it.

\begin{exmp}[An AND split and join]
  \label{ex:and-fanout}
  Let $a$ lead through an AND mediator to $b$ and $c$, and let both lead through a second AND mediator to
  $d$, with one type and weight $1$ throughout. Write $(a,b)$ for $(a,\{\ast\},b)$, and write an
  activity's name for the count assignment of its generator with one port on each leg, and $I$ for the
  monoidal unit, the empty list of typed wires. The walk gives
  $\mathcal{F}_a=\{\{(a,b),(a,c)\}\}$ and $\mathcal{B}_d=\{\{(b,d),(c,d)\}\}$, and the count assignments include
  \[
    a:\ I\to(a,b)\otimes(a,c),\qquad b:\ (a,b)\to(b,d).
  \]
  With $c:(a,c)\to(c,d)$ and $d:(b,d)\otimes(c,d)\to I$, the composite
  $f=a\,;(b\otimes c);d$ has the order $a\le b\le d$ and $a\le c\le d$. Its two linear extensions
  give $\Gamma(f)=\{abcd,\ acbd\}$. The open composite $a\,;(b\otimes\mathrm{id}_{(a,c)})$ also
  belongs to $\mathbf{F}(\Sigma)$ and gives $\{ab\}$. Had the edge leaving $a$ carried weight $2$,
  the count assignment of $a$ with one port on each leg would be unsatisfiable and $\Gamma(f)=\emptyset$
  for this $f$, while the count assignment of $a$ with two ports on each leg composes.
\end{exmp}

\begin{lemma}[$\Gamma$ is well defined]
  \label{lem:gamma-well-defined}
  Equal morphisms of $\mathbf{F}(\Sigma)$ have equal images under $\Gamma$.
\end{lemma}
\begin{proof}
  A morphism is a class of terms under the symmetric monoidal equations, and $\Gamma$ depends on a
  term only through its occurrences, their labels and which output port reaches which input. Identities
  and symmetries have no occurrences and only relay ports. The associativity and unit laws regroup a
  term without adding occurrences or redirecting ports. Interchange, $(f\otimes g);(f'\otimes g')=
  (f;f')\otimes(g;g')$, joins the same outputs to the same inputs on both sides, and naturality
  of the symmetry moves a crossing past a box with every port ending where it did. The law
  $\sigma;\sigma=\mathrm{id}$ and the hexagons relay each port unchanged. Each equation therefore
  preserves the occurrences, satisfiability and the wiring, which records which output reaches
  which input, the boundary included. The wiring of $f;g$ and of $f\otimes g$ is fixed by
  those of $f$ and $g$, so the invariance extends to every term containing an equated part, and
  $D_f$ is determined by the wiring.
\end{proof}

\subsection{The model invariant}
\label{subsec:theorem}
\label{subsec:minimal-signatures}

A selection $\mathcal{S}$ assigns to each signature $\Sigma$ a set $\mathcal{S}(\Sigma)$ of morphisms of
$\mathbf{F}(\Sigma)$ by a rule that uses $\Sigma$ alone. Its language is
$L_{\mathcal{S}}(\Sigma)=\bigcup_{f\in \mathcal{S}(\Sigma)}\Gamma(f)$.

\begin{theorem}[Model invariant]
  \label{thm:canonical-presentation}
  For every selection $\mathcal{S}$, if $\Sigma_1=\Sigma_2$ then $L_{\mathcal{S}}(\Sigma_1)=L_{\mathcal{S}}(\Sigma_2)$.
\end{theorem}
\begin{proof}
  Equal signatures have the same typed wires and the same count assignments, with the same labels and the same
  satisfiable counts, so $\mathbf{F}(\Sigma_1)=\mathbf{F}(\Sigma_2)$. A rule that uses the
  signature alone selects the same morphisms from both, and by
  Lemma~\ref{lem:gamma-well-defined} the unions defining $L_{\mathcal{S}}$ agree.
\end{proof}

The theorem holds for every selection, and how well the language separates models depends on the
selection. The default selection takes the connected process diagrams of Section~\ref{sec:hypergraphs},
the morphisms with any boundary that have at least one count assignment and whose count assignments are
joined into one piece by their ports. It covers every
initialisation, including a pure cycle $a\to b\to c\to a$ started at whichever activity receives the first object, and it
keeps the order of occurrences, whereas over all of $\mathbf{F}(\Sigma)$ count assignments placed side by side
give every ordering of the labels. The converse of the theorem fails, since different signatures
can give equal languages.

Signatures can also be compared by inclusion. Write $\Sigma_1\subseteq\Sigma_2$ when every generator of
$\Sigma_1$ is equal to a generator of $\Sigma_2$ in the sense of
Definition~\ref{def:multiplicity-decoration}, with the same label, the same typed wires and
constraint systems with the same solutions.

\begin{proposition}[Inclusion of free categories]
  \label{prop:inclusion-functor}
  If $\Sigma_1\subseteq\Sigma_2$ then there is a strict symmetric monoidal functor
  $\mathrm{inc}\colon\mathbf{F}(\Sigma_1)\to\mathbf{F}(\Sigma_2)$ that is the identity on typed wires
  and on count assignments, and $\Gamma(\mathrm{inc}\,f)=\Gamma(f)$ for every morphism $f$.
\end{proposition}
\begin{proof}
  Every typed wire of $\Sigma_1$ is a typed wire of $\Sigma_2$, so every object of
  $\mathbf{F}(\Sigma_1)$ is an object of $\mathbf{F}(\Sigma_2)$. A count assignment $g^{k}$ of
  $\Sigma_1$ has the label, legs and counts of the count assignment $g^{k}$ of the equal generator of
  $\Sigma_2$, and it satisfies one constraint system exactly when it satisfies the other, since the
  two have the same solutions. Sending each generating morphism of $\mathbf{F}(\Sigma_1)$ to the same
  generating morphism of $\mathbf{F}(\Sigma_2)$ extends to a strict symmetric monoidal functor by
  the universal property of the free symmetric monoidal category. The functor sends a term built by
  $;$ and $\otimes$ from count assignments, identities and symmetries to the same term, so
  $\mathrm{inc}\,f$ has the occurrences of $f$ with the same labels and the same satisfiability, and
  each of its ports joins the output and input it joins in $f$. Hence $D_{\mathrm{inc}\,f}=D_f$ and
  $\Gamma(\mathrm{inc}\,f)=\Gamma(f)$. By Lemma~\ref{lem:gamma-well-defined} the equation holds for
  every term representing the morphism.
\end{proof}

A selection $\mathcal{S}$ is monotone when $\Sigma_1\subseteq\Sigma_2$ implies
$\mathrm{inc}(\mathcal{S}(\Sigma_1))\subseteq\mathcal{S}(\Sigma_2)$.

\begin{corollary}[Signature inclusion]
  \label{cor:signature-inclusion}
  For every monotone selection $\mathcal{S}$, if $\Sigma_1\subseteq\Sigma_2$ then
  $L_{\mathcal{S}}(\Sigma_1)\subseteq L_{\mathcal{S}}(\Sigma_2)$.
\end{corollary}
\begin{proof}
  A word of $L_{\mathcal{S}}(\Sigma_1)$ lies in $\Gamma(f)$ for some $f\in\mathcal{S}(\Sigma_1)$. By
  Proposition~\ref{prop:inclusion-functor} it lies in $\Gamma(\mathrm{inc}\,f)$, and
  $\mathrm{inc}\,f\in\mathcal{S}(\Sigma_2)$ because $\mathcal{S}$ is monotone.
\end{proof}

The connected process diagrams form a monotone selection. Having at least one count assignment
and having all count assignments joined into one piece by ports are properties of the occurrences
and the wiring, and $\mathrm{inc}$ preserves both. The corollary needs monotonicity, whereas
Theorem~\ref{thm:canonical-presentation} holds for every selection. The rule that selects the
morphisms in which every generator of $\Sigma$ occurs uses $\Sigma$ alone, yet a larger signature
can lose morphisms under it. The corollary gives no strictness, since a proper inclusion of
signatures may still give equal languages.

\subsection{Analysis of composites}
\label{subsec:analysis}

The analysis of a particular model selects and diagnoses composites of $\mathbf{F}(\Sigma)$, and
the \texttt{proc-posets} implementation performs it without changing $\Sigma$. A start generator
provides objects for one way the process can begin, a set of first activities that receive
objects together, and an end generator takes objects at the matching finish. These generators
form $\Sigma_\gamma$, every report states them as an assumption, and a selection that uses them
needs $\Sigma_\gamma$ equal as well as $\Sigma$ for Theorem~\ref{thm:canonical-presentation} to
apply. It is monotone when the start and end generators of the smaller signature are among those
of the larger, since a closed run stays closed under $\mathrm{inc}$, and
Corollary~\ref{cor:signature-inclusion} then applies. A closed run joins a start generator to an end generator with no open port. A livelock
family is a sequence of composites that another pass of a loop can always extend and that never
closes. A base run is a closed run from which no loop pass can be removed. Removal is
read on labels, so a closed run is a base run when no shorter closed run with the same start and
end generators differs from it, as label counts, by a sum of loop passes. Every closed run is then a
base run with loop passes inserted, and the closed runs fall into classes indexed by how many
passes are inserted where (Section~\ref{sec:validation}). Loops are classed by how one pass changes the object counts. A bounded loop leaves them
unchanged, a growing but closable loop adds objects that an end generator still absorbs, an
explosive loop produces objects that nothing can absorb, and a livelock loop is never left. Run
coverage lists the generators that lie on no run and the reason for each. When counts cannot
balance from the stated start, the report gives the least extra supply that would balance them.
Open typed wires that nothing joins, typed wires whose paths carry more than one type and models whose
activities form separate components are reported as well.

The \texttt{proc-posets} test suite checks the extraction on 26 hand-derived models against an
independent token game. On its bench, equal signatures never gave different languages under any
selection tried, and the connected process diagrams gave the fewest pairs of models with different
signatures and equal languages.

%% file: sections/03-petri_nets.tex
\subsection{Petri nets}\label{sec:petri-nets}

Petri nets are a standard formalism for concurrency and process mining. We treat both classical Petri nets and object-centric Petri nets in a single typed framework. The common structural core is a finite place-transition incidence structure. The difference lies in the interpretation of tokens and markings. In the classical case, tokens are indistinguishable and markings count token multiplicity. In the object-centric case, places carry types and markings record finite collections of object identities of the appropriate type. In both settings, transitions will be interpreted as generator cospans of a cospan-algebra signature, and the flow between them as typed wires of its free symmetric monoidal category.

\begin{definition}[Object-centric Petri net~\cite{vandettenDiscoveringCompactLive2024}]
  An object-centric Petri net (OCPN) is a tuple $N=(P,T,F,\mathcal{O},\mathrm{type},\ell,M_0)$ of places $P$, transitions $T$, a flow relation $F\subseteq(P\times T)\cup(T\times P)$, a finite set of object types $\mathcal{O}$, a typing $\mathrm{type}:P\to\mathcal{O}$, a labelling $\ell:T\to\mathcal{L}\cup\{\tau\}$ over an activity alphabet $\mathcal{L}$ with silent symbol $\tau$, and an initial marking $M_0\in\mathbb{N}^{P}$. Write $T_\tau=\ell^{-1}(\tau)$ for the silent transitions and $\mathrm{src}(t),\mathrm{tgt}(t)$ for the input and output places of a transition $t$. Several transitions may share a label, the firing modes of one activity (Remark~\ref{rem:mode-refinement}).
\end{definition}

\begin{definition}[Petri net]
  A Petri net is an object-centric Petri net where $\mathrm{type}(p)=*$ for all $p\in{}P$.
\end{definition}

Both classes share this incidence structure and differ only in how tokens are read.
The classical case has indistinguishable counts, the object-centric case typed object
identities. One construction serves both. Before reading object types off the
net we record the structural condition that keeps them unambiguous across silent
transitions.

\begin{definition}[Type-consistent OCPN]\label{def:type-consistent}
  Let $(P\cup T_\tau,\,F)$ be the silent-flow subgraph of an OCPN,
  obtained by deleting the observable transitions $T\setminus T_\tau$ together
  with their incident arcs. The net is type-consistent when, along every
  directed path of the silent-flow subgraph between two places, the
  concretely-typed places carry one common object type. Untyped routing places
  and the artificial source and sink are exempt and inherit the type of the
  path on which they lie. Equivalently, no directed chain of silent transitions
  joins two places of different concrete object type. Exclusive routing is
  unaffected. A single untyped place may still fan out to alternatives of
  different type, since these lie on distinct directed paths.
\end{definition}

We assume throughout that object-centric nets are type-consistent. The
condition holds automatically for nets produced by object-centric discovery,
where each object type flows in its own coloured sub-net, and is the Petri-net
form of the prohibition, due to \citet{lissObjectCentricCausalNets2025} and reflected in the causal-net
boundary convention of Section~\ref{sec:causal-nets}, against a transition
that converts one object type into another. Remark~\ref{rem:ocpn-silent-typing}
records how the construction relies on it.

This construction instantiates the general framework of Section~\ref{sec:general-framework}.
Its content is the way we read a net as an AND/OR graph $(V,E)$. We build the mediators the
framework needs instead of reusing $P\cup T$ verbatim. Each place becomes an XOR mediator, since a
token in it is left by one of its producing transitions and taken by one of its consumers, a
choice. Each transition's firing becomes an AND synchronisation, consuming its whole pre-set and
yielding its whole post-set at once, so its input places are joined by an AND mediator and its
output places split from one. An observable transition is in addition the activity it labels,
sitting between its pre-set and post-set AND mediators. The activity does no combining of its own,
exactly as in Definition~\ref{def:lm-graph}, with all concurrency in the AND mediators
and all choice in the XOR places. A silent transition contributes only its transparent AND mediator
and no generator. These choices determine the instance map.
\medskip
\begin{center}
  \small\begin{tabularx}{\linewidth}{@{}lX@{}}
    \hline
    General framework & OCPN \\
    \hline
    Activities $\mathcal{A}$ & observable transitions $T\setminus T_\tau$ \\
    AND mediators & each transition's pre-set join and post-set split; silent transitions $T_\tau$ (transparent, no generator) \\
    XOR mediators & places $P$ \\
    Edges $E$ & the flow $F$, threaded through these mediators \\
    \hline
  \end{tabularx}
\end{center}
\medskip
The walk of Definition~\ref{def:contexts-general} runs on this graph unchanged. From the pre-set
AND mediator of an observable transition it passes through each input place, an XOR mediator that
keeps the reached sets of its producers separate, and through any silent transition, and it stops
at the first observable transitions. The post-set is read in the same way. The backward and
forward families hold one reached set per joint choice of observable neighbour, each endpoint read
through the labelling $\ell$, so transitions sharing an activity label contribute the same labelled
typed wire and are identified at composition. Each transition--context pair yields one generator
cospan (Definition~\ref{def:generator-cospan-general}), and no context is discarded.

\begin{theorem}[Canonical cospan-algebra presentation of typed Petri nets]
  \label{thm:petri-to-cospan-signature}
  Let \(N=(P,T,F,\mathcal{O},\mathrm{type},\ell,M_0)\) be an object-centric Petri net. Then \(N\) canonically determines a cospan-algebra signature \(\Sigma\), and hence the free symmetric monoidal category \(\mathbf{F}(\Sigma)\).
\end{theorem}

\begin{proof}
  The schema of Section~\ref{sec:notation-by-notation} applies at the instance map above,
  finiteness of $P$ and $T$ giving the conditions of Definition~\ref{def:lm-graph} and finite
  reach families (Definition~\ref{def:contexts-general}). The delta is typing across the silent fabric. The
  object-centric case runs over $\mathcal{O}$ with each edge typed by $\mathrm{type}$, and the
  type of each $\tau$-mediated thread is well defined precisely because $N$ is type-consistent
  (Definition~\ref{def:type-consistent}, Remark~\ref{rem:ocpn-silent-typing}).
\end{proof}

\begin{remark}[Relation to prior categorical presentations]\label{rem:smc-positioning}
  The result that Petri nets present free symmetric monoidal categories (SMCs) is due to
  \citet{meseguerPetriNetsAre1990}. The open-net cospan
  framework appears in \citet{baezOpenPetriNets2020}.
  Theorem~\ref{thm:petri-to-cospan-signature} refines these by constructing
  explicit context-indexed generator cospans. Each transition--context pair $(t,c)$
  yields a decorated cospan whose boundary data is read directly from the local
  incidence structure, equipping the SMC presentation with the generator-level
  information required by the cospan-algebra framework of
  Section~\ref{sec:hypergraphs}.
  In these terms the refinement makes the Meseguer and Montanari SMC
  presentation into an explicit generator-level transformation, positioned
  alongside the connector-algebra tradition in the related work of
  Section~\ref{sec:related-work}.
\end{remark}

\begin{corollary}[Untyped Petri nets as a special case]\label{cor:petri-untyped}
  Setting $\mathcal{O}=\{\ast\}$ and $\mathrm{type}(p)=\ast$ for all $p\in P$ in
  Theorem~\ref{thm:petri-to-cospan-signature} recovers the classical Petri-net
  presentation (Corollary~\ref{cor:untyped-general}), in which transitions induce generators,
  every typed wire carries the trivial type, and composites are the morphisms of the free
  symmetric monoidal category $\mathbf{F}(\Sigma)$ (Definition~\ref{def:free-category}), joined
  port to port along matching typed wires.
\end{corollary}

\begin{remark}[Silent transitions and type consistency]\label{rem:ocpn-silent-typing}
  The object-centric Petri-net literature attaches no typing of its own to a
  silent transition. None is needed here, because a silent transition reads its
  object types from the colours of its adjacent places,
  exactly as an observable transition does. Since silent transitions are
  transparent mediators (instance map above), a $\tau$-mediated typed wire
  $(x,w,y)$ carries the set $w$ of types on the thread it traces through
  the silent fabric, and type consistency (Definition~\ref{def:type-consistent})
  is precisely the condition that makes $w$ a single type. Under it the
  harvested boundary stays in bijection with the typed dependency arcs
  $D\subseteq T\times\mathcal{O}\times T$ of the object-centric causal net
  (Section~\ref{sec:causal-nets}). Transparent crossing introduces no new boundary
  types, and in particular no spurious duplication of a typed wire across the whole of
  $\mathcal{O}$.

  Dropping type consistency readmits exactly the wild behaviour it is there to
  tame. A type-converting silent transition would make a single thread enter as
  one concrete type and leave as another. The walk records both types in the set $w$
  of the typed wire through it, and the analysis reports a wire whose path carries more than one
  type (Section~\ref{subsec:analysis}). The construction therefore exposes the conversion
  instead of hiding it behind an artificial typed boundary. The
  only escape is the deliberately lossy projection to $\mathcal{O}=\{\ast\}$
  (Corollary~\ref{cor:petri-untyped}), which recovers the connectivity of $D$ but
  discards all typing. This is the Petri-net counterpart of the routing-only
  gateway discipline of Remark~\ref{rem:bpmn-typed-gateways}.
\end{remark}

\begin{remark}[Silent routing leaves the signature fixed]\label{rem:ab-place-vs-tau}
The sequence $a\,;\,b$ presented as one place $a\to p\to b$, and the same sequence with $p$ split by a silent transition into $a\to p_1$, $\tau$, $p_2\to b$, yield the identical signature. In both presentations activity $a$ has successor $b$ and activity $b$ has predecessor $a$. The silent transition passes transparently and contributes no generator, so the split place introduces no structural difference. The same argument absorbs an implicit place, a place duplicating a causal constraint that the net already carries. It contributes an endpoint the port set already holds, and ports are identified by producing activity, object type, and consuming activity, so the duplicate leaves the generator unchanged. This is the invariance the two Petri nets of Figure~\ref{fig:routing-collapse} exhibit. Note the scope. A redundant place carrying another object type is not absorbed, because it contributes a genuinely new typed port.
\end{remark}

\begin{remark}[Firing modes]\label{rem:mode-refinement}
  A single transition fires on its whole pre-set and yields its whole post-set, so it cannot itself
  choose among exclusive typed inputs or outputs. Such a choice is encoded the standard Petri-net
  way, by several transitions sharing one activity label, one per admissible firing mode, which the
  labelling $\ell$ re-identifies so that a downstream activity sees one labelled port rather than one
  per mode. This is the Petri-net analogue of the object-centric causal net's binding sets
  (Section~\ref{sec:causal-nets}). The same activity acquires one generator per mode, a larger
  signature that induces the same language (Section~\ref{sec:conversion-framework}).
\end{remark}

\begin{figure*}[tp]
  \centering
  \resizebox{\textwidth}{!}{\input{tikz/running_example/re_pn_example}\unskip}
  \caption{A small object-centric Petri net. An order (blue) is registered by $b$,
    optionally after a check $a$ that may be skipped through the silent transition
    $\tau$ (drawn filled). The registered order and an item (red) are then shipped
    together by $s$.}
  \label{fig:re-ocpn}
\end{figure*}
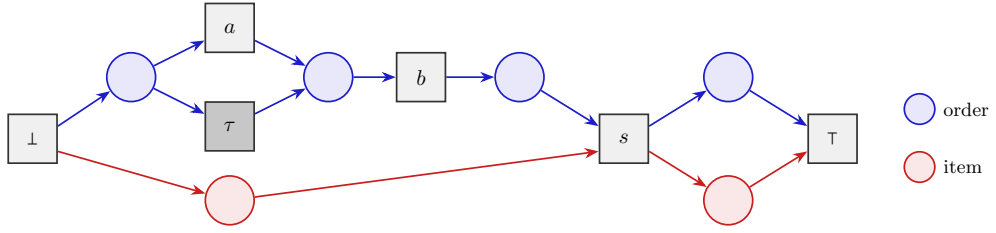

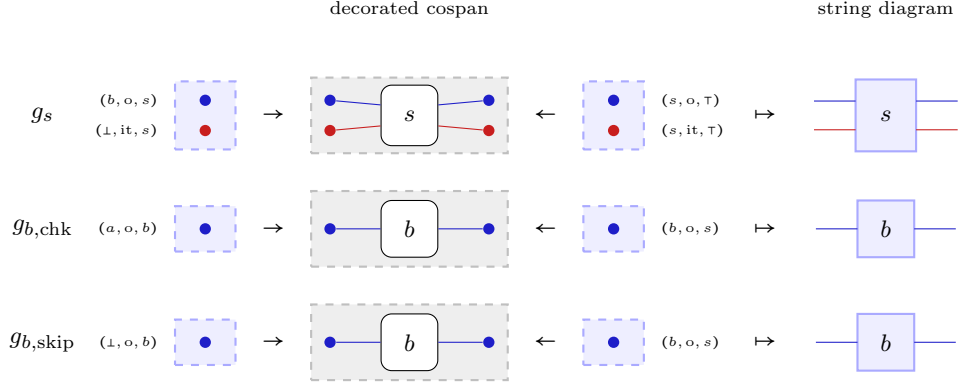
\begin{figure*}[tp]
  \centering
  \input{tikz/running_example/re_pn_cospans}\unskip
  \caption{The generators $g_s$, $g_{b,\mathrm{chk}}$ and $g_{b,\mathrm{skip}}$ of Fig.~\ref{fig:re-ocpn}, each as a decorated cospan
    (left, boundary $\to$ apex $\leftarrow$ boundary) and as the string diagram it denotes (right).
    $g_s$ is the object-centric sync. Two object types meet on the left boundary
    and leave on the right, joined by the apex hyperedge $s$. The silent $\tau$
    contributes no generator and gives $b$ two contexts, $g_{b,\mathrm{chk}}$
    (predecessor $a$) and $g_{b,\mathrm{skip}}$ (the open end $\bot$, traced
    through $\tau$), differing only in their left boundary port.}
  \label{fig:re-ocpn-cospans}
\end{figure*}

\begin{exmp}[A small OCPN and its generator cospans]
  \label{ex:running-pn-cospan}
  Take the object-centric Petri net of Fig.~\ref{fig:re-ocpn}, over the object
  types $\mathrm{order}$ and $\mathrm{item}$. Its observable transitions are
  $a,b,s$, with one silent transition $\tau$. The boundary nodes $\bot$ and $\top$ are not
  transitions of the net. They mark where objects enter and leave, and the walk ends there at the
  open ends of Definition~\ref{def:contexts-general}. Instantiating
  Theorem~\ref{thm:petri-to-cospan-signature} reads a generator cospan off
  each transition--context pair (Fig.~\ref{fig:re-ocpn-cospans}). Every edge has weight $1$,
  so each constraint system $\Lambda_{a,c}$ fixes $n_p=1$ on every leg and each wire carries a
  single object. Write $(x,\mathrm{o},y)$ for $(x,\{\mathrm{o}\},y)$.

  \emph{The check $a$.} The check has one observable predecessor route and one
  observable successor, so it has a single context and a single generator,
  \[
    g_a = \Bigl(\{(\bot,\mathrm{o},a)\}
      \xrightarrow{\iota} A \xleftarrow{\iota}
      \{(a,\mathrm{o},b)\},\; d_a\Bigr),
  \]
  with one order wire on each boundary.

  \emph{The sync $s$.} The transition $s$ consumes the registered order together
  with the item and emits both onward, so its generator
  \begin{gather*}
    g_s = \Bigl(\{(b,\mathrm{o},s),(\bot,\mathrm{it},s)\}
      \xrightarrow{\iota} A \xleftarrow{\iota}\twocolbreak 
      \{(s,\mathrm{o},\top),(s,\mathrm{it},\top)\},\; d_s\Bigr)
  \end{gather*}
  carries two typed wires on each boundary and the single apex hyperedge labelled
  $s$.

  \emph{The silent skip.} The check $a$ may be skipped through $\tau$. A silent
  transition is a transparent mediator (Remark~\ref{rem:ocpn-silent-typing}) and
  produces no generator. It gives $b$ a second context instead. Tracing the order
  thread back from $b$ reaches either $a$ (the checked route) or, through $\tau$,
  directly the open end $\bot$ (the skip route), so $b$ has two generator cospans,
  \begin{gather*}
    g_{b,\mathrm{chk}} : \{(a,\mathrm{o},b)\}\to\{(b,\mathrm{o},s)\},\qquad\twocolbreak
    g_{b,\mathrm{skip}} : \{(\bot,\mathrm{o},b)\}\to\{(b,\mathrm{o},s)\},
  \end{gather*}
  differing only in their left boundary port. This is the Petri-net form of an
  exclusive choice. The two contexts give one generator per route, re-identified at the shared
  label $b$ by the construction above.
\end{exmp}

The example shows that the split place and its silent $\tau$
introduce no structural difference in the signature, as
Remark~\ref{rem:ab-place-vs-tau} records and the internal-routing quotient of
Corollary~\ref{cor:internal-routing-quotient} makes precise in the conversion
certificate.

\begin{figure*}[tp]
  \centering
  \input{tikz/running_example/re_pn_slice}\unskip
  \caption{Markings as cuts. A vertical cut through the connected string diagram of
    the net of Fig.~\ref{fig:re-ocpn} fixes a marking. Each wire it crosses is one
    token of that wire's type. The cuts $M_1$ (after $a$) and $M_2$ (after $b$) each
    hold one order and one item token, the contents of the intervening places.
    The diagram is closed by a start generator $\gamma_1$ and an end generator $\gamma_2$, a
    choice of the analysis (Section~\ref{subsec:analysis}), and sweeping the cut from $\gamma_1$
    to $\gamma_2$ replays the net.}
  \label{fig:re-pn-slice}
\end{figure*}

\begin{remark}[Markings as cuts]\label{rem:pn-marking-cut}
  A marking of the net corresponds to a vertical cut through its connected string
  diagram. The wires the cut crosses are exactly the tokens present, one per wire and
  typed by the wire (Fig.~\ref{fig:re-pn-slice}). Sweeping the cut from $\gamma_1$ to
  $\gamma_2$ reads off one execution as a sequence of markings, and the number of
  crossings of a given colour is the token count of that object type, recovering the
  marking semantics of Section~\ref{sec:hypergraphs} on this example. This is the
  Petri-net counterpart of the causal-net slicing of Section~\ref{sec:causal-nets}.
\end{remark}

%% file: tikz/running_example/re_pn_example.tex
\begin{tikzpicture}[petriBase, scale=0.95, inline,
    silentT/.style={transition, fill=black!22}]
  \node[T] (g1) at (0,-1.0625) {$\bot$};
  \node[blueP] (po) at (1.7,0) {};
  \node[T] (a)  at (3.4,0.85) {$a$};
  \node[silentT] (ta) at (3.4,-0.85) {$\tau$};
  \node[blueP] (pc) at (5.1,0) {};
  \node[T] (b)  at (6.7,0) {$b$};
  \node[blueP] (pr) at (8.4,0) {};
  \node[redP] (pi) at (3.4,-2.125) {};
  \node[T] (s) at (10.2,-1.0625) {$s$};
  \node[blueP] (pso) at (12.0,0) {};
  \node[redP] (psi) at (12.0,-2.125) {};
  \node[T] (g2) at (13.8,-1.0625) {$\top$};

  \draw[blueflow] ($(g1.north east)!0.25!(g1.south east)$) -- (po);
  \draw[redflow]  ($(g1.north east)!0.75!(g1.south east)$) -- (pi);
  \draw[blueflow] (po) -- (a);
  \draw[blueflow] (po) -- (ta);
  \draw[blueflow] (a) -- (pc);
  \draw[blueflow] (ta) -- (pc);
  \draw[blueflow] (pc) -- (b);
  \draw[blueflow] (b) -- (pr);
  \draw[blueflow] (pr) -- ($(s.north west)!0.25!(s.south west)$);
  \draw[redflow]  (pi) -- ($(s.north west)!0.75!(s.south west)$);
  \draw[blueflow] ($(s.north east)!0.25!(s.south east)$) -- (pso);
  \draw[redflow]  ($(s.north east)!0.75!(s.south east)$) -- (psi);
  \draw[blueflow] (pso) -- ($(g2.north west)!0.25!(g2.south west)$);
  \draw[redflow]  (psi) -- ($(g2.north west)!0.75!(g2.south west)$);

  \node[blueP, minimum size=5mm, label={[font=\small]right:order}]  at (15.3,-0.55) {};
  \node[redP,  minimum size=5mm, label={[font=\small]right:item}] at (15.3,-1.55) {};
\end{tikzpicture}

%% file: tikz/running_example/re_pn_cospans.tex
\begin{tikzpicture}[inline,
    bspd/.style={spider, fill=bluecol},
    rspd/.style={spider, fill=redcol},
    mbox/.style={draw=black, rounded corners=4pt, fill=white, minimum size=0.75cm, font=\small},
    plabel/.style={font=\tiny}]

  \node[font=\footnotesize] at (2.7,1.4) {decorated cospan};
  \node[font=\footnotesize] at (9.0,1.4) {string diagram};

  \node[font=\small] at (-2.15,0) {$g_s$};
  \node[cospanblue, minimum height=0.9cm] (sL) at (0,0) {};
  \node[bspd] at (0,0.20) {};
  \node[rspd] at (0,-0.20) {};
  \node at (0.9,0) {$\rightarrow$};
  \node[cospangrey, minimum height=1.0cm, minimum width=2.6cm] (sA) at (2.7,0) {};
  \node[mbox, minimum height=0.8cm] (sf) at (2.7,0) {$s$};
  \node[bspd] (sao) at (1.65,0.20) {};
  \node[rspd] (sai) at (1.65,-0.20) {};
  \node[bspd] (sbo) at (3.75,0.20) {};
  \node[rspd] (sbi) at (3.75,-0.20) {};
  \draw[bluewire] (sao) -- (sf.160);
  \draw[redwire]  (sai) -- (sf.200);
  \draw[bluewire] (sf.20)  -- (sbo);
  \draw[redwire]  (sf.-20) -- (sbi);
  \node at (4.5,0) {$\leftarrow$};
  \node[cospanblue, minimum height=0.9cm] (sR) at (5.4,0) {};
  \node[bspd] at (5.4,0.20) {};
  \node[rspd] at (5.4,-0.20) {};
  \node[plabel, anchor=east] at ([shift={(-3pt,0.20cm)}]sL.west) {$(b,\mathrm{o},s)$};
  \node[plabel, anchor=east] at ([shift={(-3pt,-0.20cm)}]sL.west) {$(\bot,\mathrm{it},s)$};
  \node[plabel, anchor=west] at ([shift={(3pt,0.20cm)}]sR.east) {$(s,\mathrm{o},\top)$};
  \node[plabel, anchor=west] at ([shift={(3pt,-0.20cm)}]sR.east) {$(s,\mathrm{it},\top)$};
  \node at (7.4,0) {$\mapsto$};
  \node[morphismblue, minimum height=0.95cm] (sd) at (9.0,0) {$s$};
  \draw[bluewire] ($(sd.south west)!0.70!(sd.north west)$) -- ++(-0.55,0);
  \draw[redwire]  ($(sd.south west)!0.30!(sd.north west)$) -- ++(-0.55,0);
  \draw[bluewire] ($(sd.south east)!0.70!(sd.north east)$) -- ++(0.55,0);
  \draw[redwire]  ($(sd.south east)!0.30!(sd.north east)$) -- ++(0.55,0);

  \begin{scope}[yshift=-1.5cm]
    \node[font=\small] at (-2.15,0) {$g_{b,\mathrm{chk}}$};
    \node[cospanblue, minimum height=0.6cm] (bL) at (0,0) {};
    \node[bspd] at (0,0) {};
    \node at (0.9,0) {$\rightarrow$};
    \node[cospangrey, minimum height=1.0cm, minimum width=2.6cm] (bA) at (2.7,0) {};
    \node[mbox] (bf) at (2.7,0) {$b$};
    \node[bspd] (bao) at (1.65,0) {};
    \node[bspd] (bbo) at (3.75,0) {};
    \draw[bluewire] (bao) -- (bf.west);
    \draw[bluewire] (bf.east) -- (bbo);
    \node at (4.5,0) {$\leftarrow$};
    \node[cospanblue, minimum height=0.6cm] (bR) at (5.4,0) {};
    \node[bspd] at (5.4,0) {};
    \node[plabel, anchor=east] at ([shift={(-3pt,0cm)}]bL.west) {$(a,\mathrm{o},b)$};
    \node[plabel, anchor=west] at ([shift={(3pt,0cm)}]bR.east) {$(b,\mathrm{o},s)$};
    \node at (7.4,0) {$\mapsto$};
    \node[morphismblue, minimum size=0.75cm] (bsd) at (9.0,0) {$b$};
    \draw[bluewire] (bsd.west) -- ++(-0.55,0);
    \draw[bluewire] (bsd.east) -- ++(0.55,0);
  \end{scope}

  \begin{scope}[yshift=-3.0cm]
    \node[font=\small] at (-2.15,0) {$g_{b,\mathrm{skip}}$};
    \node[cospanblue, minimum height=0.6cm] (kL) at (0,0) {};
    \node[bspd] at (0,0) {};
    \node at (0.9,0) {$\rightarrow$};
    \node[cospangrey, minimum height=1.0cm, minimum width=2.6cm] (kA) at (2.7,0) {};
    \node[mbox] (kf) at (2.7,0) {$b$};
    \node[bspd] (kao) at (1.65,0) {};
    \node[bspd] (kbo) at (3.75,0) {};
    \draw[bluewire] (kao) -- (kf.west);
    \draw[bluewire] (kf.east) -- (kbo);
    \node at (4.5,0) {$\leftarrow$};
    \node[cospanblue, minimum height=0.6cm] (kR) at (5.4,0) {};
    \node[bspd] at (5.4,0) {};
    \node[plabel, anchor=east] at ([shift={(-3pt,0cm)}]kL.west) {$(\bot,\mathrm{o},b)$};
    \node[plabel, anchor=west] at ([shift={(3pt,0cm)}]kR.east) {$(b,\mathrm{o},s)$};
    \node at (7.4,0) {$\mapsto$};
    \node[morphismblue, minimum size=0.75cm] (ksd) at (9.0,0) {$b$};
    \draw[bluewire] (ksd.west) -- ++(-0.55,0);
    \draw[bluewire] (ksd.east) -- ++(0.55,0);
  \end{scope}
\end{tikzpicture}

%% file: tikz/running_example/re_pn_slice.tex
\begin{tikzpicture}[inline, scale=1.0,
    cut/.style={dashed, thick, gray!70},
    cutlab/.style={font=\footnotesize, gray!80}]

  \node[morphismblue, minimum height=1.8cm] (g1) at (0,0)   {$\gamma_1$};
  \node[morphismblue] (a)  at (2.2,0.45) {$a$};
  \node[morphismblue] (b)  at (3.9,0.45) {$b$};
  \node[morphismblue, minimum height=1.8cm] (s)  at (6.0,0) {$s$};
  \node[morphismblue, minimum height=1.8cm] (g2) at (7.8,0) {$\gamma_2$};

  \draw[bluewire] ($(g1.north east)!0.25!(g1.south east)$) -- (a.west);
  \draw[bluewire] (a.east) -- (b.west);
  \draw[bluewire] (b.east) -- ($(s.north west)!0.25!(s.south west)$);
  \draw[bluewire] ($(s.north east)!0.25!(s.south east)$) -- ($(g2.north west)!0.25!(g2.south west)$);
  \draw[redwire]  ($(g1.north east)!0.75!(g1.south east)$) -- ($(s.north west)!0.75!(s.south west)$);
  \draw[redwire]  ($(s.north east)!0.75!(s.south east)$) -- ($(g2.north west)!0.75!(g2.south west)$);

  \draw[cut] (3.05,-0.7) -- (3.05,0.7);
  \node[cutlab, above] at (3.05,0.7) {$M_1$};
  \draw[cut] (4.95,-0.7) -- (4.95,0.7);
  \node[cutlab, above] at (4.95,0.7) {$M_2$};

  \node[cutlab, align=center] at (3.9,-1.5)
    {each cut is a marking; one crossed wire $=$ one token of that type\\
     $M_1=M_2=\{1\ \textcolor{bluecol}{\mathrm{order}},\ 1\ \textcolor{redcol}{\mathrm{item}}\}$};
\end{tikzpicture}

%% file: sections/04-causal_nets.tex
\subsection{Causal nets}\label{sec:causal-nets}

Causal nets are a process-mining formalism that refines Petri nets by equipping each activity with explicit input and output bindings. Where a Petri net records token multiplicity along a flow relation, a causal net records which subsets of predecessor activities jointly enable a given activity, and which subsets of successor activities are jointly produced. We translate OC causal nets into decorated cospans by the same cospan-algebra construction as for OCPNs. The only difference is that the interface sets come from bindings rather than from pre/post-set multiplicities. The reading is direct. A binding is an AND mediator, synchronising the typed pairs that fire together. The alternative bindings of an activity exclude one another, so one XOR mediator on each side of the activity chooses among them. We begin with bindings as plain sets of typed pairs. How many objects flow along each is a separate, later layer, the constraint system $\Lambda$ (Remark~\ref{rem:cnet-binding-numbers}).

\begin{definition}[Object-centric causal net]
  \label{def:causal-net}
  The object-centric (OC) causal net used here, or OCCN, is our formulation, drawn from causal
  nets~\cite{vanderaalstCausalNetsModeling2011} and their object-centric
  extension~\cite{lissObjectCentricCausalNets2025}. It is a tuple
  \[
    CN=(T,D,\mathrm{In},\mathrm{Out},\mathcal{O}),
  \]
  where \(T\) is a finite set of activities, \(\mathcal{O}\) is a finite set of object types, and
  \[
    D\subseteq (T\cup\{\bot\})\times\mathcal{O}\times (T\cup\{\top\})
  \]
  is a typed dependency multigraph, where each arc $(s,\omega,t)\in D$ records that $s$ can produce a token of type $\omega\in\mathcal{O}$ consumed by $t$. The symbols $\bot,\top\notin T$ are the open ends of Definition~\ref{def:contexts-general}, where objects enter and leave the net. The maps
  \begin{gather*}
    \mathrm{In}:T\to \mathcal{P}(\mathcal{P}((T\cup\{\bot\})\times\mathcal{O})),\\
    \mathrm{Out}:T\to \mathcal{P}(\mathcal{P}((T\cup\{\top\})\times\mathcal{O}))
  \end{gather*}
  assign to each activity $t$ a set of input bindings and a set of output bindings. Each input binding $X\in \mathrm{In}(t)$ is a finite set of $(s,\omega)$-pairs (typed predecessors jointly enabling $t$), and each output binding $Y\in \mathrm{Out}(t)$ is a finite set of $(t',\omega)$-pairs (typed successors jointly produced when $t$ fires), all consistent with~$D$.
\end{definition}

This construction instantiates the general framework of Section~\ref{sec:general-framework}, and
the mapping is the most direct of the four. An activity is a plain node, and each binding is an
AND mediator. An activity's several input (or output) bindings are exclusive alternatives, so one
XOR mediator on each side of the activity has the AND mediators of its bindings as branches. A
singleton binding is a one-in, one-out AND that passes flow straight through. A pair $(\bot,\omega)$
or $(\top,\omega)$ in a binding gives its AND mediator an edge from a mediator with no in-edges or
to one with no out-edges, so the walk ends there at the open end $\bot$ or $\top$
(Definition~\ref{def:contexts-general}). This determines the instance map.
\medskip
\begin{center}
  \small\begin{tabularx}{\linewidth}{@{}lX@{}}
    \hline
    General framework & Causal net \\
    \hline
    AND/OR graph $(V,E)$ & the activities $T$, with one XOR mediator on each side of an activity and one AND mediator per binding; typed wires read from $D$ \\
    Activities $\mathcal{A}$ & the activities $T$ \\
    AND mediators & one per binding in $\mathrm{In}(t),\mathrm{Out}(t)$, synchronising the typed pairs that fire together (a singleton binding is a pass-through) \\
    XOR mediators & one on each side of each activity $t$, choosing among the bindings of $\mathrm{In}(t)$ or of $\mathrm{Out}(t)$ \\
    \hline
  \end{tabularx}
\end{center}
\medskip
Running the construction of Definition~\ref{def:contexts-general} through this mapping is
immediate, since the bindings are already the contexts. The walk from an activity $t$ steps
through one XOR and one AND mediator on each side and reaches its neighbours at once, so its
backward and forward families are the input and output bindings themselves, $\mathrm{In}(t)$ and
$\mathrm{Out}(t)$. Every pair of an input binding with an output binding is a context,
\[
  \mathrm{Contexts}_{CN}(t):=\mathrm{In}(t)\times \mathrm{Out}(t),
\]
and each context $c=(P,S)$ yields one generator with typed wires
$L_c=\{(s,\{\omega\},t):(s,\omega)\in P\}$ and $R_c=\{(t,\{\omega\},t'):(t',\omega)\in S\}$.
Counts are carried separately by the constraint system $\Lambda$
(Remark~\ref{rem:cnet-binding-numbers}), so a split of one type across two output legs, as at the
review $r$ of Example~\ref{ex:running-cnet-cospan} below, is a context like any other. Generators
sharing a typed wire compose along it in $\mathbf{F}(\Sigma)$ (Definition~\ref{def:free-category}).
Start and end generators attached at $\bot$ and $\top$ belong to $\Sigma_\gamma$, a choice of the
analysis (Section~\ref{subsec:analysis}).

\begin{theorem}[Canonical cospan-algebra presentation of OC causal nets]
  \label{thm:causal-to-cospan-signature}
  Let $CN = (T, D, \mathrm{In}, \mathrm{Out}, \mathcal{O})$ be an OC causal net. Then $CN$ canonically determines a cospan-algebra signature $\Sigma$, and hence the free symmetric monoidal category $\mathbf{F}(\Sigma)$.
\end{theorem}

\begin{proof}
  The schema of Section~\ref{sec:notation-by-notation} applies at the instance map above,
  finiteness of $T$ and of all binding sets giving the conditions of
  Definition~\ref{def:lm-graph}, with each edge typed by the arc types in~$D$. The delta is that
  the walk is immediate. It steps through one XOR and one AND mediator on each side of an activity and
  reaches its neighbours at once, returning $\mathrm{In}(t)$ and $\mathrm{Out}(t)$. Every pair
  of them is a context, and
  $\mathrm{Contexts}_{CN}(t)=\mathrm{In}(t)\times \mathrm{Out}(t)$ is finite because $T$ and
  the binding sets are finite.
\end{proof}

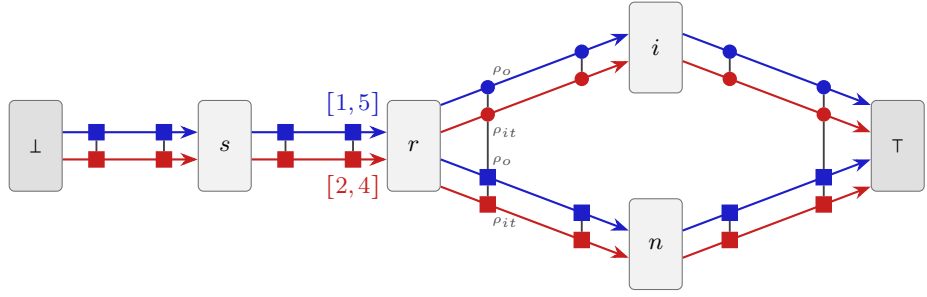
\begin{figure*}[tp]
  \centering
  \input{tikz/running_example/re_cn_example}\unskip
  \caption{An object-centric causal net over two object types, order (blue) and item (red).
    Every wire carries objects of one type, and every activity binds both types on each
    side. The solid connector joins the markers of a single binding, whose objects are
    consumed and produced together. On a wire a $\square$ marker is a batch of objects of
    that type and a $\circ$ marker is exactly one. The two bounds on the wires from the
    send $s$ to the review $r$ are declared per leg and independently of each other,
    $[1,5]$ orders and $[2,4]$ items. The review $r$ splits each type by a key of its own,
    $\rho_{o}$ over orders and $\rho_{it}$ over items. Each order and each item, taken singly,
    makes an exclusive choice between the legs that share a key, so each object leaves on one of
    them and the two legs hold the batch between them. The outcome $i$ takes one order and one
    item, and the outcome $n$ takes the remainder of each type. The boundary nodes $\bot$ and
    $\top$ are not activities of the net. They are the open ends where objects enter and leave
    (Definition~\ref{def:causal-net}).}
  \label{fig:re-occn}
\end{figure*}

\begin{figure*}[tp]
  \centering
  \input{tikz/running_example/re_cn_cospans}\unskip
  \caption{The generators $g_r$, $g_i$ and $g_n$ of Fig.~\ref{fig:re-occn}, each shown as a decorated cospan (left)
    and the string diagram it denotes (right). Order wires are blue and item wires red, and the
    labels beside each boundary name its typed ports in the order drawn. The outcome
    generators are $g_i$, one investigated order and one item, and $g_n$, the remainder of
    each type. Each generator has a constraint system of its own, set out below its row and
    flagged on the string diagram by the $\Lambda$ marker spanning both boundaries. A
    system's independent line records the bounds that generator declares on its own legs,
    and its coupling line relates its two boundaries. Composition identifies the variables of
    shared wires and conjoins the systems, so three systems that each hold on their own need not
    hold together.}
  \label{fig:re-occn-cospans}
\end{figure*}
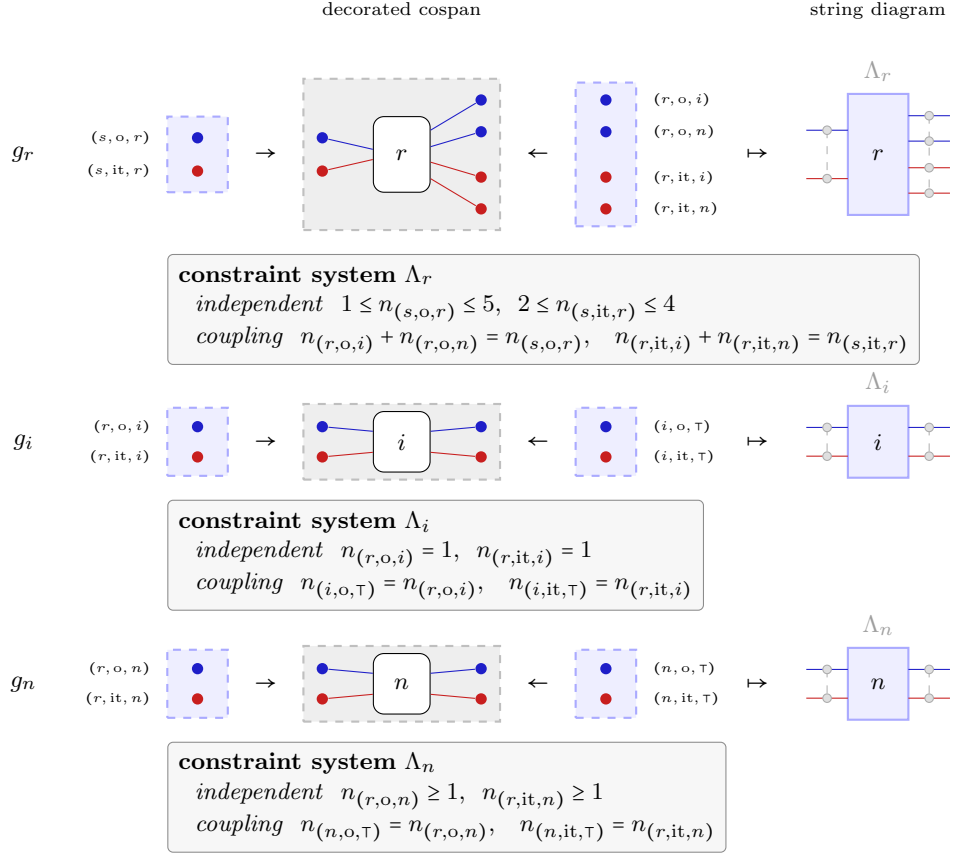

\begin{exmp}[A small OCCN and its generator cospans]
  \label{ex:running-cnet-cospan}
  Consider the OC causal net of Fig.~\ref{fig:re-occn}, over two object types,
  $\mathrm{order}$ and $\mathrm{item}$, each key-distributed at the review $r$. It has no places. Write $(x,\mathrm{o},y)$ for $(x,\{\mathrm{o}\},y)$. Each
  activity carries its bindings, which are already its contexts:
  \begin{gather*}
    \mathrm{Out}(s)=\{\{(r,\mathrm{o}),(r,\mathrm{it})\}\},\qquad\twocolbreak
    \mathrm{In}(r)=\{\{(s,\mathrm{o}),(s,\mathrm{it})\}\},
  \end{gather*}
  \begin{gather*}
    \mathrm{Out}(r)=\{\{(i,\mathrm{o}),(i,\mathrm{it}),\twocolbreak(n,\mathrm{o}),(n,\mathrm{it})\}\},\\
    \mathrm{In}(i)=\mathrm{In}(n)=\{\{(r,\mathrm{o}),(r,\mathrm{it})\}\}.
  \end{gather*}
  Activity $r$ has a single output binding $\{(i,\mathrm{o}),(i,\mathrm{it}),(n,\mathrm{o}),(n,\mathrm{it})\}$. It
  emits both legs on every firing and is a key-distribution split, not an
  exclusive choice. Reading a generator off each binding
  (Fig.~\ref{fig:re-occn-cospans}) gives, for $r$,
  \begin{gather*}
    g_r = \Bigl(\{(s,\mathrm{o},r),(s,\mathrm{it},r)\}
      \xrightarrow{\iota} A \xleftarrow{\iota}\twocolbreak 
      \{(r,\mathrm{o},i),(r,\mathrm{it},i),(r,\mathrm{o},n),(r,\mathrm{it},n)\},\; d_r\Bigr),
  \end{gather*}
  whose decoration carries the constraint system $\Lambda$
  (Remark~\ref{rem:cnet-binding-numbers}). Read off $g_r$'s boundary, the markings say what each
  wire carries in one firing of $r$. The input wire $(s,\mathrm{o},r)$ takes a batch of one to
  five orders ($1\le n_{(s,\mathrm{o},r)}\le 5$), and the investigated leg $(r,\mathrm{o},i)$
  holds exactly one ($n_{(r,\mathrm{o},i)}=1$). The two output legs hold the whole batch between
  them, tied by a within-type key $\rho_{o}$ that partitions it,
  $n_{(r,\mathrm{o},i)}+n_{(r,\mathrm{o},n)}=n_{(s,\mathrm{o},r)}$. Hence one firing of $r$
  receives a batch and splits it, one order to $i$ and the rest to $n$, and a batch of more than five
  forces $r$ to fire again. The item type is split in the same way by its own within-type key $\rho_{it}$, one item
  investigated on $i$ and the rest on $n$, so both types follow the identical key-distribution
  pattern under bounds of their own.
\end{exmp}

\begin{remark}[Binding numbers and keys as the constraint system]
  \label{rem:cnet-binding-numbers}
  Definition~\ref{def:causal-net} records which typed pairs bind, but not
  \emph{how many} objects flow along each. The full OC causal-net
  formalism~\cite{lissObjectCentricCausalNets2025} additionally equips every
  binding pair with a cardinality $[c_{\min},c_{\max}]$ and may correlate
  pairs through a shared object key $\rho$. Both enter the constraint system of
  Definition~\ref{def:multiplicity-decoration}. The binding cardinality is a
  per-leg interval, an inequality of $\rho(t)$ on a variable leg, and the shared key is an equality of $\rho(t)$ at its activity $t$ (Definition~\ref{def:lm-graph})
  among the legs of one generator, a partition of the input wire across the output wires it feeds.
  The routing
  structure, which generators exist, is untouched. The numbers decorate those
  generators, and an unannotated binding defaults to $n_p=1$, the unannotated sense of the
  generators of Example~\ref{ex:running-cnet-cospan}.

  On the small example above this attaches as follows. The send activity $s$ ships
  a batch of orders and a batch of items to $r$ in one firing, so each wire is bounded on its own,
  $1\le n_{(s,\mathrm{o},r)}\le 5$ for orders and $2\le n_{(s,\mathrm{it},r)}\le 4$ for items. The receive activity $r$
  then distributes each type between the two outcomes by its own within-type key. The order key
  $\rho_{o}$ partitions the orders across the investigated and not-investigated legs,
  $n_{(r,\mathrm{o},i)}+n_{(r,\mathrm{o},n)}-n_{(s,\mathrm{o},r)}=0$, with exactly one investigated
  per firing, $n_{(r,\mathrm{o},i)}=1$. The item key $\rho_{it}$ partitions the items the
  same way, $n_{(r,\mathrm{it},i)}+n_{(r,\mathrm{it},n)}-n_{(s,\mathrm{it},r)}=0$ with
  $n_{(r,\mathrm{it},i)}=1$, so the two key constraints sit side by side in $\Lambda$ and make the
  distribution explicit. Hence the number of investigations equals the number of $r$ firings, and a
  batch above either bound forces $r$ to fire again.

  Generators sharing a wire identify its leg variable and conjoin their
  constraints, and a closed composite exists if and only if the glued constraint system is
  solvable over $\mathbb{N}$. The symbolic signature carrying
  $\Lambda$ is bound-independent and is the form compared across notations.
\end{remark}

\begin{remark}[Causal net states]
  The state of a causal net mid-execution is the causal-net version of a Petri-net
  marking as a cut (Remark~\ref{rem:pn-marking-cut}). A cut through the string diagram
  carries the active wire ports, the bindings in progress rather than tokens in places,
  and sweeping the cut from boundary to boundary reads off the state evolution with no
  separate algebraic state machinery. On the OCCN example, after the $r$ generator fires
  the cut carries the four ports leading to $i$ and~$n$ simultaneously, an order and an item wire to each, and firing
  $i$ or $n$ first advances that outcome's two ports to $\top$, two binding sequences of one
  diagram.
\end{remark}

%% file: tikz/running_example/re_cn_example.tex
\begin{tikzpicture}[scale=1.0, inline,
    cnact/.style={rectangle, rounded corners=2pt, draw=black!55, fill=black!5,
                  minimum width=7mm, minimum height=12mm, font=\small},
    bnd/.style={rectangle, rounded corners=2pt, draw=black!55, fill=black!12,
                minimum width=7mm, minimum height=12mm, font=\small},
    omk/.style={circle, draw=#1, fill=#1, minimum size=4.6pt, inner sep=0pt},
    smk/.style={rectangle, sharp corners, draw=#1, fill=#1, minimum size=5.4pt, inner sep=0pt},
    card/.style={font=\small, inner sep=0.5pt},
    klab/.style={font=\tiny, black!70, inner sep=0.5pt},
    andconn/.style={black!75, line width=0.7pt},
    keyconn/.style={black!75, line width=0.7pt}]   %
  \def\dy{1.8mm}
  \def\rgap{4pt}
  \node[bnd]   (g1) at (0,0)      {$\bot$};
  \node[cnact] (s)  at (2.5,0)    {$s$};
  \node[cnact] (r)  at (5.0,0)    {$r$};
  \node[cnact] (i)  at (8.2,1.30) {$i$};
  \node[cnact] (n)  at (8.2,-1.30){$n$};
  \node[bnd]   (g2) at (11.4,0)   {$\top$};

  \draw[blueflow] ([yshift=\dy]g1.east) --
      node[smk=bluecol, pos=0.25] (a1){} node[smk=bluecol, pos=0.75] (a2){} ([yshift=\dy]s.west);
  \draw[redflow]  ([yshift=-\dy]g1.east) --
      node[smk=redcol, pos=0.25] (a3){} node[smk=redcol, pos=0.75] (a4){} ([yshift=-\dy]s.west);
  \draw[andconn] (a1)--(a3);   \draw[andconn] (a2)--(a4);
  \draw[blueflow] ([yshift=\dy]s.east) --
      node[smk=bluecol, pos=0.25] (b1){}
      node[smk=bluecol, pos=0.75, label={[card,bluecol,label distance=2pt]above:$[1,5]$}] (b2){}
      ([yshift=\dy]r.west);
  \draw[redflow]  ([yshift=-\dy]s.east) --
      node[smk=redcol, pos=0.25] (b3){}
      node[smk=redcol, pos=0.75, label={[card,redcol,label distance=2pt]below:$[2,4]$}] (b4){}
      ([yshift=-\dy]r.west);
  \draw[andconn] (b1)--(b3);   \draw[andconn] (b2)--(b4);

  \draw[blueflow] ([yshift=0.54cm]r.east) --
      node[omk=bluecol, pos=0.25] (c1){} node[omk=bluecol, pos=0.75] (d1){} ([yshift=\dy]i.west);   %
  \draw[redflow]  ([yshift=0.18cm]r.east) --
      node[omk=redcol, pos=0.25] (c2){} node[omk=redcol, pos=0.75] (d2){} ([yshift=-\dy]i.west);    %
  \draw[blueflow] ([yshift=-0.18cm]r.east) --
      node[smk=bluecol, pos=0.25] (c3){} node[smk=bluecol, pos=0.75] (e1){} ([yshift=\dy]n.west);   %
  \draw[redflow]  ([yshift=-0.54cm]r.east) --
      node[smk=redcol, pos=0.25] (c4){} node[smk=redcol, pos=0.75] (e2){} ([yshift=-\dy]n.west);    %
  \draw[keyconn] (c1)--(c2)--(c3)--(c4);            %
  \node[klab, anchor=south west] at ([shift={(1pt,\rgap)}]c1) {$\rho_{o}$};  \node[klab, anchor=south west] at ([shift={(1pt,\rgap)}]c3) {$\rho_{o}$};
  \node[klab, anchor=north west] at ([shift={(1pt,-\rgap)}]c2) {$\rho_{it}$}; \node[klab, anchor=north west] at ([shift={(1pt,-\rgap)}]c4) {$\rho_{it}$};
  \draw[andconn] (d1)--(d2);   %
  \draw[andconn] (e1)--(e2);   %

  \draw[blueflow] ([yshift=\dy]i.east) --
      node[omk=bluecol, pos=0.25] (f1){} node[omk=bluecol, pos=0.75] (h1){} ([yshift=0.54cm]g2.west);
  \draw[redflow]  ([yshift=-\dy]i.east) --
      node[omk=redcol, pos=0.25] (f2){} node[omk=redcol, pos=0.75] (h3){} ([yshift=0.18cm]g2.west);
  \draw[andconn] (f1)--(f2);   %
  \draw[blueflow] ([yshift=\dy]n.east) --
      node[smk=bluecol, pos=0.25] (k1){} node[smk=bluecol, pos=0.75] (h5){} ([yshift=-0.18cm]g2.west);
  \draw[redflow]  ([yshift=-\dy]n.east) --
      node[smk=redcol, pos=0.25] (k2){} node[smk=redcol, pos=0.75] (h7){} ([yshift=-0.54cm]g2.west);
  \draw[andconn] (k1)--(k2);   %
  \draw[andconn] (h1)--(h3)--(h5)--(h7);   %
\end{tikzpicture}

%% file: tikz/running_example/re_cn_cospans.tex
\begin{tikzpicture}[inline,
    bspd/.style={spider, fill=bluecol},
    rspd/.style={spider, fill=redcol},
    mbox/.style={draw=black, rounded corners=4pt, fill=white, minimum size=0.75cm, font=\small},
    plabel/.style={font=\tiny},
    cardb/.style={font=\footnotesize, bluecol, inner sep=0.5pt},
    cardr/.style={font=\footnotesize, redcol, inner sep=0.5pt}]

  \node[font=\footnotesize] at (2.7,1.9) {decorated cospan};
  \node[font=\footnotesize] at (9.0,1.9) {string diagram};

  \node[font=\small] at (-2.3,0) {$g_r$};
  \node[cospanblue, minimum height=1.0cm] (rLb) at (0,0) {};
  \node[bspd] (rLo)  at (0,0.22) {};
  \node[rspd] (rLit) at (0,-0.22) {};
  \node at (0.9,0) {$\rightarrow$};
  \node[cospangrey, minimum height=2.0cm, minimum width=2.6cm] (rA) at (2.7,0) {};
  \node[mbox, minimum height=1.0cm] (rf) at (2.7,0) {$r$};
  \node[bspd] (rao) at (1.65,0.22) {};
  \node[rspd] (rat) at (1.65,-0.22) {};
  \node[bspd] (roi) at (3.75,0.72) {};
  \node[bspd] (ron) at (3.75,0.30) {};
  \node[rspd] (rii) at (3.75,-0.30) {};
  \node[rspd] (rin) at (3.75,-0.72) {};
  \draw[bluewire] (rao) -- (rf.170);
  \draw[redwire]  (rat) -- (rf.190);
  \draw[bluewire] (rf.40)  -- (roi);
  \draw[bluewire] (rf.15)  -- (ron);
  \draw[redwire]  (rf.-15) -- (rii);
  \draw[redwire]  (rf.-40) -- (rin);
  \node at (4.5,0) {$\leftarrow$};
  \node[cospanblue, minimum height=1.9cm] (rRb) at (5.4,0) {};
  \node[bspd] (rRoi) at (5.4,0.72) {};
  \node[bspd] (rRon) at (5.4,0.30) {};
  \node[rspd] (rRii) at (5.4,-0.30) {};
  \node[rspd] (rRin) at (5.4,-0.72) {};
  \node[plabel, anchor=east] at ([shift={(-3pt,0.22cm)}]rLb.west) {$(s,\mathrm{o},r)$};
  \node[plabel, anchor=east] at ([shift={(-3pt,-0.22cm)}]rLb.west) {$(s,\mathrm{it},r)$};
  \node[plabel, anchor=west] at ([shift={(3pt,0.72cm)}]rRb.east) {$(r,\mathrm{o},i)$};
  \node[plabel, anchor=west] at ([shift={(3pt,0.30cm)}]rRb.east) {$(r,\mathrm{o},n)$};
  \node[plabel, anchor=west] at ([shift={(3pt,-0.30cm)}]rRb.east) {$(r,\mathrm{it},i)$};
  \node[plabel, anchor=west] at ([shift={(3pt,-0.72cm)}]rRb.east) {$(r,\mathrm{it},n)$};
  \node at (7.4,0) {$\mapsto$};
  \node[morphismblue, minimum height=1.6cm] (rsd) at (9.0,0) {$r$};
  \draw[bluewire] ($(rsd.south west)!0.70!(rsd.north west)$) -- ++(-0.55,0);
  \draw[redwire]  ($(rsd.south west)!0.30!(rsd.north west)$) -- ++(-0.55,0);
  \coordinate (p1) at ($($(rsd.south west)!0.70!(rsd.north west)$)-(0.275,0)$);
  \coordinate (p2) at ($($(rsd.south west)!0.30!(rsd.north west)$)-(0.275,0)$);
  \draw[gray!60, dashed] (p1) -- (p2);
  \foreach \m in {p1,p2}
    \node[circle, draw=gray!65, fill=gray!25, minimum size=3.2pt, inner sep=0] at (\m) {};
  \draw[bluewire] ($(rsd.south east)!0.82!(rsd.north east)$) -- ++(0.55,0);
  \draw[bluewire] ($(rsd.south east)!0.61!(rsd.north east)$) -- ++(0.55,0);
  \draw[redwire]  ($(rsd.south east)!0.39!(rsd.north east)$) -- ++(0.55,0);
  \draw[redwire]  ($(rsd.south east)!0.18!(rsd.north east)$) -- ++(0.55,0);
  \coordinate (m1) at ($($(rsd.south east)!0.82!(rsd.north east)$)+(0.275,0)$);
  \coordinate (m2) at ($($(rsd.south east)!0.61!(rsd.north east)$)+(0.275,0)$);
  \coordinate (m3) at ($($(rsd.south east)!0.39!(rsd.north east)$)+(0.275,0)$);
  \coordinate (m4) at ($($(rsd.south east)!0.18!(rsd.north east)$)+(0.275,0)$);
  \draw[gray!60, dashed] (m1) -- (m4);
  \foreach \m in {m1,m2,m3,m4}
    \node[circle, draw=gray!65, fill=gray!25, minimum size=3.2pt, inner sep=0] at (\m) {};
  \node[gray!75, font=\small] at ($(rsd.north)+(0,0.30)$) {$\Lambda_r$};
  \node[draw=black!45, rounded corners=2pt, fill=black!3, align=left, font=\scriptsize,
        inner sep=4pt, anchor=west] (Lam) at (-0.4,-2.05)
    {\textbf{constraint system} $\Lambda_r$\\[2pt]
     \ \ \emph{independent}\ \ $1\le n_{(s,\mathrm{o},r)}\le 5$,\ \ $2\le n_{(s,\mathrm{it},r)}\le 4$\\[1pt]
     \ \ \emph{coupling}\ \ $n_{(r,\mathrm{o},i)}+n_{(r,\mathrm{o},n)}=n_{(s,\mathrm{o},r)}$,\ \ \
     $n_{(r,\mathrm{it},i)}+n_{(r,\mathrm{it},n)}=n_{(s,\mathrm{it},r)}$};

  \begin{scope}[yshift=-3.8cm]
    \node[font=\small] at (-2.3,0) {$g_i$};
    \node[cospanblue, minimum height=0.9cm] (iL) at (0,0) {};
    \node[bspd] (iLo)  at (0,0.20) {};
    \node[rspd] (iLit) at (0,-0.20) {};
    \node at (0.9,0) {$\rightarrow$};
    \node[cospangrey, minimum height=1.0cm, minimum width=2.6cm] (iA) at (2.7,0) {};
    \node[mbox, minimum height=0.8cm] (iff) at (2.7,0) {$i$};
    \node[bspd] (iao) at (1.65,0.20) {};
    \node[rspd] (iat) at (1.65,-0.20) {};
    \node[bspd] (ibo) at (3.75,0.20) {};
    \node[rspd] (ibt) at (3.75,-0.20) {};
    \draw[bluewire] (iao) -- (iff.160);
    \draw[redwire]  (iat) -- (iff.200);
    \draw[bluewire] (iff.20)  -- (ibo);
    \draw[redwire]  (iff.-20) -- (ibt);
    \node at (4.5,0) {$\leftarrow$};
    \node[cospanblue, minimum height=0.9cm] (iR) at (5.4,0) {};
    \node[bspd] (iRo)  at (5.4,0.20) {};
    \node[rspd] (iRit) at (5.4,-0.20) {};
    \node[plabel, anchor=east] at ([shift={(-3pt,0.20cm)}]iL.west) {$(r,\mathrm{o},i)$};
    \node[plabel, anchor=east] at ([shift={(-3pt,-0.20cm)}]iL.west) {$(r,\mathrm{it},i)$};
    \node[plabel, anchor=west] at ([shift={(3pt,0.20cm)}]iR.east) {$(i,\mathrm{o},\top)$};
    \node[plabel, anchor=west] at ([shift={(3pt,-0.20cm)}]iR.east) {$(i,\mathrm{it},\top)$};
    \node at (7.4,0) {$\mapsto$};
    \node[morphismblue, minimum height=0.95cm] (isd) at (9.0,0) {$i$};
    \draw[bluewire] ($(isd.south west)!0.70!(isd.north west)$) -- ++(-0.55,0);
    \draw[redwire]  ($(isd.south west)!0.30!(isd.north west)$) -- ++(-0.55,0);
    \draw[bluewire] ($(isd.south east)!0.70!(isd.north east)$) -- ++(0.55,0);
    \draw[redwire]  ($(isd.south east)!0.30!(isd.north east)$) -- ++(0.55,0);
    \coordinate (isdp1) at ($($(isd.south west)!0.70!(isd.north west)$)-(0.275,0)$);
    \coordinate (isdp2) at ($($(isd.south west)!0.30!(isd.north west)$)-(0.275,0)$);
    \coordinate (isdm1) at ($($(isd.south east)!0.70!(isd.north east)$)+(0.275,0)$);
    \coordinate (isdm2) at ($($(isd.south east)!0.30!(isd.north east)$)+(0.275,0)$);
    \draw[gray!60, dashed] (isdp1) -- (isdp2);
    \draw[gray!60, dashed] (isdm1) -- (isdm2);
    \foreach \m in {isdp1,isdp2,isdm1,isdm2}
      \node[circle, draw=gray!65, fill=gray!25, minimum size=3.2pt, inner sep=0] at (\m) {};
    \node[gray!75, font=\small] at ($(isd.north)+(0,0.30)$) {$\Lambda_i$};
  \end{scope}

  \begin{scope}[yshift=-7.0cm]
    \node[font=\small] at (-2.3,0) {$g_n$};
    \node[cospanblue, minimum height=0.9cm] (nL) at (0,0) {};
    \node[bspd] (nLo)  at (0,0.20) {};
    \node[rspd] (nLit) at (0,-0.20) {};
    \node at (0.9,0) {$\rightarrow$};
    \node[cospangrey, minimum height=1.0cm, minimum width=2.6cm] (nA) at (2.7,0) {};
    \node[mbox, minimum height=0.8cm] (nf) at (2.7,0) {$n$};
    \node[bspd] (nao) at (1.65,0.20) {};
    \node[rspd] (nat) at (1.65,-0.20) {};
    \node[bspd] (nbo) at (3.75,0.20) {};
    \node[rspd] (nbt) at (3.75,-0.20) {};
    \draw[bluewire] (nao) -- (nf.160);
    \draw[redwire]  (nat) -- (nf.200);
    \draw[bluewire] (nf.20)  -- (nbo);
    \draw[redwire]  (nf.-20) -- (nbt);
    \node at (4.5,0) {$\leftarrow$};
    \node[cospanblue, minimum height=0.9cm] (nR) at (5.4,0) {};
    \node[bspd] (nRo)  at (5.4,0.20) {};
    \node[rspd] (nRit) at (5.4,-0.20) {};
    \node[plabel, anchor=east] at ([shift={(-3pt,0.20cm)}]nL.west) {$(r,\mathrm{o},n)$};
    \node[plabel, anchor=east] at ([shift={(-3pt,-0.20cm)}]nL.west) {$(r,\mathrm{it},n)$};
    \node[plabel, anchor=west] at ([shift={(3pt,0.20cm)}]nR.east) {$(n,\mathrm{o},\top)$};
    \node[plabel, anchor=west] at ([shift={(3pt,-0.20cm)}]nR.east) {$(n,\mathrm{it},\top)$};
    \node at (7.4,0) {$\mapsto$};
    \node[morphismblue, minimum height=0.95cm] (nsd) at (9.0,0) {$n$};
    \draw[bluewire] ($(nsd.south west)!0.70!(nsd.north west)$) -- ++(-0.55,0);
    \draw[redwire]  ($(nsd.south west)!0.30!(nsd.north west)$) -- ++(-0.55,0);
    \draw[bluewire] ($(nsd.south east)!0.70!(nsd.north east)$) -- ++(0.55,0);
    \draw[redwire]  ($(nsd.south east)!0.30!(nsd.north east)$) -- ++(0.55,0);
    \coordinate (nsdp1) at ($($(nsd.south west)!0.70!(nsd.north west)$)-(0.275,0)$);
    \coordinate (nsdp2) at ($($(nsd.south west)!0.30!(nsd.north west)$)-(0.275,0)$);
    \coordinate (nsdm1) at ($($(nsd.south east)!0.70!(nsd.north east)$)+(0.275,0)$);
    \coordinate (nsdm2) at ($($(nsd.south east)!0.30!(nsd.north east)$)+(0.275,0)$);
    \draw[gray!60, dashed] (nsdp1) -- (nsdp2);
    \draw[gray!60, dashed] (nsdm1) -- (nsdm2);
    \foreach \m in {nsdp1,nsdp2,nsdm1,nsdm2}
      \node[circle, draw=gray!65, fill=gray!25, minimum size=3.2pt, inner sep=0] at (\m) {};
    \node[gray!75, font=\small] at ($(nsd.north)+(0,0.30)$) {$\Lambda_n$};
  \end{scope}
  \node[draw=black!45, rounded corners=2pt, fill=black!3, align=left, font=\scriptsize,
        inner sep=4pt, anchor=west] (Lami) at (-0.4,-5.3)
    {\textbf{constraint system} $\Lambda_i$\\[2pt]
     \ \ \emph{independent}\ \ $n_{(r,\mathrm{o},i)}=1$,\ \ $n_{(r,\mathrm{it},i)}=1$\\[1pt]
     \ \ \emph{coupling}\ \ $n_{(i,\mathrm{o},\top)}=n_{(r,\mathrm{o},i)}$,\ \ \ $n_{(i,\mathrm{it},\top)}=n_{(r,\mathrm{it},i)}$};
  \node[draw=black!45, rounded corners=2pt, fill=black!3, align=left, font=\scriptsize,
        inner sep=4pt, anchor=west] (Lamn) at (-0.4,-8.5)
    {\textbf{constraint system} $\Lambda_n$\\[2pt]
     \ \ \emph{independent}\ \ $n_{(r,\mathrm{o},n)}\ge 1$,\ \ $n_{(r,\mathrm{it},n)}\ge 1$\\[1pt]
     \ \ \emph{coupling}\ \ $n_{(n,\mathrm{o},\top)}=n_{(r,\mathrm{o},n)}$,\ \ \ $n_{(n,\mathrm{it},\top)}=n_{(r,\mathrm{it},n)}$};
\end{tikzpicture}

%% file: sections/05-process_trees.tex
\subsection{Process trees}\label{sec:process_trees}

Process trees are a block-structured process-mining formalism widely used in discovery algorithms and model-quality evaluation~\cite{leemansDiscoveringBlockStructuredProcess2013,buijsQualityDimensionsProcess2014}. Each internal node carries one of four control-flow operators (sequence, XOR-choice, AND-parallel, redo) and the compositional structure guarantees soundness of the resulting workflow-net translation. We translate OC process trees (OCPTs) into decorated cospans by reading interface sets directly from the operator structure, following the same cospan-algebra construction as for OCPNs and OC causal nets.

\begin{definition}[Object-centric process tree]
  \label{def:process-tree}
  An object-centric (OC) process tree~\cite{leemansDiscoveringBlockStructuredProcess2013,vandettenDiscoveringCompactLive2024}
  over a finite set of object types $\mathcal{O}$ is a finite rooted tree of the following form.
  \begin{itemize}
    \item Each internal node is a control-flow operator, applied to its child subtrees, one of
      sequence $\to$, exclusive choice $\times$, concurrency $+$, or redo $\circlearrowleft$.
    \item Each leaf is an object-centric activity $(a,\mathrm{rel})$ with an activity name $a$ and
      the object types $\mathrm{rel}\subseteq\mathcal{O}$ it involves.
  \end{itemize}
  The open ends $\bot$ and $\top$ of Definition~\ref{def:contexts-general} bound the root, and the
  analysis may attach a start generator $\gamma_1$ and an end generator $\gamma_2$ there
  (Section~\ref{subsec:analysis}). Each object type runs
  through the operators on its own, independently of the others.
\end{definition}

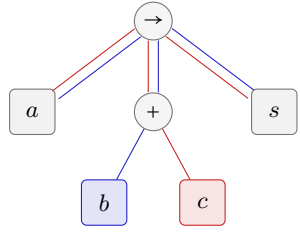
\begin{figure}[tb]
  \centering
  \input{tikz/running_example/re_pt_tree}\unskip
  \caption{The OCPT example $\to(a,+(b,c),s)$, with tree connectors coloured by object type (order
    blue, item red). A connector carrying both types is two parallel coloured wires, a single type
    one coloured wire.}
  \label{fig:re-ocpt}
\end{figure}

The operators apply per object. $\to(S_1,\ldots,S_n)$ runs its subtrees in order.
$\times(S_1,\ldots,S_n)$ runs exactly one of them. Under $+(S_1,\ldots,S_n)$ they all run
concurrently. Finally, $\circlearrowleft(S,B_1,\ldots,B_m)$ runs $S$ once, then any number of times
takes some redo $B_j$ and runs $S$ again. Each object type follows its own path, so one type may
take a $\times$-branch another does not, and a type outside a leaf's $\mathrm{rel}$ skips that
leaf. Leaf labels may repeat, each occurrence firing independently. We write $T$ for the leaf
activities and $T_\tau\subseteq T$ for the silent leaves. A silent leaf carries a type but no
observable label and models a skip, the XOR alternative $\times(a,\tau)$ to the activity it
bypasses. The induced typed directly-follows arcs are $D_{PT}\subseteq T\times\mathcal{O}\times T$,
with $(x,\omega,t)\in D_{PT}$ when type $\omega$ flows from leaf $x$ to leaf $t$. When a leaf
relates to a single type $\mathrm{type}(t)$, the arcs $D_{PT}$ are forced by the tree. The source
formalism also marks a leaf's types as divergent, convergent or deficient. Those marks are
existential statements over an event log~\cite{vandettenDiscoveringCompactLive2024}, properties of
the flattened object-centric data a tree is discovered from, and the cospan language needs none of
them. A count constraint on a leaf's legs is written in $\Lambda$, as for every notation
(Definition~\ref{def:multiplicity-decoration}).

This construction instantiates the general framework of Section~\ref{sec:general-framework}, and
the two mapping choices are simple. An observable leaf is an activity. Each control-flow operator
becomes a mediator, the choice $\times$ an XOR and the parallel $+$ an AND. Sequence $\to$ and redo
$\circlearrowleft$ need no mediator of their own. They are built from these two, and a
degree-$(1,1)$ mediator, where AND, XOR and a SEQ pass-through all coincide, simply forwards its one
branch. A silent leaf is a transparent mediator and contributes no generator, as for the silent
transitions of Section~\ref{sec:petri-nets}. Compiling the operator tree this way gives an AND/OR
graph, the compiled flow graph $\mathcal{G}_{PT}$ (Definition~\ref{def:pt-flow-graph}), over which
the construction is the walk of Definition~\ref{def:contexts-general}. This is the instance where
the walk is not immediate. Where the object-centric causal net pre-compiles its routing into binding
sets and the walk ends at once, the process tree spreads its routing across operator structure that
the walk must cross. The instance map is as follows.
\medskip
\begin{center}
  \small\begin{tabularx}{\linewidth}{@{}lX@{}}
    \hline
    General framework & Process tree \\
    \hline
    AND/OR graph $(V,E)$ & the compiled flow graph $\mathcal{G}_{PT}$ (Definition~\ref{def:pt-flow-graph}) \\
    Activities $\mathcal{A}$ & observable leaves $T\setminus T_\tau$ \\
    AND mediators & the parallel operator $+$ (split and join) \\
    XOR mediators & the choice operator $\times$, and a loop's redo-or-exit \\
    SEQ pass-throughs & the sequence operator $\to$ and the silent leaves $T_\tau$ (transparent); the degree-$(1,1)$ case where AND, XOR and SEQ coincide \\
    \hline
  \end{tabularx}
\end{center}
\medskip
\begin{definition}[Compiled flow graph]
  \label{def:pt-flow-graph}
  The compiled flow graph $\mathcal{G}_{PT}=\mathcal{G}(\Omega)$, for $\Omega$ the operator tree of $PT$, is built bottom-up from the
  operator tree, the map $\mathcal{G}(\cdot)$ sending each subtree to a directed graph with one
  entry and one exit. It is the AND/OR graph $(V,E)$ of the instance map, not the raw
  parent-to-child tree. The cases are as follows.
  \begin{itemize}
    \item An observable leaf $t\in T\setminus T_\tau$ gives the single activity node $t$.
    \item A silent leaf $\tau\in T_\tau$ gives a degree-$(1,1)$ SEQ pass-through carrying type
      $\mathrm{type}(\tau)$, transparent to the walk.
    \item For $\to(S_1,\ldots,S_n)$, chain the children through $n{-}1$ fresh SEQ pass-throughs,
      $\mathcal{G}(S_1)\to s_1\to\cdots\to s_{n-1}\to\mathcal{G}(S_n)$. The ordering is carried by the
      chain, not by any one node.
    \item For $+(S_1,\ldots,S_n)$, a fresh AND-split $p^+$ fans into each $\mathcal{G}(S_i)$ and a
      fresh AND-join $q^+$ collects them.
    \item For $\times(S_1,\ldots,S_n)$, take a fresh XOR-split $x^\times$ and XOR-join $y^\times$, each
      branch $\mathcal{G}(S_i)$ running between them. The walk keeps the reached sets of the branches separate at $x^\times$.
    \item For $\circlearrowleft(S,B_1,\ldots,B_m)$, take a mandatory $\mathcal{G}(S)$, then an XOR mediator
      that either exits to the next block or takes a back-edge through some $\mathcal{G}(B_j)$ and
      repeats $S$. A repeated pass of the loop is a composite of generators and adds no signature
      element.
  \end{itemize}
\end{definition}

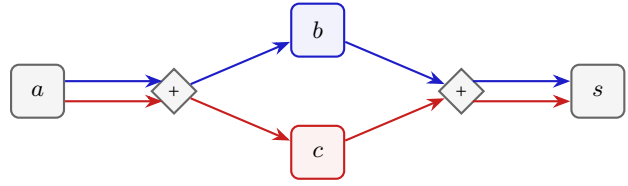
\begin{figure*}[tp]
  \centering
  \input{tikz/running_example/re_pt_flow}\unskip
  \caption{The compiled flow graph $\mathcal{G}_{PT}$ of the example tree $\to(a,+(b,c),s)$
    (Fig.~\ref{fig:re-ocpt}). The $+$ node becomes an AND-split and AND-join ($+$) around the
    two concurrent typed lifecycles, $b$ on the order and $c$ on the item. The sequence links are
    degree-$(1,1)$ pass-throughs and collapse onto the edges.}
  \label{fig:re-pt-flow}
\end{figure*}

A silent leaf is needed exactly where a typed token crosses a span of the operator tree with
no observable activity on it. Where an observable leaf already lies on that span, none is
needed. The engine of Section~\ref{sec:general-framework} then runs unchanged, with one
instance-specific fact. The walk traverses the compiled flow graph $\mathcal{G}_{PT}$ rather
than the raw tree, keeping the branches of each XOR-split separate, crossing silent leaves
transparently, and stopping at the first activity on each side of a $\circlearrowleft$ back-edge,
so no unrolling is needed. Branches that leave the root end at the open ends $\bot$ and $\top$, and
start and end generators there are a choice of the analysis (Section~\ref{subsec:analysis}). Each context $(P,S)$ of a leaf $t$
yields a generator on the typed wires of $D_{PT}$. Each leg carries a count, by default pinned to
$n_p=1$ (one object per firing), and any other constraint on the leaf's own legs enters its
constraint system $\Lambda$ as a per-leg interval or a linear equality or inequality among its counts.
Generators sharing a typed wire compose along it in $\mathbf{F}(\Sigma)$ (Definition~\ref{def:free-category}).

\begin{theorem}[Canonical cospan-algebra presentation of OC process trees]
  \label{thm:ptree-to-cospan-signature}
  Let $PT$ be an OC process tree (Definition~\ref{def:process-tree}). Then $PT$ canonically determines a cospan-algebra signature $\Sigma$ in which each typed wire $(x,w,t)$ carries the set $w$ of object types $\omega$ with $(x,\omega,t)\in D_{PT}$ on its path, and hence the free symmetric monoidal category $\mathbf{F}(\Sigma)$.
\end{theorem}

\begin{proof}
  The schema of Section~\ref{sec:notation-by-notation} applies at the instance map above,
  finiteness of $T$ and of the operator tree $\Omega$ giving the conditions of Definition~\ref{def:lm-graph}, with
  each edge typed by the arc types in~$D_{PT}$. The delta concerns the loops. The walk runs on the
  compiled flow graph $\mathcal{G}_{PT}$ (Definition~\ref{def:pt-flow-graph}), which is finite, and a
  branch stops at any node already on its path, so every reach family and every set of contexts is
  finite (Definition~\ref{def:contexts-general}). A further pass of a $\circlearrowleft$ is a
  composite of generators already present and adds none, and the firing rule is self-dual across
  all four operators.
\end{proof}

\begin{figure*}[tp]
  \centering
  \input{tikz/running_example/re_pt_cospans}\unskip
  \caption{The generator cospans of Fig.~\ref{fig:re-ocpt}, each as a decorated cospan and the
    string diagram it denotes. $g_a$ takes the order and item threads from the open end $\bot$, $g_b,g_c$ carry them, and
    $g_s$ synchronises both.}
  \label{fig:re-ocpt-cospans}
\end{figure*}
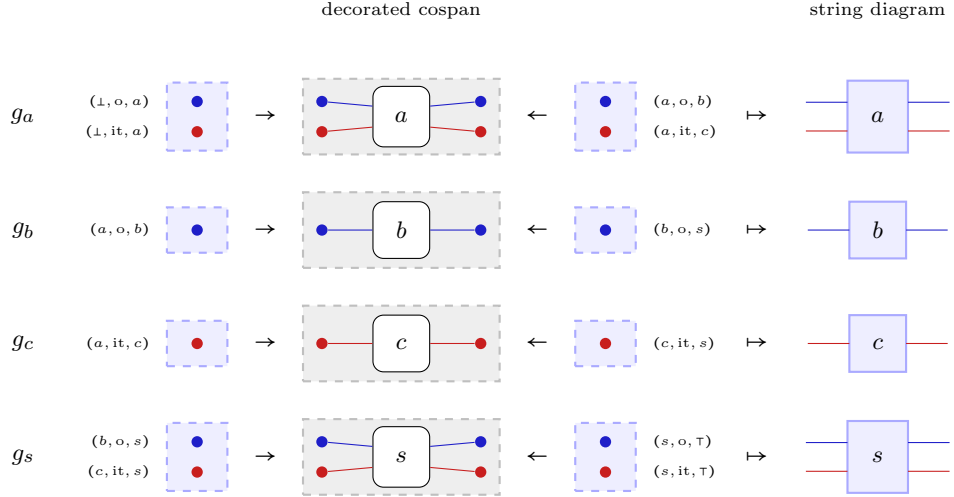

\begin{exmp}[A small OCPT and its generator cospans]
  \label{ex:running-ptree-cospan}
  Consider the OC process tree $\to(a, +(b,c), s)$ of Fig.~\ref{fig:re-ocpt}, compiled to the
  flow graph of Fig.~\ref{fig:re-pt-flow}, over the object types $\mathrm{order}$ and
  $\mathrm{item}$, with typed dependency arcs
  \[
    D_{PT} = \{(a,\mathrm{o},b),\,(a,\mathrm{it},c),\,
      (b,\mathrm{o},s),\,(c,\mathrm{it},s)\}.
  \]
  The $+$ node carries the object-centricity. It runs two concurrent typed
  object lifecycles, the order through $b$ and the item through $c$. Reading a
  generator off each leaf context (Fig.~\ref{fig:re-ocpt-cospans}) gives the first activity
  $g_a:\{(\bot,\mathrm{o},a),(\bot,\mathrm{it},a)\}\to\{(a,\mathrm{o},b),(a,\mathrm{it},c)\}$, the two lifecycle
  threads $g_b:\{(a,\mathrm{o},b)\}\to\{(b,\mathrm{o},s)\}$ and
  $g_c:\{(a,\mathrm{it},c)\}\to\{(c,\mathrm{it},s)\}$, and the synchronisation
  \begin{gather*}
    g_s = \Bigl(\{(b,\mathrm{o},s),(c,\mathrm{it},s)\}
      \xrightarrow{\iota} A \xleftarrow{\iota}\twocolbreak  \{(s,\mathrm{o},\top),(s,\mathrm{it},\top)\},\; d_s\Bigr),
  \end{gather*}
  whose left boundary carries both object types. Each typed wire here carries a single type,
  the order thread staying order-typed and the item thread item-typed, with $a$ and $s$ the only
  generators touching both. The two lifecycles are independent until $s$, which is exactly the
  concurrency the $+$ node expresses. Every leg here carries the default count of one.
  This concurrency is exactly what the certificate separates from exclusive choice (Lemma~\ref{lem:sig-complete}).
\end{exmp}

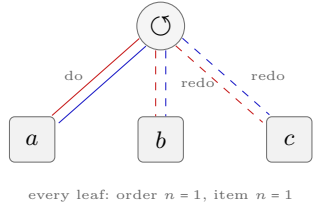
\begin{figure}[tb]
  \centering
  \input{tikz/running_example/re_loop_tree}\unskip
  \caption{The redo tree $\circlearrowleft(a,b,c)$, a mandatory body $a$ (solid ``do'' edge) and two
    redo bodies $b,c$ (dashed ``redo'' edges). Every leaf relates to both object types, order (blue) and
    item (red), one of each per firing, so each connector is two parallel coloured wires.}
  \label{fig:re-loop-tree}
\end{figure}

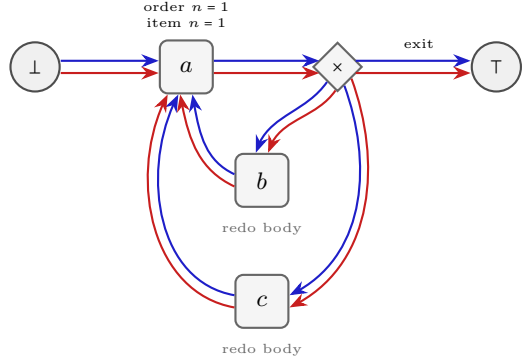
\begin{figure}[tb]
  \centering
  \input{tikz/running_example/re_loop}\unskip
  \caption{The compiled flow graph of $\circlearrowleft(a,b,c)$. After $a$, one XOR redo mediator
    either exits to the block after the loop or takes one of two back-edges, through redo body $b$ or
    through redo body $c$, and repeats $a$. Each connector carries both types as parallel coloured
    wires.}
  \label{fig:re-loop}
\end{figure}

\begin{figure*}[tp]
  \centering
  \input{tikz/running_example/re_pt_unroll}\unskip
  \caption{The redo $\circlearrowleft(a,b,c)$ unrolled to depths $0$, $1$, and $2$, both types
    threaded per pass, one order and one item. The depth-$2$ unroll takes
    both redo bodies once, $a\to b\to a\to c\to a$. The middle $a$ then sits between two distinct
    redo bodies, one of its contexts, and deeper unrolling composes only generators already present.}
  \label{fig:re-pt-unroll}
\end{figure*}
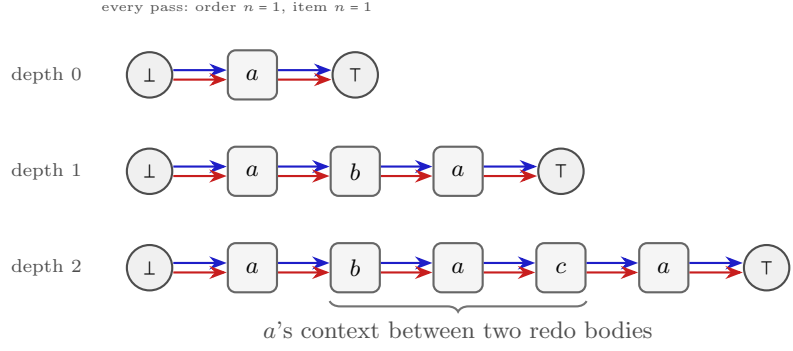

\begin{exmp}[A redo, unrolled to its cospans]
  \label{ex:running-ptree-loop}
  The examples above are loop-free. This one shows the construction on a cycle. Take the redo tree $\circlearrowleft(a,b,c)$ over two object types,
  $\mathrm{order}$ (blue) and $\mathrm{item}$ (red) (Fig.~\ref{fig:re-loop-tree}). The mandatory body
  $a$ runs once, then the redo either exits or takes one of the two redo bodies $b,c$ and runs $a$
  again, any number of times. Every leaf relates to both types, one order and one item per firing, the default
  multiplicity. The compiled flow graph
  (Definition~\ref{def:pt-flow-graph}, Fig.~\ref{fig:re-loop}) joins $a,b,c$ by one XOR redo mediator
  with two back-edges, through $b$ and through $c$.

  Unrolling the redo makes the finite generator set visible (Fig.~\ref{fig:re-pt-unroll}). The
  depth-$2$ unroll takes both redo bodies once, $a\to b\to a\to c\to a$. The middle $a$ then sits
  between two distinct redo bodies, a redo body on the input side and one on the output for each
  type, which is one of its contexts. Deeper unrolling only composes generators already present, so
  the loop reduces to finitely many generators, $g_b$, $g_c$ and one $g_a$ per pairing of a predecessor ($\bot$, $b$ or $c$) with a successor ($b$, $c$ or $\top$), which the walk from $a$ finds by
  stopping at the first activity past the redo mediator on each side.

  Reading a generator off each leaf context (Fig.~\ref{fig:re-loop-cospans}) gives the loop-entry
  $g_a$ and the two redo-body generators $g_b,g_c$, each threading one order and
  one item. Repeating the loop composes $g_a$ with the redo blocks
  $g_b\,g_a$ and $g_c\,g_a$ in any order and adds no signature element. The loop language is
  recovered at the composition level, and the typed redo needs no new machinery.
\end{exmp}

\begin{figure*}[tp]
  \centering
  \resizebox{\textwidth}{!}{\input{tikz/running_example/re_loop_cospans}\unskip}
  \caption{The generator cospans of $\circlearrowleft(a,b,c)$, each as a decorated cospan (left) and
    the string diagram it denotes (right), and the labels beside each boundary name its typed ports.
    Both legs, order (blue) and item
    (red), carry the default count of one, so no constraint system is drawn. $g_a$ is shown in its $b$-redo context, the
    $c$-context is analogous, and the entry and exit contexts pair $a$ with the open ends $\bot,\top$.}
  \label{fig:re-loop-cospans}
\end{figure*}
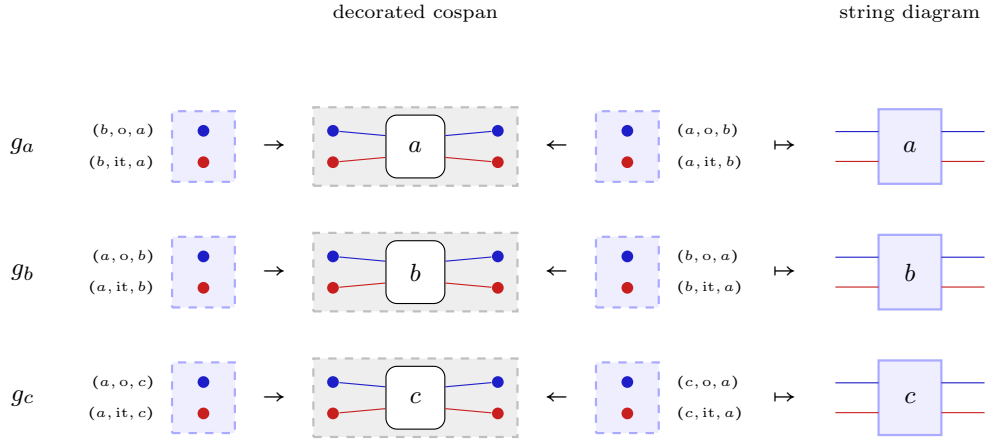

The sequence and parallel operators of Definition~\ref{def:process-tree} have an exact
characterisation in terms of series-parallel partial orders, developed in
Appendix~\ref{subsec:pt-sp-appendix} (Theorem~\ref{thm:pt-as-sp-language},
Corollary~\ref{cor:pt-sp-iff}).

%% file: tikz/running_example/re_pt_tree.tex
\newcommand{\pwire}[2]{%
  \draw[bluewire] let \p1=($(#2.center)-(#1.center)$), \n1={atan2(\y1,\x1)} in
    ($(#1.\n1)!1.9pt!90:(#2.center)$) -- ($(#2.\n1+180)!1.9pt!-90:(#1.center)$);
  \draw[redwire]  let \p1=($(#2.center)-(#1.center)$), \n1={atan2(\y1,\x1)} in
    ($(#1.\n1)!1.9pt!-90:(#2.center)$) -- ($(#2.\n1+180)!1.9pt!90:(#1.center)$);}
\begin{tikzpicture}[
    op/.style={circle, draw=black!55, fill=black!4, inner sep=1.2pt, minimum size=5mm,
               font=\small},
    lboth/.style={rectangle, rounded corners=2pt, draw=black!55, fill=black!5,
                  minimum size=6mm, font=\small},
    lord/.style={rectangle, rounded corners=2pt, draw=bluecol, fill=bluecol!12,
                 minimum size=6mm, font=\small},
    litem/.style={rectangle, rounded corners=2pt, draw=redcol, fill=redcol!12,
                  minimum size=6mm, font=\small}]
  \node[op]    (root) at (0,1.8)     {$\to$};
  \node[lboth] (a)    at (-1.6,0.6)  {$a$};
  \node[op]    (plus) at (0,0.6)     {$+$};
  \node[lboth] (s)    at (1.6,0.6)   {$s$};
  \node[lord]  (b)    at (-0.65,-0.6){$b$};
  \node[litem] (c)    at (0.65,-0.6) {$c$};
  \pwire{root}{a}
  \pwire{root}{plus}
  \pwire{root}{s}
  \draw[bluewire] (plus) -- (b);
  \draw[redwire]  (plus) -- (c);
\end{tikzpicture}

%% file: tikz/running_example/re_pt_flow.tex
\begin{tikzpicture}[petriBase, scale=0.95, inline,
    task/.style={rectangle, rounded corners=3pt, draw=black!60, fill=black!4,
                 minimum size=7mm, font=\small},
    btask/.style={task, draw=bluecol, fill=bluecol!6},
    rtask/.style={task, draw=redcol,  fill=redcol!6},
    agw/.style={diamond, draw=black!60, fill=black!4, inner sep=1pt, minimum size=6mm,
                font=\footnotesize}]
  \node[task]  (a)  at (0,0)      {$a$};
  \node[agw]   (ps) at (1.9,0)    {$+$};
  \node[btask] (b)  at (3.9,0.85) {$b$};
  \node[rtask] (c)  at (3.9,-0.85){$c$};
  \node[agw]   (qj) at (5.9,0)    {$+$};
  \node[task]  (s)  at (7.8,0)    {$s$};
  \draw[blueflow] ([yshift=1.4mm]a.east)  -- (1.74,0.14);
  \draw[redflow]  ([yshift=-1.4mm]a.east) -- (1.74,-0.14);
  \draw[blueflow] (ps) -- (b);
  \draw[redflow]  (ps) -- (c);
  \draw[blueflow] (b)  -- (qj);
  \draw[redflow]  (c)  -- (qj);
  \draw[blueflow] (6.06,0.14)  -- ([yshift=1.4mm]s.west);
  \draw[redflow]  (6.06,-0.14) -- ([yshift=-1.4mm]s.west);
\end{tikzpicture}

%% file: tikz/running_example/re_pt_cospans.tex
\begin{tikzpicture}[inline,
    bspd/.style={spider, fill=bluecol},
    rspd/.style={spider, fill=redcol},
    mbox/.style={draw=black, rounded corners=4pt, fill=white, minimum size=0.75cm, font=\small},
    plabel/.style={font=\tiny}]

  \node[font=\footnotesize] at (2.7,1.4) {decorated cospan};
  \node[font=\footnotesize] at (9.0,1.4) {string diagram};

  \node[font=\small] at (-2.3,0) {$g_a$};
  \node[cospanblue, minimum height=0.9cm] (aL) at (0,0) {};
  \node[bspd] at (0,0.20) {};
  \node[rspd] at (0,-0.20) {};
  \node at (0.9,0) {$\rightarrow$};
  \node[cospangrey, minimum height=1.0cm, minimum width=2.6cm] (aA) at (2.7,0) {};
  \node[mbox, minimum height=0.8cm] (af) at (2.7,0) {$a$};
  \node[bspd] (aao) at (1.65,0.20) {};
  \node[rspd] (aai) at (1.65,-0.20) {};
  \node[bspd] (abo) at (3.75,0.20) {};
  \node[rspd] (abi) at (3.75,-0.20) {};
  \draw[bluewire] (aao) -- (af.160);
  \draw[redwire]  (aai) -- (af.200);
  \draw[bluewire] (af.20)  -- (abo);
  \draw[redwire]  (af.-20) -- (abi);
  \node at (4.5,0) {$\leftarrow$};
  \node[cospanblue, minimum height=0.9cm] (aR) at (5.4,0) {};
  \node[bspd] at (5.4,0.20) {};
  \node[rspd] at (5.4,-0.20) {};
  \node[plabel, anchor=east] at ([shift={(-3pt,0.20cm)}]aL.west) {$(\bot,\mathrm{o},a)$};
  \node[plabel, anchor=east] at ([shift={(-3pt,-0.20cm)}]aL.west) {$(\bot,\mathrm{it},a)$};
  \node[plabel, anchor=west] at ([shift={(3pt,0.20cm)}]aR.east) {$(a,\mathrm{o},b)$};
  \node[plabel, anchor=west] at ([shift={(3pt,-0.20cm)}]aR.east) {$(a,\mathrm{it},c)$};
  \node at (7.4,0) {$\mapsto$};
  \node[morphismblue, minimum height=0.95cm] (asd) at (9.0,0) {$a$};
  \draw[bluewire] ($(asd.south west)!0.70!(asd.north west)$) -- ++(-0.55,0);
  \draw[redwire]  ($(asd.south west)!0.30!(asd.north west)$) -- ++(-0.55,0);
  \draw[bluewire] ($(asd.south east)!0.70!(asd.north east)$) -- ++(0.55,0);
  \draw[redwire]  ($(asd.south east)!0.30!(asd.north east)$) -- ++(0.55,0);

  \begin{scope}[yshift=-1.5cm]
    \node[font=\small] at (-2.3,0) {$g_b$};
    \node[cospanblue, minimum height=0.6cm] (bL) at (0,0) {};
    \node[bspd] at (0,0) {};
    \node at (0.9,0) {$\rightarrow$};
    \node[cospangrey, minimum height=1.0cm, minimum width=2.6cm] (bA) at (2.7,0) {};
    \node[mbox] (bf) at (2.7,0) {$b$};
    \node[bspd] (bao) at (1.65,0) {};
    \node[bspd] (bbo) at (3.75,0) {};
    \draw[bluewire] (bao) -- (bf.west);
    \draw[bluewire] (bf.east) -- (bbo);
    \node at (4.5,0) {$\leftarrow$};
    \node[cospanblue, minimum height=0.6cm] (bR) at (5.4,0) {};
    \node[bspd] at (5.4,0) {};
    \node[plabel, anchor=east] at ([shift={(-3pt,0cm)}]bL.west) {$(a,\mathrm{o},b)$};
    \node[plabel, anchor=west] at ([shift={(3pt,0cm)}]bR.east) {$(b,\mathrm{o},s)$};
    \node at (7.4,0) {$\mapsto$};
    \node[morphismblue, minimum size=0.75cm] (bsd) at (9.0,0) {$b$};
    \draw[bluewire] (bsd.west) -- ++(-0.55,0);
    \draw[bluewire] (bsd.east) -- ++(0.55,0);
  \end{scope}

  \begin{scope}[yshift=-3.0cm]
    \node[font=\small] at (-2.3,0) {$g_c$};
    \node[cospanblue, minimum height=0.6cm] (cL) at (0,0) {};
    \node[rspd] at (0,0) {};
    \node at (0.9,0) {$\rightarrow$};
    \node[cospangrey, minimum height=1.0cm, minimum width=2.6cm] (cA) at (2.7,0) {};
    \node[mbox] (cf) at (2.7,0) {$c$};
    \node[rspd] (cai) at (1.65,0) {};
    \node[rspd] (cbi) at (3.75,0) {};
    \draw[redwire] (cai) -- (cf.west);
    \draw[redwire] (cf.east) -- (cbi);
    \node at (4.5,0) {$\leftarrow$};
    \node[cospanblue, minimum height=0.6cm] (cR) at (5.4,0) {};
    \node[rspd] at (5.4,0) {};
    \node[plabel, anchor=east] at ([shift={(-3pt,0cm)}]cL.west) {$(a,\mathrm{it},c)$};
    \node[plabel, anchor=west] at ([shift={(3pt,0cm)}]cR.east) {$(c,\mathrm{it},s)$};
    \node at (7.4,0) {$\mapsto$};
    \node[morphismblue, minimum size=0.75cm] (csd) at (9.0,0) {$c$};
    \draw[redwire] (csd.west) -- ++(-0.55,0);
    \draw[redwire] (csd.east) -- ++(0.55,0);
  \end{scope}

  \begin{scope}[yshift=-4.5cm]
    \node[font=\small] at (-2.3,0) {$g_s$};
    \node[cospanblue, minimum height=0.9cm] (sL) at (0,0) {};
    \node[bspd] at (0,0.20) {};
    \node[rspd] at (0,-0.20) {};
    \node at (0.9,0) {$\rightarrow$};
    \node[cospangrey, minimum height=1.0cm, minimum width=2.6cm] (sA) at (2.7,0) {};
    \node[mbox, minimum height=0.8cm] (sf) at (2.7,0) {$s$};
    \node[bspd] (sao) at (1.65,0.20) {};
    \node[rspd] (sai) at (1.65,-0.20) {};
    \node[bspd] (sbo) at (3.75,0.20) {};
    \node[rspd] (sbi) at (3.75,-0.20) {};
    \draw[bluewire] (sao) -- (sf.160);
    \draw[redwire]  (sai) -- (sf.200);
    \draw[bluewire] (sf.20)  -- (sbo);
    \draw[redwire]  (sf.-20) -- (sbi);
    \node at (4.5,0) {$\leftarrow$};
    \node[cospanblue, minimum height=0.9cm] (sR) at (5.4,0) {};
    \node[bspd] at (5.4,0.20) {};
    \node[rspd] at (5.4,-0.20) {};
    \node[plabel, anchor=east] at ([shift={(-3pt,0.20cm)}]sL.west) {$(b,\mathrm{o},s)$};
    \node[plabel, anchor=east] at ([shift={(-3pt,-0.20cm)}]sL.west) {$(c,\mathrm{it},s)$};
    \node[plabel, anchor=west] at ([shift={(3pt,0.20cm)}]sR.east) {$(s,\mathrm{o},\top)$};
    \node[plabel, anchor=west] at ([shift={(3pt,-0.20cm)}]sR.east) {$(s,\mathrm{it},\top)$};
    \node at (7.4,0) {$\mapsto$};
    \node[morphismblue, minimum height=0.95cm] (ssd) at (9.0,0) {$s$};
    \draw[bluewire] ($(ssd.south west)!0.70!(ssd.north west)$) -- ++(-0.55,0);
    \draw[redwire]  ($(ssd.south west)!0.30!(ssd.north west)$) -- ++(-0.55,0);
    \draw[bluewire] ($(ssd.south east)!0.70!(ssd.north east)$) -- ++(0.55,0);
    \draw[redwire]  ($(ssd.south east)!0.30!(ssd.north east)$) -- ++(0.55,0);
  \end{scope}
\end{tikzpicture}

%% file: tikz/running_example/re_loop_tree.tex
\newcommand{\pwire}[3]{%
  \draw[bluewire,#1] let \p1=($(#3.center)-(#2.center)$), \n1={atan2(\y1,\x1)} in
    ($(#2.\n1)!1.9pt!90:(#3.center)$) -- ($(#3.\n1+180)!1.9pt!-90:(#2.center)$);
  \draw[redwire,#1]  let \p1=($(#3.center)-(#2.center)$), \n1={atan2(\y1,\x1)} in
    ($(#2.\n1)!1.9pt!-90:(#3.center)$) -- ($(#3.\n1+180)!1.9pt!90:(#2.center)$);}
\begin{tikzpicture}[
    op/.style={circle, draw=black!55, fill=black!4, inner sep=1.2pt, minimum size=5.5mm,
               font=\small},
    leaf/.style={rectangle, rounded corners=2pt, draw=black!55, fill=black!5,
                 minimum size=6mm, font=\small},
    elab/.style={font=\tiny, black!60}]
  \node[op]   (root) at (0,1.5)     {$\circlearrowleft$};
  \node[leaf] (a)    at (-1.7,0)    {$a$};
  \node[leaf] (b)    at (0,0)       {$b$};
  \node[leaf] (c)    at (1.7,0)     {$c$};
  \pwire{solid}{root}{a}
  \pwire{dashed}{root}{b}
  \pwire{dashed}{root}{c}
  \node[elab] at (-1.15,0.85) {do};
  \node[elab, anchor=west] at (0.14,0.75) {redo};
  \node[elab] at (1.42,0.85)  {redo};
  \node[font=\tiny, black!60, align=center] at (0,-0.75)
    {every leaf: order $n=1$, item $n=1$};
\end{tikzpicture}

%% file: tikz/running_example/re_loop.tex
\begin{tikzpicture}[petriBase, scale=1.0, inline,
    evt/.style={circle, draw=black!70, fill=black!6, minimum size=6.5mm, font=\small},
    task/.style={rectangle, rounded corners=3pt, draw=black!60, fill=black!4,
                 minimum size=7mm, font=\small},
    xgw/.style={diamond, draw=black!60, fill=black!4, inner sep=1pt, minimum size=6.5mm, font=\footnotesize}]
  \node[evt]  (g1) at (0,0)      {$\bot$};
  \node[task] (a)  at (2.0,0)    {$a$};
  \node[xgw]  (x)  at (4.0,0)    {$\times$};
  \node[evt]  (g2) at (6.1,0)    {$\top$};
  \node[task] (b)  at (3.0,-1.5) {$b$};
  \node[task] (c)  at (3.0,-3.1) {$c$};

  \draw[blueflow] ([yshift=2.3pt]g1.east) -- ([yshift=2.3pt]a.west);
  \draw[redflow]  ([yshift=-2.3pt]g1.east) -- ([yshift=-2.3pt]a.west);
  \draw[blueflow] ([yshift=2.3pt]a.east)  -- (3.76,0.081);
  \draw[redflow]  ([yshift=-2.3pt]a.east) -- (3.76,-0.081);
  \draw[blueflow] (4.24,0.081)  -- ([yshift=2.3pt]g2.west) node[pos=0.55, above=1pt, font=\tiny] {exit};
  \draw[redflow]  (4.24,-0.081) -- ([yshift=-2.3pt]g2.west);

  \draw[blueflow] (x.235) to[out=235,in=75]  ([xshift=-2.3pt]b.north);
  \draw[redflow]  (x.265) to[out=235,in=75]  ([xshift=2.3pt]b.north);
  \draw[blueflow] ([yshift=2.3pt]b.west)  to[out=155,in=-80] ([xshift=2.3pt]a.south);
  \draw[redflow]  ([yshift=-2.3pt]b.west) to[out=155,in=-80] ([xshift=-2.3pt]a.south);

  \draw[blueflow] (x.290) to[out=290,in=25]  ([yshift=2.3pt]c.east);
  \draw[redflow]  (x.320) to[out=290,in=25]  ([yshift=-2.3pt]c.east);
  \draw[blueflow] ([yshift=2.3pt]c.west)  to[out=170,in=-118] ([xshift=-3pt]a.south);
  \draw[redflow]  ([yshift=-2.3pt]c.west) to[out=170,in=-118] ([xshift=-7pt]a.south);

  \node[font=\tiny, align=center, above=1pt of a] {order $n=1$\\ item $n=1$};
  \node[font=\tiny, black!60, below=2pt of b] {redo body};
  \node[font=\tiny, black!60, below=2pt of c] {redo body};
\end{tikzpicture}

%% file: tikz/running_example/re_pt_unroll.tex
\newcommand{\dwire}[2]{%
  \draw[blueflow] ([yshift=2.1pt]#1.east) -- ([yshift=2.1pt]#2.west);%
  \draw[redflow]  ([yshift=-2.1pt]#1.east) -- ([yshift=-2.1pt]#2.west);}
\begin{tikzpicture}[petriBase, scale=0.85, inline,
    evt/.style={circle, draw=black!70, fill=black!6, minimum size=6mm, font=\small},
    task/.style={rectangle, rounded corners=3pt, draw=black!60, fill=black!4,
                 minimum size=6.5mm, font=\small, inner sep=2.5pt},
    dlab/.style={font=\footnotesize, anchor=east, black!70}]

  \node[font=\tiny, black!70, align=left, anchor=west] at (-0.9,1.05)
    {every pass: order $n=1$, item $n=1$};

  \node[dlab] at (-0.9,0) {depth $0$};
  \node[evt]  (a0) at (0,0)   {$\bot$};
  \node[task] (b0) at (1.6,0) {$a$};
  \node[evt]  (e0) at (3.2,0) {$\top$};
  \dwire{a0}{b0} \dwire{b0}{e0}

  \node[dlab] at (-0.9,-1.5) {depth $1$};
  \node[evt]  (a1) at (0,-1.5)   {$\bot$};
  \node[task] (b1) at (1.6,-1.5) {$a$};
  \node[task] (r1) at (3.2,-1.5) {$b$};
  \node[task] (c1) at (4.8,-1.5) {$a$};
  \node[evt]  (e1) at (6.4,-1.5) {$\top$};
  \dwire{a1}{b1} \dwire{b1}{r1} \dwire{r1}{c1} \dwire{c1}{e1}

  \node[dlab] at (-0.9,-3.0) {depth $2$};
  \node[evt]  (a2) at (0,-3.0)    {$\bot$};
  \node[task] (b2) at (1.6,-3.0)  {$a$};
  \node[task] (r2) at (3.2,-3.0)  {$b$};
  \node[task] (c2) at (4.8,-3.0)  {$a$};
  \node[task] (s2) at (6.4,-3.0)  {$c$};
  \node[task] (d2) at (8.0,-3.0)  {$a$};
  \node[evt]  (e2) at (9.6,-3.0)  {$\top$};
  \dwire{a2}{b2} \dwire{b2}{r2} \dwire{r2}{c2} \dwire{c2}{s2} \dwire{s2}{d2} \dwire{d2}{e2}
  \draw[decorate, decoration={brace, amplitude=4pt, mirror}, black!55]
    ($(r2.south west)+(0,-0.12)$) -- ($(s2.south east)+(0,-0.12)$)
    node[midway, below=4pt, font=\scriptsize, black!70]
    {$a$'s context between two redo bodies};
\end{tikzpicture}

%% file: tikz/running_example/re_loop_cospans.tex
\begin{tikzpicture}[inline,
    bspd/.style={spider, fill=bluecol},
    rspd/.style={spider, fill=redcol},
    mbox/.style={draw=black, rounded corners=4pt, fill=white, minimum size=0.75cm, font=\small},
    plabel/.style={font=\tiny}]

  \node[font=\footnotesize] at (2.7,1.7) {decorated cospan};
  \node[font=\footnotesize] at (9.0,1.7) {string diagram};

  \foreach \row/\g/\src/\tgt in {0/a/b/b, -1.6/b/a/a, -3.2/c/a/a}{%
    \begin{scope}[yshift=\row cm]
      \node[font=\small] at (-2.3,0) {$g_{\g}$};
      \node[cospanblue, minimum height=0.9cm] (L) at (0,0) {};
      \node[bspd] at (0,0.20) {};
      \node[rspd] at (0,-0.20) {};
      \node at (0.9,0) {$\rightarrow$};
      \node[cospangrey, minimum height=1.0cm, minimum width=2.6cm] (A) at (2.7,0) {};
      \node[mbox, minimum height=0.8cm] (f) at (2.7,0) {$\g$};
      \node[bspd] (ao) at (1.65,0.20) {};
      \node[rspd] (ai) at (1.65,-0.20) {};
      \node[bspd] (bo) at (3.75,0.20) {};
      \node[rspd] (bi) at (3.75,-0.20) {};
      \draw[bluewire] (ao) -- (f.160);
      \draw[redwire]  (ai) -- (f.200);
      \draw[bluewire] (f.20)  -- (bo);
      \draw[redwire]  (f.-20) -- (bi);
      \node at (4.5,0) {$\leftarrow$};
      \node[cospanblue, minimum height=0.9cm] (R) at (5.4,0) {};
      \node[bspd] at (5.4,0.20) {};
      \node[rspd] at (5.4,-0.20) {};
      \node[plabel, anchor=east] at ([shift={(-3pt,0.20cm)}]L.west) {$(\src,\mathrm{o},\g)$};
      \node[plabel, anchor=east] at ([shift={(-3pt,-0.20cm)}]L.west) {$(\src,\mathrm{it},\g)$};
      \node[plabel, anchor=west] at ([shift={(3pt,0.20cm)}]R.east) {$(\g,\mathrm{o},\tgt)$};
      \node[plabel, anchor=west] at ([shift={(3pt,-0.20cm)}]R.east) {$(\g,\mathrm{it},\tgt)$};
      \node at (7.4,0) {$\mapsto$};
      \node[morphismblue, minimum height=0.95cm] (sd) at (9.0,0) {$\g$};
      \draw[bluewire] ($(sd.south west)!0.70!(sd.north west)$) -- ++(-0.55,0);
      \draw[redwire]  ($(sd.south west)!0.30!(sd.north west)$) -- ++(-0.55,0);
      \draw[bluewire] ($(sd.south east)!0.70!(sd.north east)$) -- ++(0.55,0);
      \draw[redwire]  ($(sd.south east)!0.30!(sd.north east)$) -- ++(0.55,0);
    \end{scope}}
\end{tikzpicture}

%% file: sections/06-bpmn.tex
\FloatBarrier
\subsection{BPMN}\label{sec:bpmn}

Business Process Model and Notation (BPMN) is a widely-used graphical notation for
specifying business processes~\cite{BusinessProcessModel}, standardised as ISO/IEC 19510:2013.
BPMN is thus a first-class modelling language in its own right.
In the process-mining context the core elements are activities, XOR gateways (exclusive
choice), AND gateways (parallel split/join), OR gateways (inclusive choice/merge), and
sequence flows. We type BPMN's control flow directly, each sequence flow and gateway carrying an
object type, and give the resulting object-centric BPMN (OCBPMN) its cospan-algebra presentation, one further
instance of the general framework. Earlier object-centric treatments of BPMN take a different route
(Section~\ref{subsec:bpmn-related}).
Read at the level of models, this generalises the Kalenkova conversion between BPMN and Petri nets~\cite{kalenkovaProcessMiningUsing2017} to a single framework that covers four notations at once.

\begin{definition}[BPMN model]\label{def:bpmn}
  A BPMN model is a tuple
  \[
    B = (N, F, \ell, \mathcal{O}),
  \]
  where $N$ is a finite set of flow objects and $\mathcal{O}$ is a finite set of
  object types.
  Here $F \subseteq N \times \mathcal{O} \times N$ is the set of typed sequence flows,
  a typed multigraph in which an arc $(x,\omega,y)\in F$ records that an object of type
  $\omega\in\mathcal{O}$ flows from $x$ to $y$, mirroring the typed dependency multigraph
  $D$ of Definition~\ref{def:causal-net}, and
  $\ell : N \to \{\mathrm{activity}, \mathrm{event}, \mathrm{AND}, \mathrm{XOR}, \mathrm{OR}\}$
  assigns a BPMN node type to each flow object. A start event has no incoming flow and an end event
  has no outgoing flow.
  Write $A = \ell^{-1}(\mathrm{activity})$ for the activities and $G = N \setminus A$ for the
  gateways and events. Each flow carries an object type $\omega$, which is the edge typing the
  general framework uses. Gateways carry no
  object type of their own and are control-flow routing constructs
  (Remark~\ref{rem:bpmn-typed-gateways}).
  The control-flow graph of $B$ is the directed graph $(N, F)$.
\end{definition}

BPMN's OR gateways (inclusive choice) are the one construct without a direct AND/XOR
reading, and the least settled part of the standard's token semantics. The difficulty is the
OR-join, which must fire once exactly those incoming branches the matching split activated
have completed, in general a non-local decision that operational token semantics settle only
with global lookahead~\cite{dijkmanSemanticsAnalysisBusiness2008}. We sidestep the
operational question denotationally. Our concern is the structural (trace) language of a
model, what it \emph{allows} to happen, not whether a particular token run deadlocks, and
read this way an OR gateway is just the finite family of branch subsets it may activate,
each subset an ordinary AND-binding. The join then consumes exactly the subset the split
committed, so no global state is read, which is the reading object-centric causal nets and
Petri nets already use for inclusive choice.

\begin{definition}[OR-gateway semantics]\label{def:or-semantics}
  We read an OR gateway denotationally, as the set of branch subsets it allows.
  For a $k$-ary OR gateway with branch set $S$, write
  $\mathcal{P}^+(S) = \{ U \subseteq S \mid U \neq \emptyset \}$ for its $2^k - 1$
  non-empty subsets. The gateway denotes the inclusive-choice family indexed by
  $\mathcal{P}^+(S)$. An OR-split commits one $U \in \mathcal{P}^+(S)$ and activates exactly the
  branches in $U$ concurrently, and the matching OR-join fires once exactly the branches in $U$ have
  completed, the committed $U$ recorded as the gateway's binding. Each $U$ is thus an ordinary
  AND-binding and the gateway is their XOR. It is the object-centric causal-net activity with
  output-binding set $\mathrm{Out}(g)=\mathcal{P}^+(S)$, and in OCPN form the $2^k-1$ silent AND-transitions
  selected by an XOR.
\end{definition}

Because the committed $U$ is recorded as a binding, the OR-join reads exactly what the split wrote
and consults no global state. Replacing each OR gateway by its $\mathcal{P}^+(S)$ alternatives
(Definition~\ref{def:or-semantics}), each an XOR-selected AND-binding, leaves a finite AND/OR graph whose mediators are all tagged AND or XOR (Definition~\ref{def:lm-graph}), with the same signature. From there the construction is the
OCPN and OCCN one unchanged. The AND/OR graph is walked, and each activity's
gateway-mediated neighbourhoods are read off it.

This instantiates the general framework of Section~\ref{sec:general-framework}, and the two mapping
choices mirror the earlier sections. An activity is a generator. A gateway is a mediator, an AND
gateway an AND, an XOR gateway an XOR, and an OR gateway the XOR-of-ANDs above. A $1$-ary gateway is
the degree-$(1,1)$ mediator where AND, XOR and a SEQ pass-through all coincide, and simply forwards
its one branch. The gateways carry no generator of their own, exactly as the operators of the
process tree do not (Remark~\ref{rem:bpmn-typed-gateways}). Given
$B=(N,F,\ell,\mathcal{O})$, the instance map is as follows.
\medskip
\begin{center}
  \small\begin{tabularx}{\linewidth}{@{}lX@{}}
    \hline
    General framework & BPMN \\
    \hline
    AND/OR graph $(V,E)$ & the control-flow graph $(N,F)$ \\
    Activities $\mathcal{A}$ & tasks $A$ \\
    AND mediators & AND gateways, and the start and end events, where a walk ends at $\bot$ or $\top$ \\
    XOR mediators & XOR gateways, and OR gateways (exploded to XOR-of-ANDs over $\mathcal{P}^+(S)$) \\
    SEQ pass-throughs & $1$-ary gateways, where AND, XOR and SEQ coincide \\
    \hline
  \end{tabularx}
\end{center}
\medskip
The signature comes from $B$ as given. Inserting a type-preserving gateway of one input and one
output on any flow leaves every reached family, every context and hence every generator unchanged,
since a $1$-ary mediator forwards under both the AND and the XOR rule and wires are named by the
activities at their ends. The construction is therefore invariant under the normalisation of
\citet{dijkmanSemanticsAnalysisBusiness2008} without requiring it, which
matters because that normalisation is stated for their Petri-net mapping and leaves the OR-join
aside. Nothing here depends on resolving an OR-join. The signature enumerates the paths a model
permits and says nothing about how a running system would choose between them.

The engine of Section~\ref{sec:general-framework} runs unchanged, with one instance-specific
fact. Reachability is through the gateways only, an AND gateway joining one reached set from each
of its branches, since they occur together, and an XOR gateway keeping the reached sets of its
branches separate, since exactly one of them occurs. The context, generator, and composition steps
are those of Definitions~\ref{def:contexts-general} and~\ref{def:generator-cospan-general} on the
typed flows.

\begin{figure*}[tp]
  \centering
  \input{tikz/running_example/re_bpmn_example}\unskip
  \caption{A small object-centric BPMN model in canonical form, over order (blue)
    and item (red). The start emits an order and an item, and an AND-split gateway
    ($+$) routes them down parallel branches. An XOR gateway ($\times$) verifies
    ($b$) or expedites ($e$) the order while $c$ packs the item, and an AND-join gateway
    ($+$) synchronises them for shipping ($s$). Every gateway preserves object
    types, order flows staying order-typed and item flows item-typed.}
  \label{fig:re-obpmn}
\end{figure*}
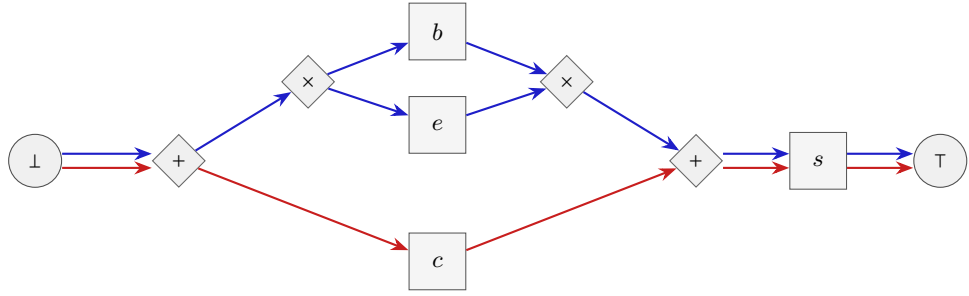

\begin{figure*}[tp]
  \centering
  \input{tikz/running_example/re_bpmn_cospans}\unskip
  \caption{The generator cospans of Fig.~\ref{fig:re-obpmn}. Gateways are
    mediators and yield no generators, while the activities do. The XOR gateway gives one
    generator per branch ($g_b$, $g_e$). The item thread is $g_c$, and $g_s$ is the
    object-centric synchronisation, its order port arriving from $b$ or $e$ (the
    XOR resolved upstream), so $s$ has two contexts.}
  \label{fig:re-obpmn-cospans}
\end{figure*}
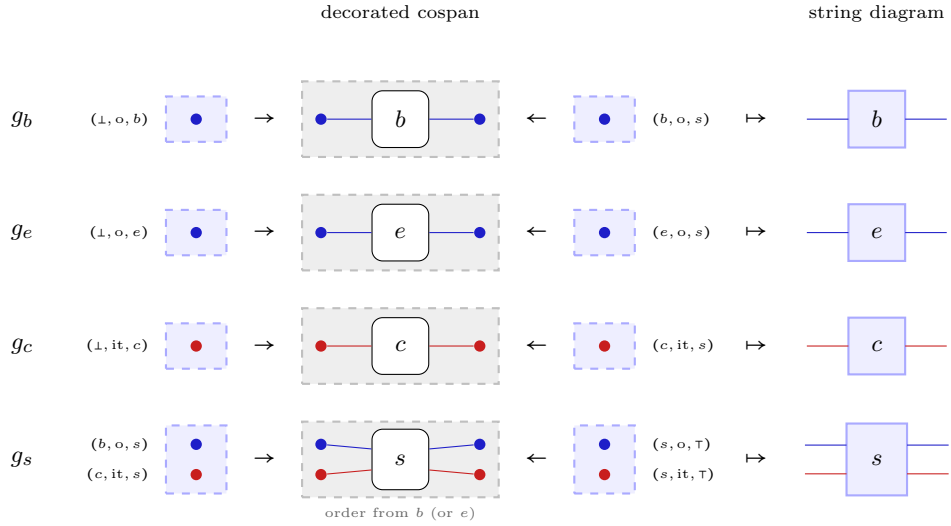

\begin{exmp}[A small OCBPMN and its generator cospans]
  \label{ex:running-bpmn-cospan}
  Consider the OCBPMN of Fig.~\ref{fig:re-obpmn}, over the object types
  $\mathrm{order}$ and $\mathrm{item}$. Every split and join is an explicit gateway,
  and each gateway preserves object types. The AND-split routes the order and item
  by type, the XOR offers two same-type alternatives for the order, and nothing
  creates or converts a type (Remark~\ref{rem:bpmn-typed-gateways}).

  The gateways are mediators, so only the activities yield generators
  (Fig.~\ref{fig:re-obpmn-cospans}). The XOR gateway gives one generator per branch,
  $g_b:\{(\bot,\mathrm{o},b)\}\to\{(b,\mathrm{o},s)\}$ and
  $g_e:\{(\bot,\mathrm{o},e)\}\to\{(e,\mathrm{o},s)\}$. The item thread is
  $g_c:\{(\bot,\mathrm{it},c)\}\to\{(c,\mathrm{it},s)\}$. The AND-join is
  the object-centric synchronisation, shown in its $b$-context,
  \begin{gather*}
    g_s = \Bigl(\{(b,\mathrm{o},s),(c,\mathrm{it},s)\}
      \xrightarrow{\iota} A \xleftarrow{\iota}\twocolbreak 
      \{(s,\mathrm{o},\top),(s,\mathrm{it},\top)\},\; d_s\Bigr).
  \end{gather*}
  Its left boundary carries both object types. The order arrives from $b$ or from $e$ (the XOR
  resolved upstream), so $s$ has two contexts, the $e$-context replacing $b$ by $e$. The AND
  gateways correspond to $s$'s AND-binding and the XOR gateway to the order's alternatives, matching
  the causal-net bindings of Section~\ref{sec:causal-nets}. Any count constraint (a key-distribution
  split, say) relates the leg counts of $s$, is recorded in its constraint system $\Lambda$
  (Definition~\ref{def:multiplicity-decoration}) and needs no further machinery.
\end{exmp}

\begin{figure*}[tp]
  \centering
  \resizebox{\textwidth}{!}{\input{tikz/running_example/re_bpmn_or}\unskip}
  \caption{(a) An OR gadget over $S=\{v,u\}$, order-typed. It denotes the XOR over the non-empty
    subsets $\mathcal{P}^+(S)=\{\{v\},\{u\},\{v,u\}\}$, the last an AND-binding. (b) The resulting
    three contexts of $s$, each a decorated cospan and the string diagram it denotes. The
    $\{v,u\}$ context gives $s$ two order ports.}
  \label{fig:re-bpmn-or}
\end{figure*}

\begin{exmp}[An OR gateway and its alternatives]
  \label{ex:bpmn-or}
  Suppose an order must undergo at least one of two optional checks, verify
  ($v$) and audit ($u$), before ship ($s$). An OR-split over the branch
  set $S=\{v,u\}$ feeds an OR-join into $s$ (Fig.~\ref{fig:re-bpmn-or}). By
  Definition~\ref{def:or-semantics} the gateway denotes the family indexed by
  $\mathcal{P}^+(S)=\bigl\{\,\{v\},\{u\},\{v,u\}\,\bigr\}$, so $s$ gains one context per committed
  subset, with left boundaries
  \[
    \{(v,\mathrm{o},s)\},\quad \{(u,\mathrm{o},s)\},\quad
    \{(v,\mathrm{o},s),(u,\mathrm{o},s)\},
  \]
  the last being the AND-binding in which both checks were taken. The branches
  contribute $g_v:\{(\bot,\mathrm{o},v)\}\to\{(v,\mathrm{o},s)\}$ and, likewise,
  $g_u$. The inclusive choice is thus an XOR over these three AND-bindings, exactly the
  $\mathcal{P}^+(S)$ explosion, and the OR-join consumes precisely the binding the split
  committed, so no global token state is consulted.
\end{exmp}

\begin{remark}[Typed gateways are routing-only]
  \label{rem:bpmn-typed-gateways}
  Every gateway is type-preserving, carrying the same set of object types on its incoming and
  outgoing flows and serving only to route objects. A gateway is a mediator, so it has no port of
  its own and the walk of Definition~\ref{def:contexts-general} passes through it. A gateway may not
  convert one type into another. Any count constraint on the routing (a key-distribution split, a
  batch cardinality) relates the leg counts of an activity it routes to and enters that activity's
  constraint system $\Lambda$ (Definition~\ref{def:multiplicity-decoration}) without further
  machinery.
\end{remark}

\begin{theorem}[Canonical cospan-algebra presentation of BPMN]\label{thm:bpmn-cospan}
  Let $B$ be any finite BPMN model.
  Then $B$ canonically determines a cospan-algebra signature $\Sigma$, and hence the free
  symmetric monoidal category $\mathbf{F}(\Sigma)$.
\end{theorem}

\begin{proof}
  The delta is the inclusive choice. By Definition~\ref{def:or-semantics} each OR gateway is the
  $\mathcal{P}^+(S)$ family of XOR-selected AND-bindings, a finite replacement that changes no
  signature, and afterwards every mediator is an AND or an XOR, a $1$-ary gateway being the
  degree-$(1,1)$ case, and an edge joining two activities directly is the degenerate case,
  traversed in one step (Definition~\ref{def:lm-graph}), and inserting a type-preserving $1$-ary
  gateway on such an edge changes no reached family and so no generator. The schema of Section~\ref{sec:notation-by-notation} then applies at the instance
  map above, the tasks being the activities, each edge carrying its flow type $\omega$
  (Definition~\ref{def:bpmn}), and finiteness of $N$ giving the conditions of
  Definition~\ref{def:lm-graph}. The mediators are the same AND/XOR kinds as the Petri- and
  causal-net instances, and the reach families are finite even under cyclic gateways, since the
  graph is finite and a branch stops at any node already on its path
  (Definition~\ref{def:contexts-general}).
\end{proof}

The presentation needs no structural precondition on $B$, and the OR-join, which operational token semantics settle only with global lookahead~\cite{dijkmanSemanticsAnalysisBusiness2008}, is handled denotationally.

%% file: tikz/running_example/re_bpmn_example.tex
\begin{tikzpicture}[scale=0.95, inline,
    evt/.style={circle, draw=black!70, fill=black!6, minimum size=7mm, font=\small},
    task/.style={rectangle, draw=black!65, fill=black!4, minimum size=7.5mm, font=\small},
    gw/.style={diamond, draw=black!60, fill=black!6, inner sep=1pt, minimum size=7mm, font=\small}]
  \node[evt]  (g1)  at (0,0)      {$\bot$};
  \node[gw]   (ag1) at (2.0,0)    {$+$};
  \node[gw]   (xg1) at (3.8,1.1)  {$\times$};
  \node[task] (b)   at (5.6,1.8)  {$b$};
  \node[task] (e)   at (5.6,0.5)  {$e$};
  \node[gw]   (xg2) at (7.4,1.1)  {$\times$};
  \node[task] (c)   at (5.6,-1.4) {$c$};
  \node[gw]   (ag2) at (9.2,0)    {$+$};
  \node[task] (s)   at (10.9,0)   {$s$};
  \node[evt]  (g2)  at (12.6,0)   {$\top$};

  \draw[blueflow] ([yshift=2.8pt]g1.east) -- ([yshift=2.8pt]ag1.west);
  \draw[redflow]  ([yshift=-2.8pt]g1.east) -- ([yshift=-2.8pt]ag1.west);
  \draw[blueflow] (ag1) -- (xg1);
  \draw[redflow]  (ag1) -- (c);
  \draw[blueflow] (xg1) -- (b);
  \draw[blueflow] (xg1) -- (e);
  \draw[blueflow] (b)   -- (xg2);
  \draw[blueflow] (e)   -- (xg2);
  \draw[blueflow] (xg2) -- (ag2);
  \draw[redflow]  (c)   -- (ag2);
  \draw[blueflow] ([yshift=2.8pt]ag2.east) -- ([yshift=2.8pt]s.west);
  \draw[redflow]  ([yshift=-2.8pt]ag2.east) -- ([yshift=-2.8pt]s.west);
  \draw[blueflow] ([yshift=2.8pt]s.east) -- ([yshift=2.8pt]g2.west);
  \draw[redflow]  ([yshift=-2.8pt]s.east) -- ([yshift=-2.8pt]g2.west);
\end{tikzpicture}

%% file: tikz/running_example/re_bpmn_cospans.tex
\begin{tikzpicture}[inline,
    bspd/.style={spider, fill=bluecol},
    rspd/.style={spider, fill=redcol},
    mbox/.style={draw=black, rounded corners=4pt, fill=white, minimum size=0.75cm, font=\small},
    plabel/.style={font=\tiny}]

  \node[font=\footnotesize] at (2.7,1.4) {decorated cospan};
  \node[font=\footnotesize] at (9.0,1.4) {string diagram};

  \foreach \row/\g in {0/b, -1.5/e}{%
    \begin{scope}[yshift=\row cm]
      \node[font=\small] at (-2.3,0) {$g_{\g}$};
      \node[cospanblue, minimum height=0.6cm] (L) at (0,0) {};
      \node[bspd] at (0,0) {};
      \node at (0.9,0) {$\rightarrow$};
      \node[cospangrey, minimum height=1.0cm, minimum width=2.6cm] (A) at (2.7,0) {};
      \node[mbox] (f) at (2.7,0) {$\g$};
      \node[bspd] (ai) at (1.65,0) {};
      \node[bspd] (ao) at (3.75,0) {};
      \draw[bluewire] (ai) -- (f.west);
      \draw[bluewire] (f.east) -- (ao);
      \node at (4.5,0) {$\leftarrow$};
      \node[cospanblue, minimum height=0.6cm] (R) at (5.4,0) {};
      \node[bspd] at (5.4,0) {};
      \node[plabel, anchor=east] at ([shift={(-3pt,0cm)}]L.west) {$(\bot,\mathrm{o},\g)$};
      \node[plabel, anchor=west] at ([shift={(3pt,0cm)}]R.east) {$(\g,\mathrm{o},s)$};
      \node at (7.4,0) {$\mapsto$};
      \node[morphismblue, minimum size=0.75cm] (sd) at (9.0,0) {$\g$};
      \draw[bluewire] (sd.west) -- ++(-0.55,0);
      \draw[bluewire] (sd.east) -- ++(0.55,0);
    \end{scope}}

  \begin{scope}[yshift=-3.0cm]
    \node[font=\small] at (-2.3,0) {$g_c$};
    \node[cospanblue, minimum height=0.6cm] (cL) at (0,0) {};
    \node[rspd] at (0,0) {};
    \node at (0.9,0) {$\rightarrow$};
    \node[cospangrey, minimum height=1.0cm, minimum width=2.6cm] (cA) at (2.7,0) {};
    \node[mbox] (cf) at (2.7,0) {$c$};
    \node[rspd] (cai) at (1.65,0) {};
    \node[rspd] (cao) at (3.75,0) {};
    \draw[redwire] (cai) -- (cf.west);
    \draw[redwire] (cf.east) -- (cao);
    \node at (4.5,0) {$\leftarrow$};
    \node[cospanblue, minimum height=0.6cm] (cR) at (5.4,0) {};
    \node[rspd] at (5.4,0) {};
    \node[plabel, anchor=east] at ([shift={(-3pt,0cm)}]cL.west) {$(\bot,\mathrm{it},c)$};
    \node[plabel, anchor=west] at ([shift={(3pt,0cm)}]cR.east) {$(c,\mathrm{it},s)$};
    \node at (7.4,0) {$\mapsto$};
    \node[morphismblue, minimum size=0.75cm] (csd) at (9.0,0) {$c$};
    \draw[redwire] (csd.west) -- ++(-0.55,0);
    \draw[redwire] (csd.east) -- ++(0.55,0);
  \end{scope}

  \begin{scope}[yshift=-4.5cm]
    \node[font=\small] at (-2.3,0) {$g_s$};
    \node[cospanblue, minimum height=0.9cm] (sL) at (0,0) {};
    \node[bspd] at (0,0.20) {};
    \node[rspd] at (0,-0.20) {};
    \node at (0.9,0) {$\rightarrow$};
    \node[cospangrey, minimum height=1.0cm, minimum width=2.6cm] (sA) at (2.7,0) {};
    \node[mbox, minimum height=0.8cm] (sf) at (2.7,0) {$s$};
    \node[bspd] (sao) at (1.65,0.20) {};
    \node[rspd] (sai) at (1.65,-0.20) {};
    \node[bspd] (sbo) at (3.75,0.20) {};
    \node[rspd] (sbi) at (3.75,-0.20) {};
    \draw[bluewire] (sao) -- (sf.160);
    \draw[redwire]  (sai) -- (sf.200);
    \draw[bluewire] (sf.20)  -- (sbo);
    \draw[redwire]  (sf.-20) -- (sbi);
    \node at (4.5,0) {$\leftarrow$};
    \node[cospanblue, minimum height=0.9cm] (sR) at (5.4,0) {};
    \node[bspd] at (5.4,0.20) {};
    \node[rspd] at (5.4,-0.20) {};
    \node[plabel, anchor=east] at ([shift={(-3pt,0.20cm)}]sL.west) {$(b,\mathrm{o},s)$};
    \node[plabel, anchor=east] at ([shift={(-3pt,-0.20cm)}]sL.west) {$(c,\mathrm{it},s)$};
    \node[plabel, anchor=west] at ([shift={(3pt,0.20cm)}]sR.east) {$(s,\mathrm{o},\top)$};
    \node[plabel, anchor=west] at ([shift={(3pt,-0.20cm)}]sR.east) {$(s,\mathrm{it},\top)$};
    \node at (7.4,0) {$\mapsto$};
    \node[morphismblue, minimum height=0.95cm] (ssd) at (9.0,0) {$s$};
    \draw[bluewire] ($(ssd.south west)!0.70!(ssd.north west)$) -- ++(-0.55,0);
    \draw[redwire]  ($(ssd.south west)!0.30!(ssd.north west)$) -- ++(-0.55,0);
    \draw[bluewire] ($(ssd.south east)!0.70!(ssd.north east)$) -- ++(0.55,0);
    \draw[redwire]  ($(ssd.south east)!0.30!(ssd.north east)$) -- ++(0.55,0);
    \node[plabel, black!60] at (2.7,-0.72) {order from $b$ (or $e$)};
  \end{scope}
\end{tikzpicture}

%% file: tikz/running_example/re_bpmn_or.tex
\begin{tikzpicture}[inline,
    evt/.style={circle, draw=black!70, fill=black!6, minimum size=6.5mm, font=\small},
    task/.style={rectangle, draw=black!65, fill=black!4, minimum size=7mm, font=\small},
    gw/.style={diamond, draw=black!60, fill=black!6, inner sep=0.5pt, minimum size=7mm,
               font=\scriptsize},
    bspd/.style={spider, fill=bluecol},
    mbox/.style={draw=black, rounded corners=4pt, fill=white, minimum size=0.75cm, font=\small},
    plabel/.style={font=\tiny}]

  \node[font=\footnotesize, anchor=west] at (-1.6,1.9) {(a) OCBPMN OR gadget};
  \node[evt]  (g1) at (0,0)     {$\bot$};
  \node[gw]   (os) at (1.5,0)   {OR};
  \node[task] (v)  at (3.0,0.7) {$v$};
  \node[task] (u)  at (3.0,-0.7){$u$};
  \node[gw]   (oj) at (4.5,0)   {OR};
  \node[task] (s)  at (6.0,0)   {$s$};
  \draw[blueflow] (g1) -- (os);
  \draw[blueflow] (os) -- (v);  \draw[blueflow] (os) -- (u);
  \draw[blueflow] (v)  -- (oj);  \draw[blueflow] (u) -- (oj);
  \draw[blueflow] (oj) -- (s);
  \node[font=\footnotesize, anchor=west] at (7.1,0.35)
    {denotes the XOR over $\mathcal{P}^+(\{v,u\})$};
  \node[font=\footnotesize, anchor=west] at (7.1,-0.35)
    {$=\{\{v\},\{u\},\{v,u\}\}$ (last is $\mathrm{AND}(v,u)$)};

  \node[font=\footnotesize, anchor=west] at (-1.6,-1.5) {(b) the three contexts of $s$};
  \begin{scope}[xshift=1.0cm]
  \node[font=\footnotesize] at (2.7,-2.1) {decorated cospan};
  \node[font=\footnotesize] at (9.0,-2.1) {string diagram};
  \foreach \row/\lab/\src in {-3.1/{U=\{v\}}/v, -4.6/{U=\{u\}}/u}{%
    \begin{scope}[yshift=\row cm]
      \node[font=\scriptsize] at (-2.3,0) {$\lab$};
      \node[cospanblue, minimum height=0.6cm] (L) at (0,0) {};
      \node[bspd] at (0,0) {};
      \node at (0.9,0) {$\rightarrow$};
      \node[cospangrey, minimum height=1.0cm, minimum width=2.6cm] (A) at (2.7,0) {};
      \node[mbox] (f) at (2.7,0) {$s$};
      \node[bspd] (ai) at (1.65,0) {};
      \node[bspd] (ao) at (3.75,0) {};
      \draw[bluewire] (ai) -- (f.west);
      \draw[bluewire] (f.east) -- (ao);
      \node at (4.5,0) {$\leftarrow$};
      \node[cospanblue, minimum height=0.6cm] (R) at (5.4,0) {};
      \node[bspd] at (5.4,0) {};
      \node[plabel, anchor=east] at ([shift={(-3pt,0cm)}]L.west) {$(\src,\mathrm{o},s)$};
      \node[plabel, anchor=west] at ([shift={(3pt,0cm)}]R.east) {$(s,\mathrm{o},\top)$};
      \node at (7.4,0) {$\mapsto$};
      \node[morphismblue, minimum size=0.75cm] (sd) at (9.0,0) {$s$};
      \draw[bluewire] (sd.west) -- ++(-0.55,0);
      \draw[bluewire] (sd.east) -- ++(0.55,0);
    \end{scope}}
  \begin{scope}[yshift=-6.1cm]
    \node[font=\scriptsize] at (-2.3,0) {$U=\{v,u\}$};
    \node[cospanblue, minimum height=0.9cm] (L) at (0,0) {};
    \node[bspd] at (0,0.20) {};
    \node[bspd] at (0,-0.20) {};
    \node at (0.9,0) {$\rightarrow$};
    \node[cospangrey, minimum height=1.0cm, minimum width=2.6cm] (A) at (2.7,0) {};
    \node[mbox, minimum height=0.8cm] (f) at (2.7,0) {$s$};
    \node[bspd] (ai1) at (1.65,0.20) {};
    \node[bspd] (ai2) at (1.65,-0.20) {};
    \node[bspd] (ao) at (3.75,0) {};
    \draw[bluewire] (ai1) -- (f.160);
    \draw[bluewire] (ai2) -- (f.200);
    \draw[bluewire] (f.east) -- (ao);
    \node at (4.5,0) {$\leftarrow$};
    \node[cospanblue, minimum height=0.6cm] (R) at (5.4,0) {};
    \node[bspd] at (5.4,0) {};
    \node[plabel, anchor=east] at ([shift={(-3pt,0.20cm)}]L.west) {$(v,\mathrm{o},s)$};
    \node[plabel, anchor=east] at ([shift={(-3pt,-0.20cm)}]L.west) {$(u,\mathrm{o},s)$};
    \node[plabel, anchor=west] at ([shift={(3pt,0cm)}]R.east) {$(s,\mathrm{o},\top)$};
    \node at (7.4,0) {$\mapsto$};
    \node[morphismblue, minimum height=0.95cm] (sd) at (9.0,0) {$s$};
    \draw[bluewire] ($(sd.south west)!0.70!(sd.north west)$) -- ++(-0.55,0);
    \draw[bluewire] ($(sd.south west)!0.30!(sd.north west)$) -- ++(-0.55,0);
    \draw[bluewire] (sd.east) -- ++(0.55,0);
  \end{scope}
  \end{scope}
\end{tikzpicture}

%% file: sections/07-validation.tex
\section{A ground-truth recovery study}\label{sec:validation}

We exercise the framework in two complementary ways. First we show that the four
worked examples together exercise the object-centric constructs the construction must absorb, one
notation at a time. Then we run the construction end to end on a constructed object-centric log,
discovering two notations from the same data and checking that both recover a known ground-truth
signature. The second part is the detailed case study, and it exhibits the notation-independence of
Section~\ref{sec:general-framework} on discovered models rather than on hand-built ones.

\subsection{Coverage across the notation space}\label{subsec:coverage}

Table~\ref{tab:coverage} records the object-centric construct each notation section
stresses. Every notation is object-typed throughout, and each adds one feature the construction
absorbs uniformly. Taken together the four examples cover typing, object multiplicity,
true concurrency, silent routing, iteration, and inclusive choice, and the same
construction reduces each to its cospan signature.

\begin{table*}[t]
  \centering
  \caption{Object-centric constructs exercised by the four worked examples. Each notation is
    reduced to its cospan signature by the single construction of
    Section~\ref{sec:general-framework}.}
  \label{tab:coverage}
  \begin{tabular}{l c c c c c c l}
    \hline
    & typing & mult. & conc. & silent & loop & OR & example \\
    \hline
    OCPN   & \checkmark &            & \checkmark & \checkmark &            &            & Ex.~\ref{ex:running-pn-cospan} \\
    OCCN   & \checkmark & \checkmark & \checkmark &            &            &            & Ex.~\ref{ex:running-cnet-cospan} \\
    OCPT   & \checkmark &            & \checkmark &            & \checkmark &            & Ex.~\ref{ex:running-ptree-cospan}, \ref{ex:running-ptree-loop} \\
    OCBPMN & \checkmark &            & \checkmark &            &            & \checkmark & Ex.~\ref{ex:running-bpmn-cospan}, \ref{ex:bpmn-or} \\
    \hline
  \end{tabular}
\end{table*}

\subsection{Two discoveries, nested signatures}\label{subsec:recovery}

The case study runs the construction on a single object-centric event log and asks
how the signatures of two different discovery algorithms, applied to the same data, relate to each
other and to the process that produced the data. We fix a known ground-truth model, generate a log
from it, discover an object-centric Petri net and an object-centric causal net, extract a signature
from each, and compare the signatures and their runs.

The ground-truth model is a chest-pain serial-troponin rule-out over three object types, a
patient thread and two test channels for troponin and ECG. A round examines the patient and reads one
troponin and one ECG. The round continues only while both come back clear, and it stops at a
disposition otherwise. We define the model directly as a signature $\Sigma_{\mathrm{GT}}$ in the form of
Section~\ref{sec:general-framework}, with no model to extract it from. Its 48 generators are written
from these rules, each leg carrying one object. A round sends the patient to a next step that agrees
with both results, each result goes on to an activity that reads it, and a case leaves through one
disposition. The model is object-centric by construction and carries a loop.

We keep the ground-truth model free of key-bound object-distribution relations, the
shared-key partitions that an object-centric causal net can express but an object-centric Petri net
cannot recover. This is deliberate. On a clinically meaningful process that stays within the
constructs both notations recover, we can ask whether two independent discoveries agree, rather than
picking an example whose behaviours the two notations could never match and where a disagreement
would say nothing.

We build the object-centric event log from $\Sigma_{\mathrm{GT}}$ by a selection of its closed
runs. A start generator places one patient, one lab and one imaging object before \emph{arrive}, an
end generator takes them after \emph{depart}, and the selected runs hold at most two
re-examinations. Every linear extension of every selected run becomes one case. Each leg carries one
object, so each object type forms one chain of typed wires through a run, and a case relates one
patient, one lab and one imaging object. The log is therefore constructed rather than observed, a
point we return to below. From the log we discover an object-centric Petri net with the
standard algorithm of \citet{vanderaalstDiscoveringObjectCentricPetri2020} and an object-centric
causal net with the object-centric causal-net miner of \citet{lissObjectCentricCausalNets2025}. Each discovered model is passed through the construction of
Section~\ref{sec:general-framework} to extract its cospan signature. As diagrams the two discovered
models look quite different (Figure~\ref{fig:ed-native-models}), and their signatures below
record exactly how they differ.

\begin{figure*}[tp]
  \centering
  \begin{subfigure}[b]{0.85\linewidth}
    \includegraphics[width=\linewidth]{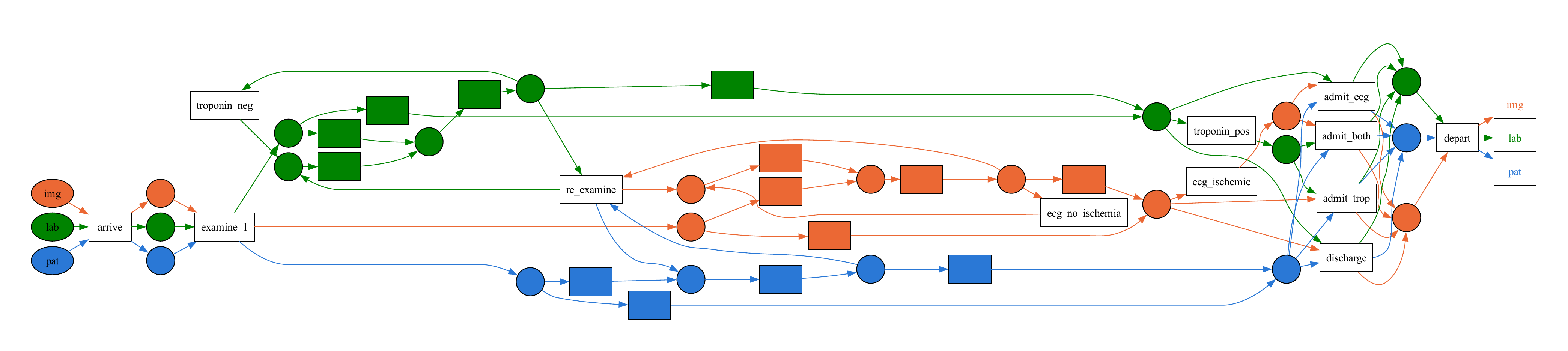}
    \caption{discovered object-centric Petri net}
  \end{subfigure}
  \\[1.5ex]
  \begin{subfigure}[b]{0.85\linewidth}
    \includegraphics[width=\linewidth]{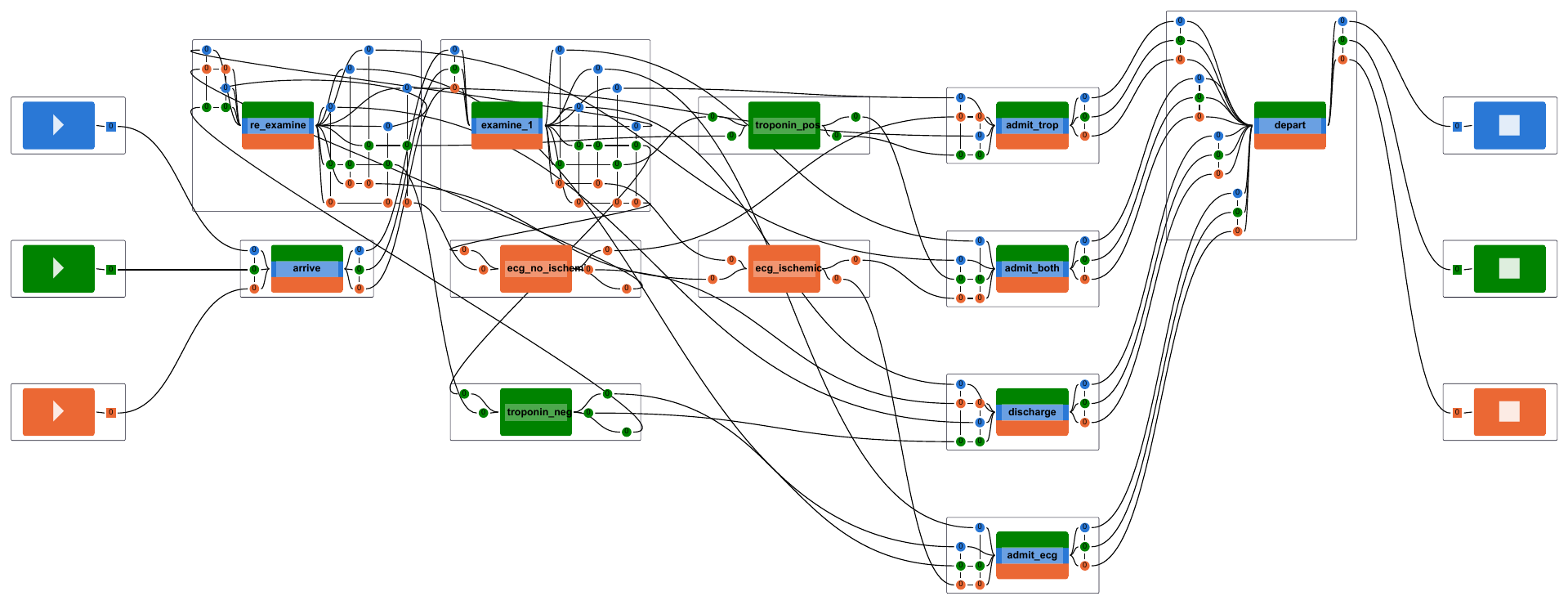}
    \caption{discovered object-centric causal net, in the reference interactive visualiser of
      \citet{lissObjectCentricCausalNets2025}}
  \end{subfigure}
  \caption{Models discovered from the same ED chest-pain log. Panel~(a) is an object-centric Petri
    net and panel~(b) an object-centric causal net, drawn in the reference interactive visualiser of
    Liss et al. The typed wires of the two extracted signatures are compared in
    Figure~\ref{fig:ed-wire-incidence}.}
  \label{fig:ed-native-models}
\end{figure*}

The extraction of Section~\ref{sec:general-framework} gives the causal net 48 generators and the
Petri net 2514. The causal-net signature is equal to $\Sigma_{\mathrm{GT}}$, generator for
generator, with the same labels, typed wires and constraint systems and the same start and end
generators. Every generator of $\Sigma_{\mathrm{GT}}$ is also equal to a generator of the Petri net,
so
\[
  \Sigma_{\mathrm{CN}}=\Sigma_{\mathrm{GT}}\subsetneq\Sigma_{\mathrm{PN}}.
\]
By Theorem~\ref{thm:canonical-presentation} the causal net has the language of the ground truth under
every selection, and by Corollary~\ref{cor:signature-inclusion} the Petri net's language contains it
under every monotone selection. The rest of this section shows
where the two differ, first on typed wires and generators and then on runs.

\begin{figure}[t]
  \centering
  \includegraphics[width=0.9\linewidth]{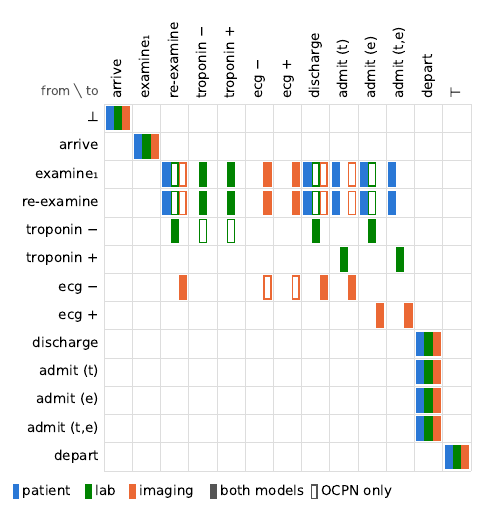}
  \caption{Typed wires of the two extracted signatures. A cell in row $a$ and column $b$ holds the
    typed wires $(a,w,b)$, coloured by object type. Solid cells occur in both signatures, 49 typed
    wires in all. The 16 hollow cells occur only in the object-centric Petri net, and each is a silent
    skip in one per-type net of the discovered Petri net.}
  \label{fig:ed-wire-incidence}
\end{figure}

Figure~\ref{fig:ed-wire-incidence} compares the typed wires. The causal net's 49 typed wires all
occur in the Petri net, which has 16 more. Each of the 16 is a silent skip in one of the per-type nets
the Petri-net discovery produces. The lab and imaging wires from \emph{examine} to
\emph{re-examine}, for instance, pass two silent transitions that route a test object past its test.

The wire chart is a projection of the signature and records which typed wires meet at an activity
and not which of them occur together in one generator. Every generator of either model takes one
typed wire per object type on each side, so choosing such a wire in every possible way bounds the
number of generators an activity can have. The Petri net meets
that bound at every activity, since its discovery builds one net per object type and joins them at
the activity, so the chart determines all 2514 of its generators. The causal net keeps 48 of the 153
combinations its own typed wires allow, because its bindings keep the object types of an activity
correlated. At \emph{re-examine} the chart allows 40 combinations and the causal net has 10, two input
bindings with five output bindings, since the next step of the patient must agree with both test
results. Of the Petri net's 2466 further generators, 2361 use a silent-skip wire and 105 combine
shared typed wires in ways the causal net never records.

The runs are compared under one selection. A start generator places one patient, one lab and one
imaging object before \emph{arrive}, and an end generator takes them after \emph{depart}. Both
signatures receive the same start and end generators, so the selection is monotone and the
inclusion of signatures carries over to their closed runs. The runs of both models are built from
the same pieces. A round holds some clear troponin and ECG results and ends in a re-examination,
\[
  R_{b,c}=\bigl((\mathrm{Tro}^{-})^{b}\otimes(\mathrm{ECG}^{-})^{c}\bigr)\,;\,\mathrm{Rex},
\]
where a power $f^{b}=f;\cdots;f$ is $b$ copies of $f$ composed in sequence, here $b$ clear troponin
results on the lab wire and $c$ clear ECG results on the imaging wire. A list $u=((b_1,c_1),\dots,(b_a,c_a))$ of any length $a\geq0$ gives
$R_u=R_{b_1,c_1};\cdots;R_{b_a,c_a}$, with $R_u$ the identity for the empty list. The order of the
list matters, since two rounds swapped give a different diagram.

Every closed run of either model is one of four classes, one per disposition $D$,
\[
  \mathrm{Arr}\,;\,\mathrm{Exa}_1\,;\,X\,;\,F_D\,;\,D\,;\,\mathrm{Dep},
\]
with $X$ and $F_D$ given in Table~\ref{tab:ed-run-classes}. The causal net's rounds are all clear
rounds $C_1=R_{1,1}$, so $X=C_1^{\,n}$ and a run has the single index $n$, and before its disposition
it holds exactly one result per test. These four classes are the closing diagrams of
Figure~\ref{fig:ed-decomposition} with $C_1$ inserted $n$ times. The Petri net takes any list $u$ and
any counts $b'$ and $c'$ in the final part, and Figure~\ref{fig:ed-pn-class} draws its class for
admission on the ECG.

\begin{table*}[t]
  \centering
  \caption{Run classes of the two discovered models, one per disposition $D$. Every closed run is
    $\mathrm{Arr};\mathrm{Exa}_1;X;F_D;D;\mathrm{Dep}$, with $X=C_1^{\,n}$ for the causal net and
    $X=R_u$ for the Petri net. A power is repeated sequential composition, and all indices range
    over $\mathbb{N}$.}
  \label{tab:ed-run-classes}
  \small
  \begin{tabular}{@{}l l l@{}}
    \hline
    $D$ & $F_D$, causal net & $F_D$, Petri net \\
    \hline
    Dis & $\mathrm{Tro}^{-}\otimes\mathrm{ECG}^{-}$ & $(\mathrm{Tro}^{-})^{b'}\otimes(\mathrm{ECG}^{-})^{c'}$ \\
    Adm$_{\mathrm{t}}$ & $\mathrm{Tro}^{+}\otimes\mathrm{ECG}^{-}$ & $((\mathrm{Tro}^{-})^{b'};\mathrm{Tro}^{+})\otimes(\mathrm{ECG}^{-})^{c'}$ \\
    Adm$_{\mathrm{e}}$ & $\mathrm{Tro}^{-}\otimes\mathrm{ECG}^{+}$ & $(\mathrm{Tro}^{-})^{b'}\otimes((\mathrm{ECG}^{-})^{c'};\mathrm{ECG}^{+})$ \\
    Adm$_{\mathrm{te}}$ & $\mathrm{Tro}^{+}\otimes\mathrm{ECG}^{+}$ & $((\mathrm{Tro}^{-})^{b'};\mathrm{Tro}^{+})\otimes((\mathrm{ECG}^{-})^{c'};\mathrm{ECG}^{+})$ \\
    \hline
  \end{tabular}
\end{table*}

\begin{figure*}[tp]
  \centering
  \begin{tabular}{@{}c@{\hspace{0.03\linewidth}}c@{}}
    \resizebox{0.48\linewidth}{!}{\input{tikz/ed_chest_pain/closing_1_discharge}\unskip} &
    \resizebox{0.48\linewidth}{!}{\input{tikz/ed_chest_pain/closing_2_admit_trop}\unskip} \\
    \resizebox{0.48\linewidth}{!}{\input{tikz/ed_chest_pain/closing_3_admit_ecg}\unskip} &
    \resizebox{0.48\linewidth}{!}{\input{tikz/ed_chest_pain/closing_4_admit_both}\unskip} \\
    \resizebox{0.48\linewidth}{!}{\input{tikz/ed_chest_pain/loop_L1}\unskip} &
    \resizebox{0.48\linewidth}{!}{\input{tikz/ed_chest_pain/legend}\unskip} \\
  \end{tabular}
  \caption{The four run classes of the causal net, drawn with one horizontal channel per object
    type and without the \emph{arrive} and \emph{depart} boxes. Each closing diagram is one
    disposition, and the clear round $C_1=R_{1,1}$ is inserted $n$ times at the dashed seam. The
    dashed seams at both ends of $C_1$ mark that same interface. Box labels are abbreviated, with the
    key beside $C_1$.}
  \label{fig:ed-decomposition}
\end{figure*}
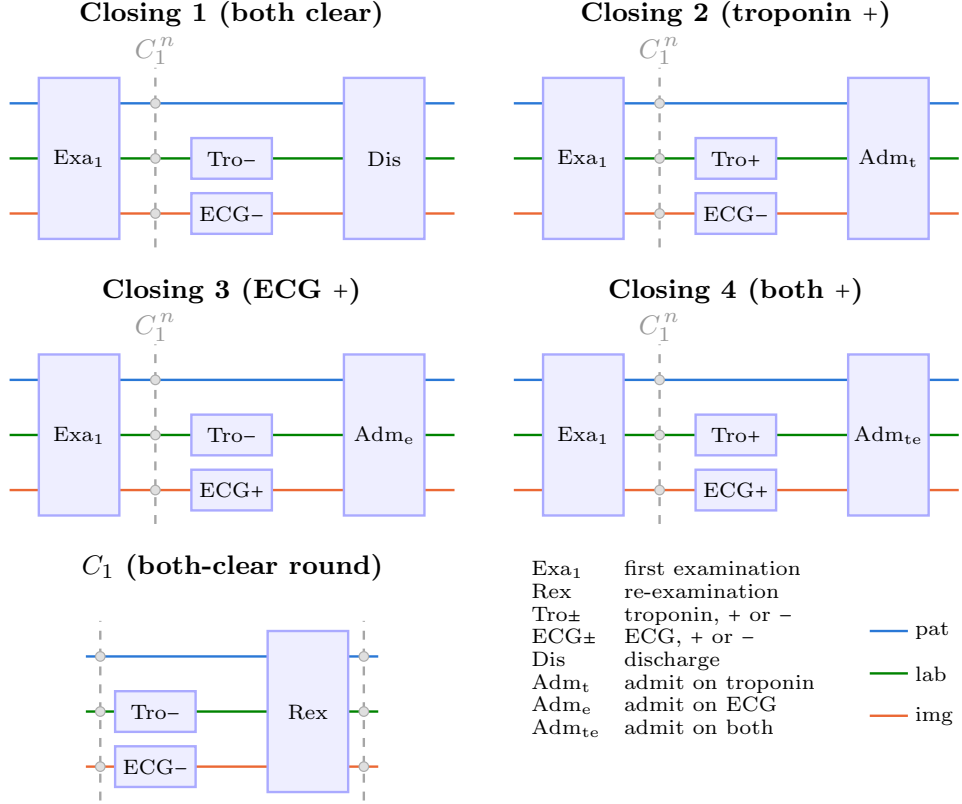

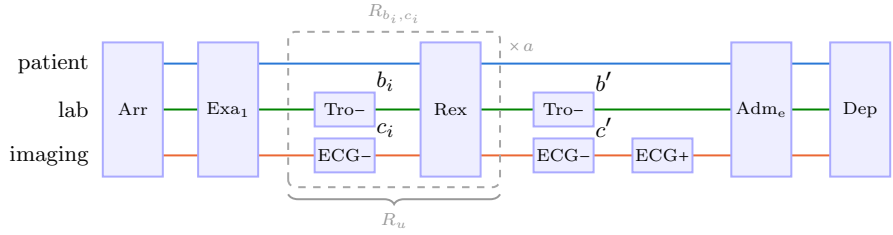
\begin{figure*}[tp]
  \centering
  \resizebox{0.9\linewidth}{!}{\input{tikz/ed_chest_pain/ocpn_class3_family}\unskip}
  \caption{The Petri net's run class for admission on the ECG. The dashed frame is one round
    $R_{b_i,c_i}$, and the frames for $i=1,\dots,a$ compose in order to $R_u$. A number beside a box
    counts its copies composed in sequence on that wire. The causal net's class for the same
    disposition sets $u=((1,1),\dots,(1,1))$, $b'=1$ and $c'=0$.}
  \label{fig:ed-pn-class}
\end{figure*}

A round of the Petri net may hold any number
of clear results on each test, including none, and the final part may repeat a clear result before
an abnormal one or hold no result for a test at all. Each of these uses a silent-skip wire of
Figure~\ref{fig:ed-wire-incidence}. A round without tests joins \emph{examine} or \emph{re-examine} to
\emph{re-examine} on the lab and imaging wires, a repeated result joins a clear test result to itself
or to the abnormal result of the same test, and a
disposition without a test joins \emph{examine} or \emph{re-examine} to the disposition. The two
abnormal results of an admission on both tests lie on different wires and are concurrent in one
diagram, so the two orders in which they can be logged are two linear extensions of one run and not
two classes. With every index at zero both models have four base runs. The causal net's are its four
closing diagrams. Three of the Petri net's omit a clear result the causal net requires, and only
admission on both tests keeps both.

The classes are exact. Enumerating every closed run of up to twelve occurrences, each run of the
Petri net matches one index tuple of one class, and each tuple in that bound is a run, 4180, 1596,
1596 and 609 runs for the four dispositions. The causal net's twelve runs are its classes with
$n\leq2$. The minimal automata of the two languages, with 13 and 8 states, decide the inclusion
without a length bound. The causal net's language is strictly included in the Petri net's, and the
shortest word only the Petri net admits is \emph{arrive examine discharge depart}. The causal net
therefore recovers the ground-truth signature exactly, and the Petri net's signature and language
strictly contain it.

\subsection{Applicability and limitations}\label{subsec:threats}

The extraction and comparison are implemented. A discovery output in a standard
interchange format is read into a signature, and the certificate of
Section~\ref{sec:conversion-framework} decides structural equivalence by set-equality of the two
signatures, which is convention-free and strictly finer than trace equivalence by
Lemma~\ref{lem:equiv-strictness}. Set-inclusion of signatures decides the inclusion used here.

The statement this supports runs in one direction. Equal signatures give equal trace languages,
and included signatures give included trace languages under a monotone selection. Models with
different signatures are structurally different, and whether they nonetheless agree on their trace
languages is a separate question, which the automata settle here. An inclusion of signatures, as
found here, is therefore a stronger statement than an inclusion of trace languages would be.

\paragraph{Contents of the log.}
The log is a complete enumeration and not a sample. It holds 576 events over 168 objects under
twelve event labels, one case for every linear extension of every closed run of
$\Sigma_{\mathrm{GT}}$ with at most two re-examinations. The run-space is a disposition with
$n$ copies of one clear round, the $C_1$ of Figure~\ref{fig:ed-decomposition}, inserted after the
first examination, so a single
integer $n$ characterises every unrolling and nested loops would need a set of numbers. The depth realised in the log is a free parameter of the enumeration rather than a
property of the model, and the recovery does not depend on it, because the run classes of
Table~\ref{tab:ed-run-classes} hold for every $n$. Fixing $n$ from data is an
estimation question outside the scope of this paper. The log has to be complete, because a
recovery measurement asks whether discovery returns the model that produced the log, and that
question only makes sense when the log realises that model in full. The result should be read as
establishing that the extraction is well defined and non-vacuous on discovery output, and not as
a measurement of miner robustness under noise or incompleteness. The log records the start and
end of each case as the activities \emph{arrive} and \emph{depart}, which both signatures contain.
The start and end generators of the run comparison are separate from them, belong to
$\Sigma_\gamma$ (Section~\ref{subsec:analysis}) and are stated as its assumption.

\paragraph{Reasons for a constructed log.}
The study measures recovery, which compares each discovered signature against the signature of
the model that produced the log. A ground-truth model is therefore a precondition of the
measurement rather than a property of the data that one could go looking for. No log observed in
a running information system supplies one, because the only candidate target available for such a
log is itself the output of a discovery algorithm, and scoring a discovery against a discovery
measures agreement between two miners instead of recovery of a known model. This bounds what any
public data could contribute and is not a gap in our search.

The simulated public logs do not close the gap either. The object-centric order-management log
of the OCEL 2.0 collection is generated from a coloured Petri net, and a coloured Petri net types
its tokens by colour set rather than by object type. OCEL 2.0 and coloured Petri nets are not
intrinsically coupled, and we draw this observation only for the order-management log we examined.
Reading such a model as an object-centric one fixes a
correspondence between colour sets and object types that the published artefact leaves open, so
the object structure would be ours and not the dataset's. In any case the object-centric Petri nets published
alongside these logs are obtained from the logs by discovery, and are therefore targets of the
same kind as the models under comparison.

Everything the construction does not decide is left to the two miners. Both are independently authored
third-party implementations, neither was written with cospan signatures in mind, and each targets
a different notation. The causal net's recovery of the ground-truth signature, and the inclusion of
that signature in the Petri net's, are therefore not artefacts of how the log was generated. A construction reading the generator
rather than the miners' output would agree with itself.

The absence of a ground truth in observed data is also why comparing signatures by equality and
inclusion is worth having. Where no reference model exists, two models can still be compared against each
other, and that comparison is the subject of Section~\ref{sec:conversion-framework}.

Three limitations bound these claims. The log is noise-free and enumerates the behaviour of the
ground-truth model completely, so nothing here measures what the extraction does under incomplete
traces, mislabelled activities, or missing object relations. The four worked examples are
deliberately minimal. The reverse maps of Section~\ref{sec:conversion-framework} are
one-directional. The implementation is released as the \texttt{proc-posets} package alongside the
paper.

%% file: tikz/ed_chest_pain/closing_1_discharge.tex
\begin{tikzpicture}[
    act/.style={morphismblue, font=\footnotesize, inner sep=1pt},
    patw/.style={thick, draw=typepat},
    labw/.style={thick, draw=typelab},
    imgw/.style={thick, draw=typeimg},
    seam/.style={dashed, thick, gray!70},
    seamdot/.style={circle, draw=gray!65, fill=gray!25, minimum size=3.2pt, inner sep=0pt},
    ptitle/.style={font=\small\bfseries, inner sep=2pt}]

  \useasboundingbox (-0.15,-0.55) rectangle (5.45,2.6);
  \begin{scope}[xshift=0.00cm]

  \draw[patw] (0,1.30) -- (5.25,1.30);
  \draw[labw] (0,0.65) -- (5.25,0.65);
  \draw[imgw] (0,0.00) -- (5.25,0.00);

  \node[act, minimum width=0.95cm, minimum height=1.9cm] at (0.82,0.65) {Exa$_1$};
  \node[act, minimum width=0.95cm, minimum height=0.48cm] at (2.62,0.65) {Tro$-$};
  \node[act, minimum width=0.95cm, minimum height=0.48cm] at (2.62,0.00) {ECG$-$};
  \node[act, minimum width=0.95cm, minimum height=1.9cm] at (4.42,0.65) {Dis};

  \draw[seam] (1.72,-0.40) -- (1.72,1.72);
  \foreach \y in {1.30,0.65,0.00} \node[seamdot] at (1.72,\y) {};
  \node[font=\small, gray!80, anchor=south, inner sep=1pt] at (1.72,1.72) {$C_1^{\,n}$};

  \node[ptitle] at (2.62,2.35) {Closing 1 (both clear)};
  \end{scope}
\end{tikzpicture}

%% file: tikz/ed_chest_pain/closing_2_admit_trop.tex
\begin{tikzpicture}[
    act/.style={morphismblue, font=\footnotesize, inner sep=1pt},
    patw/.style={thick, draw=typepat},
    labw/.style={thick, draw=typelab},
    imgw/.style={thick, draw=typeimg},
    seam/.style={dashed, thick, gray!70},
    seamdot/.style={circle, draw=gray!65, fill=gray!25, minimum size=3.2pt, inner sep=0pt},
    ptitle/.style={font=\small\bfseries, inner sep=2pt}]

  \useasboundingbox (-0.15,-0.55) rectangle (5.45,2.6);
  \begin{scope}[xshift=0.00cm]

  \draw[patw] (0,1.30) -- (5.25,1.30);
  \draw[labw] (0,0.65) -- (5.25,0.65);
  \draw[imgw] (0,0.00) -- (5.25,0.00);

  \node[act, minimum width=0.95cm, minimum height=1.9cm] at (0.82,0.65) {Exa$_1$};
  \node[act, minimum width=0.95cm, minimum height=0.48cm] at (2.62,0.65) {Tro$+$};
  \node[act, minimum width=0.95cm, minimum height=0.48cm] at (2.62,0.00) {ECG$-$};
  \node[act, minimum width=0.95cm, minimum height=1.9cm] at (4.42,0.65) {Adm$_{\mathrm{t}}$};

  \draw[seam] (1.72,-0.40) -- (1.72,1.72);
  \foreach \y in {1.30,0.65,0.00} \node[seamdot] at (1.72,\y) {};
  \node[font=\small, gray!80, anchor=south, inner sep=1pt] at (1.72,1.72) {$C_1^{\,n}$};

  \node[ptitle] at (2.62,2.35) {Closing 2 (troponin $+$)};
  \end{scope}
\end{tikzpicture}

%% file: tikz/ed_chest_pain/closing_3_admit_ecg.tex
\begin{tikzpicture}[
    act/.style={morphismblue, font=\footnotesize, inner sep=1pt},
    patw/.style={thick, draw=typepat},
    labw/.style={thick, draw=typelab},
    imgw/.style={thick, draw=typeimg},
    seam/.style={dashed, thick, gray!70},
    seamdot/.style={circle, draw=gray!65, fill=gray!25, minimum size=3.2pt, inner sep=0pt},
    ptitle/.style={font=\small\bfseries, inner sep=2pt}]

  \useasboundingbox (-0.15,-0.55) rectangle (5.45,2.6);
  \begin{scope}[xshift=0.00cm]

  \draw[patw] (0,1.30) -- (5.25,1.30);
  \draw[labw] (0,0.65) -- (5.25,0.65);
  \draw[imgw] (0,0.00) -- (5.25,0.00);

  \node[act, minimum width=0.95cm, minimum height=1.9cm] at (0.82,0.65) {Exa$_1$};
  \node[act, minimum width=0.95cm, minimum height=0.48cm] at (2.62,0.65) {Tro$-$};
  \node[act, minimum width=0.95cm, minimum height=0.48cm] at (2.62,0.00) {ECG$+$};
  \node[act, minimum width=0.95cm, minimum height=1.9cm] at (4.42,0.65) {Adm$_{\mathrm{e}}$};

  \draw[seam] (1.72,-0.40) -- (1.72,1.72);
  \foreach \y in {1.30,0.65,0.00} \node[seamdot] at (1.72,\y) {};
  \node[font=\small, gray!80, anchor=south, inner sep=1pt] at (1.72,1.72) {$C_1^{\,n}$};

  \node[ptitle] at (2.62,2.35) {Closing 3 (ECG $+$)};
  \end{scope}
\end{tikzpicture}

%% file: tikz/ed_chest_pain/closing_4_admit_both.tex
\begin{tikzpicture}[
    act/.style={morphismblue, font=\footnotesize, inner sep=1pt},
    patw/.style={thick, draw=typepat},
    labw/.style={thick, draw=typelab},
    imgw/.style={thick, draw=typeimg},
    seam/.style={dashed, thick, gray!70},
    seamdot/.style={circle, draw=gray!65, fill=gray!25, minimum size=3.2pt, inner sep=0pt},
    ptitle/.style={font=\small\bfseries, inner sep=2pt}]

  \useasboundingbox (-0.15,-0.55) rectangle (5.45,2.6);
  \begin{scope}[xshift=0.00cm]

  \draw[patw] (0,1.30) -- (5.25,1.30);
  \draw[labw] (0,0.65) -- (5.25,0.65);
  \draw[imgw] (0,0.00) -- (5.25,0.00);

  \node[act, minimum width=0.95cm, minimum height=1.9cm] at (0.82,0.65) {Exa$_1$};
  \node[act, minimum width=0.95cm, minimum height=0.48cm] at (2.62,0.65) {Tro$+$};
  \node[act, minimum width=0.95cm, minimum height=0.48cm] at (2.62,0.00) {ECG$+$};
  \node[act, minimum width=0.95cm, minimum height=1.9cm] at (4.42,0.65) {Adm$_{\mathrm{te}}$};

  \draw[seam] (1.72,-0.40) -- (1.72,1.72);
  \foreach \y in {1.30,0.65,0.00} \node[seamdot] at (1.72,\y) {};
  \node[font=\small, gray!80, anchor=south, inner sep=1pt] at (1.72,1.72) {$C_1^{\,n}$};

  \node[ptitle] at (2.62,2.35) {Closing 4 (both $+$)};
  \end{scope}
\end{tikzpicture}

%% file: tikz/ed_chest_pain/loop_L1.tex
\begin{tikzpicture}[
    act/.style={morphismblue, font=\footnotesize, inner sep=1pt},
    patw/.style={thick, draw=typepat},
    labw/.style={thick, draw=typelab},
    imgw/.style={thick, draw=typeimg},
    seam/.style={dashed, thick, gray!70},
    seamdot/.style={circle, draw=gray!65, fill=gray!25, minimum size=3.2pt, inner sep=0pt},
    ptitle/.style={font=\small\bfseries, inner sep=2pt}]

  \useasboundingbox (-0.15,-0.55) rectangle (5.45,2.6);
  \begin{scope}[xshift=0.90cm]

  \draw[patw] (0,1.30) -- (3.45,1.30);
  \draw[labw] (0,0.65) -- (3.45,0.65);
  \draw[imgw] (0,0.00) -- (3.45,0.00);

  \node[act, minimum width=0.95cm, minimum height=0.48cm] at (0.82,0.65) {Tro$-$};
  \node[act, minimum width=0.95cm, minimum height=0.48cm] at (0.82,0.00) {ECG$-$};
  \node[act, minimum width=0.95cm, minimum height=1.9cm] at (2.62,0.65) {Rex};

  \draw[seam] (0.17,-0.40) -- (0.17,1.72);
  \foreach \y in {1.30,0.65,0.00} \node[seamdot] at (0.17,\y) {};
  \draw[seam] (3.28,-0.40) -- (3.28,1.72);
  \foreach \y in {1.30,0.65,0.00} \node[seamdot] at (3.28,\y) {};

  \node[ptitle] at (1.73,2.35) {$C_1$ (both-clear round)};
  \end{scope}
\end{tikzpicture}

%% file: tikz/ed_chest_pain/legend.tex
\begin{tikzpicture}[
    act/.style={morphismblue, font=\footnotesize, inner sep=1pt},
    patw/.style={thick, draw=typepat},
    labw/.style={thick, draw=typelab},
    imgw/.style={thick, draw=typeimg},
    seam/.style={dashed, thick, gray!70},
    seamdot/.style={circle, draw=gray!65, fill=gray!25, minimum size=3.2pt, inner sep=0pt},
    ptitle/.style={font=\small\bfseries, inner sep=2pt}]
  \useasboundingbox (-0.15,-0.55) rectangle (5.45,2.6);

  \node[anchor=north west, font=\footnotesize, inner sep=0pt] at (0.20,2.45) {%
    \renewcommand{\arraystretch}{0.95}%
    \begin{tabular}{@{}l@{\quad}l@{}}
      Exa$_1$ & first examination\\
      Rex & re-examination\\
      Tro$\pm$ & troponin, $+$ or $-$\\
      ECG$\pm$ & ECG, $+$ or $-$\\
      Dis & discharge\\
      Adm$_{\mathrm{t}}$ & admit on troponin\\
      Adm$_{\mathrm{e}}$ & admit on ECG\\
      Adm$_{\mathrm{te}}$ & admit on both
    \end{tabular}};

  \draw[patw] (4.20,1.60) -- (4.65,1.60) node[anchor=west, font=\footnotesize, inner sep=2pt] {pat};
  \draw[labw] (4.20,1.10) -- (4.65,1.10) node[anchor=west, font=\footnotesize, inner sep=2pt] {lab};
  \draw[imgw] (4.20,0.60) -- (4.65,0.60) node[anchor=west, font=\footnotesize, inner sep=2pt] {img};
\end{tikzpicture}

%% file: tikz/ed_chest_pain/ocpn_class3_family.tex
\begin{tikzpicture}[
    act/.style={morphismblue, font=\footnotesize, inner sep=1pt},
    patw/.style={thick, draw=typepat},
    labw/.style={thick, draw=typelab},
    imgw/.style={thick, draw=typeimg},
    pw/.style={font=\scriptsize, inner sep=0.5pt, anchor=south west},
    round/.style={dashed, thick, gray!70, rounded corners=3pt}]

  \draw[patw] (0.55,1.30) -- (11.25,1.30);
  \draw[labw] (0.55,0.65) -- (11.25,0.65);
  \draw[imgw] (0.55,0.00) -- (11.25,0.00);

  \node[act, minimum width=0.85cm, minimum height=1.9cm] at (0.55,0.65) {Arr};
  \node[act, minimum width=0.85cm, minimum height=1.9cm] at (1.90,0.65) {Exa$_1$};

  \draw[round] (2.75,-0.45) rectangle (5.75,1.75);
  \node[act, minimum width=0.85cm, minimum height=0.48cm] (t1) at (3.55,0.65) {Tro$-$};
  \node[pw] at (t1.north east) {$b_i$};
  \node[act, minimum width=0.85cm, minimum height=0.48cm] (e1) at (3.55,0.00) {ECG$-$};
  \node[pw] at (e1.north east) {$c_i$};
  \node[act, minimum width=0.85cm, minimum height=1.9cm] at (5.05,0.65) {Rex};
  \node[font=\footnotesize, gray!80, anchor=south] at (4.25,1.75) {$R_{b_i,c_i}$};
  \node[font=\footnotesize, gray!80, anchor=west] at (5.75,1.55) {$\times\, a$};
  \draw[decorate, decoration={brace, amplitude=4pt, mirror}, gray!80, thick] (2.75,-0.55) -- (5.75,-0.55)
    node[midway, below=4pt, font=\footnotesize] {$R_u$};

  \node[act, minimum width=0.85cm, minimum height=0.48cm] (t2) at (6.65,0.65) {Tro$-$};
  \node[pw] at (t2.north east) {$b'$};
  \node[act, minimum width=0.85cm, minimum height=0.48cm] (e2) at (6.65,0.00) {ECG$-$};
  \node[pw] at (e2.north east) {$c'$};
  \node[act, minimum width=0.85cm, minimum height=0.48cm] at (8.05,0.00) {ECG$+$};

  \node[act, minimum width=0.85cm, minimum height=1.9cm] at (9.45,0.65) {Adm$_{\mathrm{e}}$};
  \node[act, minimum width=0.85cm, minimum height=1.9cm] at (10.85,0.65) {Dep};

  \node[font=\scriptsize, anchor=east] at (0.05,1.30) {patient};
  \node[font=\scriptsize, anchor=east] at (0.05,0.65) {lab};
  \node[font=\scriptsize, anchor=east] at (0.05,0.00) {imaging};
\end{tikzpicture}

%% file: sections/09-conversion_framework.tex
\section{Cross-notation equivalence and conversion}\label{sec:conversion-framework}

The forward maps of Sections~\ref{sec:petri-nets}--\ref{sec:bpmn} are model-to-model
transformations that carry each notation into a common framework. Equality of the
resulting signatures is then a model-equivalence check that can validate such a
transformation, confirming that two models agree as processes.

Sections~\ref{sec:petri-nets}--\ref{sec:bpmn} establish one direction, each notation mapping
canonically to a signature $\Sigma$, and hence to $\mathbf{F}(\Sigma)$, by one engine at different instance maps
(Remark~\ref{rem:ptree-recovers-pn}). Two questions follow, when two models, possibly in
different notations, produce equal signatures and what that equality certifies, and whether a
model can be recovered from a signature. We settle the first in
Section~\ref{subsec:equiv-certificate}. The second, and any full round-trip faithfulness, is
developed only as far as stating what each reconstruction requires
(Sections~\ref{subsec:reverse-map}--\ref{subsec:open-problems}) and is left as future work.

\begin{remark}[Notation-independence of the construction]
  \label{rem:ptree-recovers-pn}
  The notation constructions of Sections~\ref{sec:petri-nets}--\ref{sec:bpmn} are one engine at
  different instance maps. Each instance map reads typed wires $(x,w,t)$, treats silent leaves or
  transitions as transparent mediators with no elimination step, and composes along typed wires. A
  given object-centric behaviour therefore receives the same minimal signature whether presented as
  an OCPN, an OC causal net, or an OC process tree. The presentations differ only in how contexts
  are derived, from Petri place flow, from causal binding sets, or from operator contexts filtered
  by $D_{PT}$. A process tree also reaches its signature indirectly, by van Detten's OCPT-to-OCPN
  translation~\cite{vandettenDiscoveringCompactLive2024} composed with the OCPN construction.
  Provided that translation preserves behaviour, the indirect route and the direct construction of
  Section~\ref{sec:process_trees} yield the same signature, a commuting triangle that reads the
  notation-independence as a consistency check.
\end{remark}

\subsection{Signature equality as an equivalence certificate}\label{subsec:equiv-certificate}

Fix a process-mining notation $\mathcal{N}$ with the forward compilation $m\mapsto\Sigma(m)$ of
Definition~\ref{def:generator-cospan-general}, which consults no modelling convention. Write $C(m)$ for
the \emph{context structure} of $m$. The context structure is the family, indexed by activities, of the contexts
$c=(P,S)\in\mathrm{Contexts}(a)$ that the instance map reads off $m$'s compiled AND/OR graph
(Definition~\ref{def:contexts-general}), each carrying its constraint system $\Lambda_{a,c}$
(Definition~\ref{def:multiplicity-decoration}), compared by its solutions. All
four instance maps emit contexts in this one format, over the shared activity set and object types,
so $C(m)$ is defined in a notation-independent space (Remark~\ref{rem:ptree-recovers-pn}).

Fix a selection $\mathcal{S}$ in the sense of Section~\ref{subsec:theorem}, by default the
connected process diagrams, and call $L_{\mathcal{S}}(\Sigma)$, the union of the trace sets
$\Gamma(f)$ of Definition~\ref{def:trace-map} over $f\in\mathcal{S}(\Sigma)$, the trace language
of the model.

\begin{definition}[Structural and trace equivalence]\label{def:str-tr-equiv}
  Two models $m_1,m_2$, possibly in different notations, are \emph{structurally equivalent},
  $m_1\equiv_{\mathrm{str}}m_2$, if they induce the same context structure, $C(m_1)=C(m_2)$. They
  are \emph{trace-equivalent}, $m_1\equiv_{\mathrm{tr}}m_2$, if their trace languages coincide,
  $L_{\mathcal{S}}(\Sigma(m_1))=L_{\mathcal{S}}(\Sigma(m_2))$.
\end{definition}

\begin{lemma}[The signature is a complete invariant of the context structure]\label{lem:sig-complete}
  For models in any of the notations of Sections~\ref{sec:petri-nets}--\ref{sec:bpmn},
  \[
    \Sigma(m_1)=\Sigma(m_2) \quad\Longleftrightarrow\quad m_1\equiv_{\mathrm{str}}m_2,
  \]
  up to canonical generator relabelling. Structural equivalence is therefore decidable by direct
  comparison of the finite generator sets, across notations.
\end{lemma}

\begin{proof}
  The generate step (Definition~\ref{def:generator-cospan-general}) sends each context
  $c=(P,S)\in\mathrm{Contexts}(a)$ to the generator $g_{a,c}$ with label $\ell(a)$, input legs the
  typed wires $(x,w,a)\in P$, output legs the typed wires $(a,w,y)\in S$, and constraint system
  $\Lambda_{a,c}$. Because every typed wire names both endpoints and its set of types
  (Definition~\ref{def:contexts-general}), distinct contexts give distinct generators, so this map is injective. Reading each generator's boundary
  recovers its context, so it is a bijection between context structures and canonically-labelled
  signatures. Equal signatures therefore correspond exactly to equal context structures.
\end{proof}

\begin{lemma}[Soundness]\label{lem:equiv-soundness}
  $m_1 \equiv_{\mathrm{str}} m_2$ implies $m_1 \equiv_{\mathrm{tr}} m_2$.
\end{lemma}

\begin{proof}
  By Lemma~\ref{lem:sig-complete}, $\equiv_{\mathrm{str}}$ gives $\Sigma(m_1)=\Sigma(m_2)$, and Theorem~\ref{thm:canonical-presentation} gives
  $L_{\mathcal{S}}(\Sigma(m_1))=L_{\mathcal{S}}(\Sigma(m_2))$.
\end{proof}

\begin{lemma}[Strictness]\label{lem:equiv-strictness}
  $\equiv_{\mathrm{str}}$ is strictly finer than $\equiv_{\mathrm{tr}}$. Trace equivalence does not
  imply structural equivalence.
\end{lemma}

\begin{proof}
  Take one object type and weight $1$ throughout. Let $m_1$ be the sequence $a$ then $b$, and let
  $m_2$ offer three exclusive alternatives, $a$ then $b$, $a$ alone and $b$ alone, its occurrences
  of each activity being distinct nodes with the same label. Writing $(x,y)$ for the typed wire
  $(x,\{\ast\},y)$, $\Sigma(m_1)$ has the generators $a:(\bot,a)\to(a,b)$ and $b:(a,b)\to(b,\top)$.
  $\Sigma(m_2)$ has these two and also $a:(\bot,a)\to(a,\top)$ and $b:(\bot,b)\to(b,\top)$. The
  connected process diagrams of both signatures give the language $\{a,b,ab\}$, so
  $m_1\equiv_{\mathrm{tr}}m_2$, yet $C(m_1)\ne C(m_2)$ and $m_1\not\equiv_{\mathrm{str}}m_2$. The
  context structure records which occurrences can be joined to which, and the language loses
  that information.
\end{proof}

\begin{theorem}[Cross-notation equivalence certificate]\label{thm:signature-equivalence}
  Let $\mathcal{N}$ be any of the notations of Sections~\ref{sec:petri-nets}--\ref{sec:bpmn}.
  Then three statements hold.
  (i)~$\Sigma(m)$ is determined by $m$ alone, silent $\tau$-mediators contributing no generator and
  multiplicities carried in the constraint systems $\Lambda_{a,c}$, so no elimination step or routing convention is
  consulted. (ii)~$\Sigma(m_1)=\Sigma(m_2)$ decides $m_1\equiv_{\mathrm{str}}m_2$, cross-notation, by
  finite generator comparison (Lemma~\ref{lem:sig-complete}). The comparison is equality
  of $\Sigma\setminus\Sigma_\gamma$. The start and end generators of $\Sigma_\gamma$ state an
  assumption of the analysis (Section~\ref{subsec:analysis}) and take no part, so two models written
  in notations that open and close a system differently are compared without either being
  rewritten. (iii)~$\equiv_{\mathrm{str}}$ certifies $\equiv_{\mathrm{tr}}$ for every selection
  $\mathcal{S}$, by Theorem~\ref{thm:canonical-presentation}, and the converse fails
  (Lemmas~\ref{lem:equiv-soundness} and~\ref{lem:equiv-strictness}). In particular, since $\Sigma$ is convention-free,
  $\Sigma(m_1)\ne\Sigma(m_2)$ witnesses a structural difference that no encoding choice
  can produce, even when $m_1\equiv_{\mathrm{tr}}m_2$.
\end{theorem}

\begin{proof}
  Point (i) is the instance-map classification of silent transitions as mediators, which the walk
  passes through, with multiplicities kept as the constraint systems $\Lambda_{a,c}$
  (Section~\ref{subsec:firing-graphs}, Definitions~\ref{def:contexts-general}
  and~\ref{def:multiplicity-decoration}). Point (ii) is
  Lemma~\ref{lem:sig-complete}, and point (iii) is Lemmas~\ref{lem:equiv-soundness}
  and~\ref{lem:equiv-strictness}, whose contrapositive gives the negative certificate.
\end{proof}

For a practitioner the certificate answers whether two discovered models have the same structure,
and equal signatures then certify that they have the same trace language. It is strictly finer
than trace equivalence, as the sequence against choice case of
Lemma~\ref{lem:equiv-strictness} shows, so a pair the trace languages identify may still be
separated on structure. Inclusion gives the ordered form of the same check. When
$\Sigma(m_1)\subseteq\Sigma(m_2)$, every trace $m_1$ gives under a monotone selection is a trace of
$m_2$ (Corollary~\ref{cor:signature-inclusion}), and the generators of $\Sigma(m_2)$ outside
$\Sigma(m_1)$ name where $m_2$ allows more. Section~\ref{sec:validation} finds exactly this for two
discovered models, whose signatures nest with the causal net's inside the Petri net's.

\paragraph{Granularity.}
Structural equivalence is equality of context structures, which is coarser than isomorphism of the
AND/OR graphs. An OR gateway and its XOR/AND expansion are identified, as are two silent-mediator arrangements
with the same activity neighbourhoods. $\Sigma$ records the behavioural space, not the
syntax. Read as a model-equivalence relation, this coarseness compares discovered models over
their behavioural space rather than their surface syntax. Matching
$\equiv_{\mathrm{str}}$ to each notation's \emph{native}
equivalence is the reverse-map question of Section~\ref{subsec:reverse-map}, which is partly open.

\begin{corollary}[Internal routing is quotiented]\label{cor:internal-routing-quotient}
Let $m_1$ and $m_2$ differ only in internal places or in silent $\tau$ routing, so that they present the same observable contexts $C(m_1)=C(m_2)$. Then $\Sigma(m_1)=\Sigma(m_2)$.
\end{corollary}

\begin{proof}
Silent transitions contribute no generator, and internal places are XOR mediators that the walk of Definition~\ref{def:contexts-general} passes through. By Theorem~\ref{thm:signature-equivalence}(i) they leave the context structure $C(m)$ unchanged, so Lemma~\ref{lem:sig-complete} gives $\Sigma(m_1)=\Sigma(m_2)$.
\end{proof}

\subsection{Reverse maps and faithfulness}\label{subsec:reverse-map}

The converse direction, recovering a model from a signature, we develop only as far as stating what
each reconstruction requires. The target is a single schema parameterised by the notation
$\mathcal{N}$ and its native behavioural equivalence $\sim_{\mathcal{N}}$. For a forward map
$\Phi_{\mathcal{N}}$ and a proposed reverse map $\Psi_{\mathcal{N}}$, we want $\Phi_{\mathcal{N}}
\circ \Psi_{\mathcal{N}}$ the identity on signatures and $\Psi_{\mathcal{N}} \circ
\Phi_{\mathcal{N}}$ the identity on models up to $\sim_{\mathcal{N}}$. Where each instance is
canonical, and where it weakens to a normal form, is summarised in Table~\ref{tab:faithfulness}.

\begin{table*}[t]
  \centering
  \caption{Reverse maps and the canonicity of each round trip, by notation. SP abbreviates series--parallel throughout.}
  \label{tab:faithfulness}
  \small\begin{tabularx}{\linewidth}{@{}llX@{}}
    \hline
    Notation & Reverse map & Canonicity \\
    \hline
    Causal nets   & read boundaries back as bindings & identity up to isomorphism \\
    Petri nets    & apex hyperedge $+$ glued ports    & up to silent ($\tau$) saturation \\
    Process trees & SP decomposition tree ($\to,+$)   & canonical on $\to,+$; open with XOR/loop \\
    BPMN          & XOR/AND normal form               & normal form, not original gateways \\
    \hline
  \end{tabularx}
\end{table*}

For causal nets the boundary data of the generator cospans are exactly the binding sets
(Section~\ref{sec:causal-nets}), so the round trip is the identity up to isomorphism. For Petri nets
each generator reads back as a transition (apex label, glued ports as its places), the identity up
to $\tau$-saturation. Silent transitions produce no generators, so the reverse map recovers the
$\tau$-free representative. For process trees the $\to,+$ fragment is canonical and constructive
(Corollary~\ref{cor:pt-sp-iff}), while recovering XOR and loop structure from a set of SP posets is
open (Section~\ref{subsec:open-problems}). For BPMN the OR explosion is many-to-one, so a signature
fixes a canonical XOR/AND normal form, not the OR-gateway syntax that may have produced it,
consistent with the denotational reading of Section~\ref{sec:bpmn}.
These four reverse maps are the model-to-model transformations that would slot into
existing modelling toolchains, generalising the BPMN, Petri net, causal net, and process tree
conversions of \citet{kalenkovaProcessMiningUsing2017}.
Synthesis of Petri nets from partial-order languages has been studied in general through
token-flow regions~\cite{bergenthumSynthesisPetriNets2008}. The nets produced by synthesis are
typically hard for an analyst to read~\cite{leemansPartialorderbasedProcessMining2023}, which
motivates recovering the structured notations above instead.

\subsection{Open problems}\label{subsec:open-problems}

\begin{enumerate}
  \item \emph{Process-tree definability.} Characterise which SP-poset languages are of the form
    $\mathrm{sem}_{\mathrm{pt}}(T)$ for some process tree $T$. Process-tree languages are a strict
    subclass of SP-poset languages. The relevant setting is that of series-parallel
    languages~\cite{lodayaSeriesParallelLanguages2000}.
  \item \emph{Canonical process trees.} Even for definable languages, process trees admit no
    canonical form without a normal-form theorem. Define a congruence (building on the equational
    theory of pomsets~\cite{gischerEquationalTheoryPomsets1988}) and prove a minimisation result, in
    particular for the loop operator.
  \item \emph{BPMN faithfulness.} Make precise the sense in which the recovered XOR/AND normal form
    is behaviourally equivalent to the original OR-gateway model.
  \item \emph{POWL as a fifth instance.} The partially ordered workflow language (POWL) composes submodels along arbitrary partial
    orders together with choice and loop operators~\cite{kouraniPOWLPartiallyOrdered2023}. Its
    poset-plus-operators syntax suggests an AND/OR graph in the sense of
    Definition~\ref{def:lm-graph} and hence a signature
    (Definition~\ref{def:generator-cospan-general}). Carrying this out would add a
    partial-order-native notation to the four treated here and relate the certificate of
    Theorem~\ref{thm:signature-equivalence} to POWL's language-preserving workflow-net
    translation~\cite{kouraniTranslatingWorkflowNets2025}.
\end{enumerate}

%% file: sections/08-related_work.tex
\section{Related work}\label{sec:related-work}

We situate the contribution within process modelling broadly rather than against categorical semantics alone~\cite{chechikFormalMethodsScope2025}. The threads on notation diversity, BPMN semantics, and object-centric mining carry the positioning, and the categorical-semantics thread is one supporting strand.

\subsection{Notation diversity and representational bias}

Process mining uses many representational formalisms, Petri nets, process trees, causal nets,
directly-follows graphs, BPMN, Declare, and DCR graphs~\cite{deboisDeclarativeProcessMining2017},
each coupled to particular discovery algorithms. Three lines of work relate them. The first is
conversion-based translation with behavioural
preservation~\cite{kalenkovaProcessMiningUsing2017}. The second is van der Aalst's
representational-bias programme, which shows the target formalism constrains the discovery search
space and motivates sound-by-construction
classes~\cite{vanderaalstRepresentationalBiasProcess2011,%
  vanderaalstImprovingRepresentationalBias2012,%
  vanderaalstUsingFreeChoiceNets2021,%
vanderaalstFoundationsProcessDiscovery2022}. The third is the categorical-semantics tradition, in which Petri
nets present free symmetric monoidal categories~\cite{meseguerPetriNetsAre1990}, extended via open
Petri nets~\cite{baezOpenPetriNets2020} and string-diagrammatic
PROPs~\cite{bonchiDiagrammaticAlgebraLinear2019}. A common alternative is to route every notation through a single pivot such as workflow
nets~\cite{vanderaalstVerificationWorkflowNets1997}. That relocates the difficulty rather than
settling it, because the translations into the pivot are themselves uneven in rigour, they remain
incomplete on constructs such as the BPMN OR-join~\cite{dijkmanSemanticsAnalysisBusiness2008}, and
the comparison made afterwards is still a comparison of traces.
The present work adds a common compositional setting for the four formalisms it treats,
generalising conversion-based translation between BPMN, Petri nets, causal nets, and process
trees~\cite{kalenkovaProcessMiningUsing2017} to a single setting. These formalisms are
produced and consumed by process-mining tools such as ProM and PM4Py, so a shared framework also
serves the interchange between tools that emit different notations.

\subsection{Categorical and compositional process semantics}

Treating Petri nets as generators of symmetric monoidal categories originates with
\citet{meseguerPetriNetsAre1990}. Reading interface composition as an algebra of
connectors is the Montanari connector-algebra tradition, from the stateless connectors of
\citet{bruniBasicAlgebraStateless2006} through the concatenable-process
axiomatisation of \citet{sassoneAxiomatizationAlgebraPetri1996} and the classification of
models of concurrency of \citet{sassoneModelsConcurrencyClassification1996}.
Cospan composition here is the process-mining instance of that connector-composition tradition, with
the connectors read off event-log structure rather than fixed in advance. \citet{baezOpenPetriNets2020} extended
this to open Petri nets via decorated cospans. The present construction specialises that framework,
decorating cospan apices with process-mining data (activity labels, binding constraints) and
restricting to the set of runs $\mathbf{F}(\Sigma)$, the closed connected diagrams, to exclude non-process-like
wirings. \citet{bonchiDiagrammaticAlgebraLinear2019} unify nets and process calculi through
string-diagrammatic PROPs. Our aim is complementary, recovering partial-order structure from
event-log data rather than reasoning equationally about concurrency laws.
Similarly, \citet{lechenneCompositionalFrameworkPetri2024}
equip Petri nets with open-ended interfaces and a PROP-based graphical language in order to compute
reachability compositionally. Like the open-net line, the framework treats nets alone, with no
event-log or cross-notation discussion.

Outside the categorical approach, the Heraklit programme of \citet{fettkeSystemsMiningHeraklit2022} composes
Petri-net modules along left and right interfaces by fusing equally labelled boundary elements, an
associative composition proposed for process mining, and
they have recently shown that the partially ordered runs of a composed system are exactly the
composites of the runs of its parts~\cite{fettkeCompositionalitySystemsPartially2026}. This is the
compositionality that a hypergraph category provides by construction. Interface fusion is cospan
composition, and the run-level statement then holds uniformly across the four notations treated
here rather than being proved per formalism.

\citet{haymanUnfoldingGeneralPetri2008} showed that the cofree unfolding of a
general Petri net is recoverable only up to a symmetry on the folding morphism's kernel. The cospan
construction encounters the same multiplicity phenomenon, non-uniqueness of mediating morphisms from
indistinguishable tokens, but computes its quotient directly from generator boundaries, a form
estimable from event-log traces without solving a morphism-lifting problem.

\subsection{Partial orders, unfoldings, and untangling}

Partial-order-based process mining, spanning partially ordered event data, discovery,
and conformance, is surveyed by \citet{leemansPartialorderbasedProcessMining2023}.
\citet{Unfoldings2008} build finite complete prefixes of occurrence-net
unfoldings for model checking. Like the present work they take partial orders as primary, but
address the forward problem (verifying a given net) and need boundedness for finiteness, whereas the
cospan construction works directly from generators, essential where object-centric place
multiplicities are unbounded.

\citet{polyvyanyyUntanglingsNovelApproach2015} extract untanglings, representative
replay-compatible executions of a marked net. In the cospan setting a single string diagram
represents the whole concurrency class of executions with the same causal structure, and building
all diagrams from a signature enumerates all untanglings, most cleanly for 1-safe nets.
The three approaches differ in how much they enumerate. Unfoldings are fully enumerative
(a node per configuration), untanglings keep a pruned representative family, and the cospan
signature is structural, encoding causal roles rather than runs, with size
$\propto |A|\times|\mathrm{Contexts}|$ independent of run count and specific partial orders produced
on demand by composition.
The same interleaving blow-up motivates conformance checking over partially ordered
traces. Partially ordered alignments were introduced for this purpose by
\citet{luConformanceCheckingBased2015}. \citet{leemansPartiallyOrderedStochastic2025} encode observed
behaviour as labelled partial orders precisely to avoid enumerating the interleavings of concurrent
activities. The diagrams here are a model-side counterpart, each representing its concurrency class
directly, so that comparison can stay at the partial-order level.

The same frequency-aware agenda extends to stochastic process discovery. \citet{leemansStochasticProcessDiscovery2024} study when this discovery can be done optimally, and
\citet{liDiscoveringStochasticCausal2025} discover stochastic causal nets, a stochastic
counterpart to the object-centric causal nets treated in Section~\ref{sec:causal-nets}. The
construction here is qualitative, and decorating its generators with execution weights is a natural
extension.

Model-to-model comparison over event structures has also been studied. \citet{armas-cervantesBehavioralComparisonProcess2014} compare process models through
canonically reduced event structures and report their behavioural differences. The certificate of
Section~\ref{sec:conversion-framework} plays the same role across notations, and its canonical form
is fixed by the construction rather than by a reduction procedure.

\subsection{Process trees and series-parallel structure}\label{subsec:pt-sp-related}

Process trees are a central block-structured notation, used both as a model class and as a discovery
intermediate~\cite{leemansDiscoveringBlockStructuredProcess2013,buijsQualityDimensionsProcess2014}.
Their sequence and parallel operators are the classical series and parallel compositions of partial
orders. Finite series-parallel (SP) posets, the $N$-free finite
posets~\cite{valdesRecognitionSeriesParallel1982,mohringComputationallyTractableClasses1989}, have a
finitely axiomatisable pomset theory~\cite{gischerEquationalTheoryPomsets1988} linking them to the
concurrency semantics of \citet{prattModelingConcurrencyPartial1986} and
\citet{lodayaSeriesParallelLanguages2000}. The process-mining literature invokes this
correspondence implicitly through block-structured syntax.
Appendix~\ref{sec:pt-extras} states it as an explicit occurrence-level representation theorem
(Corollary~\ref{cor:pt-sp-iff}), with Theorem~\ref{thm:pt-as-sp-language} separating the $\to,+$
fragment from the full language, where XOR contributes set-union and loop bounded unfoldings.
The series-parallel restriction has also been relaxed at the notation level. POWL
models~\cite{kouraniPOWLPartiallyOrdered2023} compose submodels along arbitrary partial orders
rather than series and parallel operators alone, remain sound by construction as a subclass of
workflow nets, and have been extended to non-block-structured
choice~\cite{kouraniUnlockingNonBlockStructuredDecisions2025}. A language-preserving translation from safe
and sound workflow nets into POWL exists~\cite{kouraniTranslatingWorkflowNets2025}. POWL thus
removes at the syntax level the $N$-freeness that Corollary~\ref{cor:pt-sp-iff} characterises at
the semantic level, and its translation guarantee is stated on trace languages, whereas the
certificate of Section~\ref{sec:conversion-framework} compares structure, which by
Lemma~\ref{lem:equiv-strictness} is strictly finer.

\subsection{BPMN and formal process semantics}\label{subsec:bpmn-related}

Formal BPMN semantics are mainly operational. \citet{dijkmanSemanticsAnalysisBusiness2008} normalise a model and encode it into a Petri net
whose reachable markings are its valid token states. Section~\ref{sec:bpmn} uses no soundness notion. The soundness taxonomy of the workflow-net
literature~\cite{vanderaalstSoundnessWorkflowNets2011} has decision problems that are
\textsc{ExpSpace}- and \textsc{PSpace}-complete~\cite{blondinComplexitySoundnessWorkflow2022}. The
cospan construction is denotational instead. It extracts the interface structure with no firing
strategy or structural precondition, handling OR gateways by finite XOR/AND explosion, so it
avoids those decision problems, which then bound only the composition it defers.

Object-centric BPMN has been approached differently before. \citet{seidelObjectcentricBPMNProcess2024} give a fragment-based dialect with data objects and object
lifecycles and informal semantics, within the object-centric process-mining
programme~\cite{vanderaalstObjectCentricProcessMining2023,bertiAdvancementsChallengesObjectCentric2023}. Adjacent formal
extensions target the artifact-centric~\cite{lohmannArtifactCentricModelingUsing2012} and
data-aware~\cite{calvaneseFormalModelingSMTBased2019} paradigms. None types BPMN's control flow with object
types, as Section~\ref{sec:bpmn} does.

\subsection{Object-centric process mining}

Object-centric process mining sharpens the inter-notation problem.
OCEL~2.0~\cite{bertiOCELObjectCentricEvent2024} dropped the single-case constraint and object-centric Petri
nets~\cite{vanderaalstDiscoveringObjectCentricPetri2020} added typed tokens, with more
expensive reachability and less mature conformance. The standard does not natively carry
hierarchical structure, temporal inter-object relations, or explicit
state~\cite{khayatbashiAdvancingObjectcentricProcess2026,bertiStateAwareObjectCentricProcess2025}.
\citet{lissObjectCentricCausalNets2025} extend causal nets to typed binding sets, which
the cospan construction accommodates directly (Section~\ref{sec:causal-nets}). In the cospan
framework object types are wire labels. Boundary-matching composability enforces type-compatibility
structurally, the categorical analogue of a typed slice, with no reachability bookkeeping, and the
same mechanism covers the object-centric causal nets and process trees of Liss et al.\ and
\citet{vandettenDiscoveringCompactLive2024} unchanged. These object-centric
concerns are active in the modelling community as well as in mining. Data-aware enterprise process
modelling in MERODE~\cite{snoeckSupportingDataawareProcesses2023} and the object-centric
event-logging and data-centric modelling terminology of
iDOCEM~\cite{verbruggenIDOCEMDefiningCommon2024} formalise objects and their life-cycles at the
conceptual level, and model-driven management of BPMN-based process
families~\cite{delgadoModeldrivenManagementBPMNbased2022} organises related models within one
notation. The framework developed here is complementary and extends these by comparing and
translating models across the four notations at once, with object types carried structurally on the
wires rather than reconstructed from a conceptual schema.

%% file: sections/09b-conclusion.tex
\section{Conclusion}\label{sec:conclusion}

\paragraph{Summary.}
We have given four standard process-mining notations, object-centric Petri nets, causal nets,
process trees, and BPMN, a single canonical presentation as a cospan-algebra signature in a
hypergraph category (Theorems~\ref{thm:petri-to-cospan-signature},
\ref{thm:causal-to-cospan-signature}, \ref{thm:ptree-to-cospan-signature},
and~\ref{thm:bpmn-cospan}), as four instances of one construction
(Section~\ref{sec:general-framework}). The resulting minimal signature is notation-independent. The
same object-centric behaviour receives the same signature whichever notation presents it
(Remark~\ref{rem:ptree-recovers-pn}).

\paragraph{Consequences of the construction.}
String diagrams thereby give a common ground for comparing and translating models across notations.
Equality of signatures is a structural certificate, decidable by direct comparison of finite
generator sets, that implies trace equivalence and is strictly finer than it
(Theorem~\ref{thm:signature-equivalence}). The certificate distinguishes genuine concurrency $a\otimes b$ from
the exclusive choice between $a\,;\,b$ and $b\,;\,a$, which trace languages conflate, and unequal
signatures witness genuine structural difference regardless of trace language. Because the
translations preserve sequential and parallel composition, an equality on a fragment persists under
composition into larger models. Object-centric typing rides on the wires, enforced by composition
rather than separate bookkeeping, with the classical untyped notations the one-type specialisation.

\paragraph{Limitations.}
The construction is one-directional. The forward maps are canonical, but the reverse maps are
developed only as far as stating what each reconstruction requires
(Section~\ref{subsec:reverse-map}), and a full round-trip faithfulness result is left open. The BPMN reverse map recovers an XOR/AND normal form rather than
the original OR-gateway syntax.

\paragraph{Future work.}
Section~\ref{subsec:open-problems} collects four open problems. It asks which
series-parallel-poset languages a process tree can define, whether process trees admit a canonical
form covering the loop operator, in what precise sense the BPMN normal form is faithful to its
OR-gateway original, and whether POWL fits as a further instance of the general framework. A single faithfulness theorem parameterised by the notation and
its native equivalence would subsume the per-notation reverse maps. A further direction is
complexity. The cospan view grounds a discovery problem against the space of partial orders over an
activity set, with $\mathbf{F}(\Sigma)$ the subspace generated by $\Sigma$, so $|\Sigma|$ and diagram size
should bound the search rather than the raw count of relations.

\paragraph{Outlook.}
The wider aim is a compositional basis for comparing the process models organisations actually run.
This work supplies its structural layer, the partial-order object-centric skeleton on which any two
models sit in one category. The natural next layer is stochastic.
Stochastic Petri nets~\cite{baezQuantumTechniquesStochastic2018} attach rates or routing
probabilities to transitions. Those rates and probabilities fit as a further decoration on the same generators, a weight
riding alongside each generator's constraint system $\Lambda_{a,c}$ and composed by the same pushout, so stochastic object-centric causal
nets and process trees follow unchanged. A stochastic signature would compare models on
what behaviour they admit and on how often, bringing the certificate of
Section~\ref{sec:conversion-framework} closer to the statistical comparison that conformance and
discovery require.

Finally, the framework sits between two communities that have largely developed apart. Process
mining gains a principled account of local composition, assembling global models from reusable
generators. Applied category theory gains a data-driven setting in which its constructions are
estimated from event logs rather than postulated. We hope the cospan-algebra presentation makes that
exchange concrete.

%% file: sections/10-hyper_axioms.tex
\section{Hypergraph Category Axioms}\label{sec:axioms}

In a hypergraph category every object $X$ carries a special commutative
Frobenius algebra (SCFA, also called a Frobenius monoid) consisting of four morphisms
\begin{gather*}
  \mu_X \colon X \otimes X \to X,\quad
  \eta_X \colon I \to X,\\
  \delta_X \colon X \to X \otimes X,\quad
  \epsilon_X \colon X \to I,
\end{gather*}
subject to the axioms below, together with coherence conditions that fix the
SCFA on every tensor product $A \otimes B$ in terms of those on $A$ and $B$.
Throughout, $\sigma_{X,Y}\colon X\otimes Y\to Y\otimes X$ is the symmetry of the underlying
symmetric monoidal category, the wire crossing, and $;$ denotes composition in diagrammatic order.

\paragraph{Monoid.}
\begin{align}
  (\mu_X \otimes \mathrm{id}_X) ; \mu_X
    &= (\mathrm{id}_X \otimes \mu_X) ; \mu_X
    \label{eq:ax-assoc}
    \tag{assoc}\\
  (\eta_X \otimes \mathrm{id}_X) ; \mu_X
    &= \mathrm{id}_X
    = (\mathrm{id}_X \otimes \eta_X) ; \mu_X
    \label{eq:ax-unit}
    \tag{unit}
\end{align}

\paragraph{Comonoid.}
\begin{align}
  \delta_X ; (\delta_X \otimes \mathrm{id}_X)
    &= \delta_X ; (\mathrm{id}_X \otimes \delta_X)
    \label{eq:ax-coassoc}
    \tag{coassoc}\\
  \delta_X ; (\epsilon_X \otimes \mathrm{id}_X)
    &= \mathrm{id}_X
    = \delta_X ; (\mathrm{id}_X \otimes \epsilon_X)
    \label{eq:ax-counit}
    \tag{counit}
\end{align}

\paragraph{Frobenius.}
\begin{align}
  (\mathrm{id}_X \otimes \delta_X) ; (\mu_X \otimes \mathrm{id}_X)
    &= \mu_X ; \delta_X\notag\\
    &= (\delta_X \otimes \mathrm{id}_X) ; (\mathrm{id}_X \otimes \mu_X)
    \label{eq:ax-frob}
    \tag{Frob}
\end{align}

\paragraph{Commutativity and cocommutativity.}
\begin{align}
  \sigma_{X,X} ; \mu_X &= \mu_X
    \label{eq:ax-comm}
    \tag{comm}\\
  \delta_X ; \sigma_{X,X} &= \delta_X
    \label{eq:ax-cocomm}
    \tag{cocomm}
\end{align}

\paragraph{Speciality.}
\begin{align}
  \delta_X ; \mu_X &= \mathrm{id}_X
    \label{eq:ax-special}
    \tag{special}
\end{align}

\paragraph{Monoidal coherence.}
The SCFA on $A \otimes B$ is determined by those on $A$ and $B$:
\begin{align}
  \mu_{A \otimes B}
    &= (\mathrm{id}_A \otimes \sigma_{B,A} \otimes \mathrm{id}_B)
       ; (\mu_A \otimes \mu_B)
    \label{eq:ax-coh-mu}
    \tag{coh-$\mu$}\\
  \eta_{A \otimes B}
    &= \eta_A \otimes \eta_B
    \label{eq:ax-coh-eta}
    \tag{coh-$\eta$}\\
  \delta_{A \otimes B}
    &= (\delta_A \otimes \delta_B)
       ; (\mathrm{id}_A \otimes \sigma_{A,B} \otimes \mathrm{id}_B)
    \label{eq:ax-coh-delta}
    \tag{coh-$\delta$}\\
  \epsilon_{A \otimes B}
    &= \epsilon_A \otimes \epsilon_B
    \label{eq:ax-coh-eps}
    \tag{coh-$\epsilon$}
\end{align}
The monoidal unit $I$ carries the trivial SCFA, with $\mu_I = \eta_I = \delta_I = \epsilon_I = \mathrm{id}_I$.

%% file: sections/11-process-tree-extras.tex
\section{Process tree supplements}\label{sec:pt-extras}

This appendix characterises the sequence and parallel fragment of the process-tree construction of
Section~\ref{sec:process_trees} as series-parallel partial orders.

\subsection{Process trees and series-parallel posets}\label{subsec:pt-sp-appendix}

The sequence and parallel operators of Definition~\ref{def:process-tree} have an exact
characterisation in terms of series-parallel (SP) partial orders.
The definition and results below make this precise. They are also used in the
related-work discussion of Section~\ref{subsec:pt-sp-related}.

\begin{definition}[Occurrence-level semantics for process trees]\label{def:pt-occ-semantics}
  Let $\mathcal{A}$ be the set of activity names, and let $\mathcal{SP}$ denote the class of finite series-parallel (SP) posets.
  A behaviour is a finite labelled occurrence poset $(E,\leq,\lambda)$, where $E$ is a finite set of event occurrences, $\leq$ is a partial order, and $\lambda:E\to\mathcal{A}$ is the activity-label map.

  For a process tree $T$ over operators $\to$ (sequence), $+$ (AND), $\times$ (XOR), and $\circlearrowleft$ (loop), define its denotation
  \[
    \mathrm{sem}_{\mathrm{pt}}(T) \subseteq \mathcal{SP}
  \]
  inductively as follows:
  \begin{align*}
    \mathrm{sem}_{\mathrm{pt}}(a)
    &:= \{(\{e\},\{(e,e)\},\{e\mapsto a\})\},\\
    \mathrm{sem}_{\mathrm{pt}}(\to(T_1,T_2))
    &:= \{P_1;P_2 \mid\twocolalignbreak P_1\in\mathrm{sem}_{\mathrm{pt}}(T_1),\twocolalignbreak P_2\in\mathrm{sem}_{\mathrm{pt}}(T_2)\},\\
    \mathrm{sem}_{\mathrm{pt}}(+(T_1,T_2))
    &:= \{P_1\otimes P_2 \mid\twocolalignbreak P_1\in\mathrm{sem}_{\mathrm{pt}}(T_1),\twocolalignbreak P_2\in\mathrm{sem}_{\mathrm{pt}}(T_2)\},\\
    \mathrm{sem}_{\mathrm{pt}}(\times(T_1,\dots,T_k))
    &:= \bigcup_{i=1}^k \mathrm{sem}_{\mathrm{pt}}(T_i),
  \end{align*}
  \begin{gather*}
    \mathrm{sem}_{\mathrm{pt}}(\circlearrowleft(S,B_1,\dots,B_m)) :=\twocolbreak 
    \bigcup_{n\ge 0}\ \bigcup_{j_1,\dots,j_n\in\{1,\dots,m\}}\twocolbreak 
    \mathrm{sem}_{\mathrm{pt}}\bigl(S;B_{j_1};S;\cdots;B_{j_n};S\bigr).
  \end{gather*}
  Here $(;)$ is series composition of posets and $(\otimes)$ is parallel composition of posets.
  The binary $\to$ and $+$ extend to arbitrary arity by associativity.
\end{definition}

\begin{theorem}[Process trees as sets of SP partial orders]\label{thm:pt-as-sp-language}
  For every process tree $T$, $\mathrm{sem}_{\mathrm{pt}}(T)$ is a (finite or countably infinite) set of finite SP posets.
  \begin{enumerate}
    \item Every $P\in\mathrm{sem}_{\mathrm{pt}}(T)$ is obtained by
      (i) unrolling loops (choosing a finite iteration count and redo-body sequence at each $\circlearrowleft$),
      then (ii) exploding XOR (choosing one branch at each $\times$ node of the unrolled tree, independently for each copy),
      yielding a residual tree over only $\to$ and $+$ whose unique denotation is $P$.
    \item Conversely, every such finite XOR-choice and loop-unrolling yields some $P\in\mathrm{sem}_{\mathrm{pt}}(T)$.
  \end{enumerate}
  Hence the full behaviour of $T$ is exactly characterised as a language of SP posets.
\end{theorem}

\begin{proof}[Proof sketch]
  Structural induction on $T$. SP posets are closed under series and parallel composition, XOR
  contributes a union, and a loop at a fixed iteration count and redo-body sequence is a finite
  $\to,+,\times$ expression, so every denotation is a set of finite SP posets. The two-sided
  characterisation follows by unrolling each loop and choosing one branch at each XOR of the
  unrolled tree, which leaves a $\to,+$ tree denoting a single SP poset, and the converse is
  immediate from the inductive definition as unions over exactly those choices.
\end{proof}

\begin{remark}[Depth-$n$ unrolling]\label{rem:depth-n-unrolling}
For a loop $\circlearrowleft(S,B_1,\dots,B_m)$, fixing an iteration count $n$ and a redo-body sequence gives the loop-free residual expression $S;B_{j_1};S;\cdots;B_{j_n};S$. Call it the depth-$n$ unrolling $D_n$. Exploding any remaining XOR in $D_n$ leaves a tree over $\to$ and $+$ only, whose denotation is a single SP poset by item~(1) of Theorem~\ref{thm:pt-as-sp-language}, a behaviour of the loop at that depth.
\end{remark}

\begin{remark}[Unique occurrence labels and DAG recovery]\label{rem:unique-occurrence-dag}
  To avoid ambiguity from repeated activity names, each leaf occurrence is given a fresh event identifier in $E$.
  The activity name is recovered by the map $\lambda:E\to\mathcal{A}$.
  Given $P=(E,\leq,\lambda)$, its Hasse DAG $H(P)=(E,\prec,\lambda)$ (where $\prec$ is the cover relation) is unique up to isomorphism, and the original order $\leq$ is recovered from it by taking the reflexive-transitive closure of $\prec$, so the DAG representation preserves exactly the same partial-order behaviour.
\end{remark}

\begin{corollary}[Exact fragment without XOR/loop]\label{cor:pt-sp-iff}
  If $T$ uses only $\to$ and $+$, then $\mathrm{sem}_{\mathrm{pt}}(T)=\{P\}$ for some finite SP poset $P$.
  Conversely, every finite SP poset $P$ is denoted by a $\to,+$-only process tree, and this tree is canonical. Taking the decomposition tree of $P$~\cite{valdesRecognitionSeriesParallel1982,mohringComputationallyTractableClasses1989} (the unique alternating series/parallel tree in which no $\to$ node has a $\to$ child and no $+$ node has a $+$ child) yields a process tree that is unique up to associativity of $\to$ and associativity and commutativity of $+$, and is computable in linear time in $|P|$.
  Thus $\to,+$ process trees are exactly finite SP posets (up to isomorphism), and the correspondence is constructive in both directions.
\end{corollary}

The inverse problem for the full operator set, recovering a process tree from a set of SP posets including XOR and loop structure, is more delicate and is taken up as part of the general conversion framework of Section~\ref{sec:conversion-framework}.